\documentclass{article}

\usepackage{arxiv}

\usepackage[utf8]{inputenc} 
\usepackage[T1]{fontenc}    
\usepackage{hyperref}       
\usepackage{url}            
\usepackage{booktabs}       
\usepackage{amsfonts}       
\usepackage{nicefrac}       
\usepackage{microtype}      
\usepackage{lipsum}		
\usepackage{graphicx}
\usepackage{natbib}
\usepackage{doi}
\usepackage[labelfont=bf]{caption}
\usepackage{bm}
\usepackage{graphicx}
\usepackage[space]{grffile}
\usepackage{latexsym}
\usepackage{textcomp}
\usepackage{longtable}
\usepackage{tabulary}
\usepackage{booktabs, array, multirow}
\usepackage{amsfonts,amsmath,amssymb}
\usepackage{natbib}
\usepackage{url}
\usepackage{hyperref}
\hypersetup{colorlinks=false,pdfborder={0 0 0}}
\usepackage{etoolbox}
\usepackage{mathtools}

\usepackage{siunitx}

\usepackage{caption}
\usepackage{subcaption}

\DeclareMathOperator*{\dif}{\mathrm{d} \!}
\DeclareMathOperator*{\vech}{vech}
\DeclareMathOperator*{\diag}{diag}
\DeclareMathOperator*{\tr}{Tr}
\DeclareMathOperator*{\cov}{cov}

\DeclareMathOperator*{\argmin}{arg\,min}

\newtheorem{theorem}{Theorem}[section]
\newtheorem{lemma}[theorem]{Lemma}
\newtheorem{proposition}[theorem]{Proposition}
\newtheorem{corollary}[theorem]{Corollary}
\newtheorem{definition}[theorem]{Definition}
\newtheorem{remark}{Remark}

\allowdisplaybreaks

\newenvironment{proof}[1][Proof]{\begin{trivlist}
\item[\hskip \labelsep {\bfseries #1}]}{\end{trivlist}}

\newcommand{\qed}{\nobreak \ifvmode \relax \else
      \ifdim\lastskip<1.5em \hskip-\lastskip
      \hskip1.5em plus0em minus0.5em \fi \nobreak
      \vrule height0.75em width0.5em depth0.25em\fi}

\makeatletter
\newcommand{\myitem}[1]{%
\item[#1]\protected@edef\@currentlabel{#1}%
} 
\makeatother

\makeatletter
\DeclareRobustCommand\widecheck[1]{{\mathpalette\@widecheck{#1}}}
\def\@widecheck#1#2{%
    \setbox\z@\hbox{\m@th$#1#2$}%
    \setbox\tw@\hbox{\m@th$#1%
       \widehat{%
          \vrule\@width\z@\@height\ht\z@
          \vrule\@height\z@\@width\wd\z@}$}%
    \dp\tw@-\ht\z@
    \@tempdima\ht\z@ \advance\@tempdima2\ht\tw@ \divide\@tempdima\thr@@
    \setbox\tw@\hbox{%
       \raise\@tempdima\hbox{\scalebox{1}[-1]{\lower\@tempdima\box
\tw@}}}%
    {\ooalign{\box\tw@ \cr \box\z@}}}
\makeatother

\title{Strang Splitting for Parametric Inference in Second-order Stochastic Differential Equations}

\author{ \href{https://orcid.org/0000-0002-8890-421X}{\includegraphics[scale=0.06]{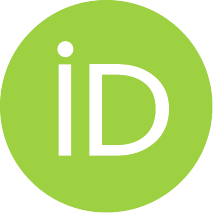}\hspace{1mm} Predrag ~Pilipovic}\\
	Department of Mathematical Sciences\\
	University of Copenhagen\\
	2100 Copenhagen, Denmark \\
    \texttt{predrag@math.ku.dk} \\
    Bielefeld Graduate School of Economics and Management\\
    University of Bielefeld\\
    33501 Bielefeld, Germany\\
	\texttt{predrag.pilipovic@uni-bielefeld.de} \\
	\And
	Adeline ~Samson \\
	Univ. Grenoble Alpes\\
	CNRS, Grenoble INP, LJK\\
	38000 Grenoble, France\\
	\texttt{adeline.leclercq-samson@univ-grenoble-alpes.fr} \\
	\And
	\href{https://orcid.org/0000-0002-1998-2783}{\includegraphics[scale=0.06]{orcid.pdf}\hspace{1mm} Susanne ~Ditlevsen} \\
	Department of Mathematical Sciences\\
	University of Copenhagen\\
	2100 Copenhagen, Denmark \\
	\texttt{susanne@math.ku.dk}\\
}

\date{}

\renewcommand{\shorttitle}{Strang Splitting Parameter Estimator for Second-order SDEs}

\hypersetup{
pdftitle={Strang Splitting for Parametric Inference in Second-order Stochastic Differential Equations},
pdfsubject={math.ST},
pdfauthor={Predrag ~Pilipovic, Adeline ~Samson, Susanne ~Ditlevsen},
pdfkeywords={First keyword, Second keyword, More},
}

\begin{document}
\maketitle

\begin{abstract}

We address parameter estimation in second-order stochastic differential equations (SDEs), which are prevalent in physics, biology, and ecology. The second-order SDE is converted to a first-order system by introducing an auxiliary velocity variable, which raises two main challenges. First, the system is hypoelliptic since the noise affects only the velocity, making the Euler-Maruyama estimator ill-conditioned. We propose an estimator based on the Strang splitting scheme to overcome this. Second, since the velocity is rarely observed, we adapt the estimator to partial observations. We present four estimators for complete and partial observations, using the full pseudo-likelihood or only the velocity-based partial pseudo-likelihood. These estimators are intuitive, easy to implement, and computationally fast, and we prove their consistency and asymptotic normality. Our analysis demonstrates that using the full pseudo-likelihood with complete observations reduces the asymptotic variance of the diffusion estimator. With partial observations, the asymptotic variance increases as a result of information loss but remains unaffected by the likelihood choice. However, a numerical study on the Kramers oscillator reveals that using the partial pseudo-likelihood for partial observations yields less biased estimators. We apply our approach to paleoclimate data from the Greenland ice core by fitting the Kramers oscillator model, capturing transitions between metastable states reflecting observed climatic conditions during glacial eras.
\end{abstract}


\section{Introduction} \label{sec:Intro}

Second-order stochastic differential equations (SDEs) are an effective instrument for modeling complex systems showcasing both deterministic and stochastic dynamics, which incorporate the second derivative of a variable - the acceleration. These models are extensively applied in many fields, including physics \citep{RosenblumPikovsky2003}, molecular dynamics \citep{leimkuhler2015molecular}, ecology \citep{CTCRWJohnson2008, CTCRWMichelot2019}, paleoclimatology \citep{ditlevsen2002fast}, and neuroscience \citep{Ziv1994, jansen1995electroencephalogram}.

The general form of a second-order SDE in Langevin form is given as follows
\begin{equation}
\ddot{\mathbf{X}}_t = \mathbf{F}(\mathbf{X}_t, \dot{\mathbf{X}}_t,  \bm{\beta}) + \bm{\Sigma} \bm{\xi}_t. \label{eq:2ndOrderLangevin}
\end{equation}
Here, $\mathbf{X}_t$ denotes the random variable of interest taking values in $\mathbb{R}^d$, the dot indicates derivative with respect to time $t$, drift $\mathbf{F}$ represents the deterministic force, and $\bm{\xi}_t$ is a white noise representing the system's random perturbations around the deterministic force. We assume that $\bm{\Sigma} \in \mathbb{R}^{d\times d}$ is constant, i.e., the noise is additive.

The main goal of this study is to estimate parameters in second-order SDEs. We first reformulate the $d$-dimensional second-order SDE \eqref{eq:2ndOrderLangevin} into a $2d$-dimensional first-order SDE. We define an auxiliary velocity variable and express the second-order SDE in terms of its position $\mathbf{X}_t$ and velocity $\mathbf{V}_t$
\begin{equation}
\begin{alignedat}{2}\label{eq:sdeXV}
\dif \mathbf{X}_t &= \mathbf{V}_t \dif t, \qquad &&\mathbf{X}_0 = \mathbf{x}_0, \\
\dif \mathbf{V}_t &= \mathbf{F}\left(\mathbf{X}_t, \mathbf{V}_t;\bm{\beta}\right) \dif t + \bm{\Sigma}\dif \mathbf{W}_t, \qquad &&\mathbf{V}_0 = \mathbf{v}_0,
\end{alignedat}
\end{equation}
where $\mathbf{W}_t$ is a standard Wiener process. We refer to $\mathbf{X}_t$ and $\mathbf{V}_t$ as the smooth and rough coordinates, respectively. SDE \eqref{eq:sdeXV} is sometimes referred to as an integrated diffusion (see, e.g., \cite{DitlevsenSorensen2004, Gloter2000, Gloter2006}).

A specific example of model \eqref{eq:sdeXV} is $\mathbf{F}(\mathbf{x}, \mathbf{v}) = -\mathbf{c} (\mathbf{x}, \mathbf{v}) \mathbf{v} - \nabla \mathbf{U}(\mathbf{x})$, for some function $\mathbf{c}(\cdot)$ and potential $\mathbf{U}(\cdot)$. Then, model \eqref{eq:sdeXV} is called a stochastic damping Hamiltonian system. This system describes the motion of a particle subjected to potential, dissipative, and random forces \citep{WU2001205}. An example of a stochastic damping Hamiltonian system is the Kramers oscillator introduced in Section \ref{sec:Kramers}.

Let
\begin{equation*}
    \mathbf{Y}_t = \begin{bmatrix}
        \mathbf{X}_t\\
        \mathbf{V}_t
    \end{bmatrix}, \quad
    \widetilde{\mathbf{F}}(\mathbf{x}, \mathbf{v}; \bm{\beta}) = \begin{bmatrix}
        \mathbf{v}\\
        \mathbf{F}(\mathbf{x}, \mathbf{v}; \bm{\beta})
    \end{bmatrix}, \quad \text{and} \quad
    \widetilde{\bm{\Sigma}} = \begin{bmatrix}
        \mathbf{0} & \mathbf{0}\\
        \mathbf{0} & \bm{\Sigma}
    \end{bmatrix}.
\end{equation*}
Then \eqref{eq:sdeXV} is formulated as
\begin{align}
\dif \mathbf{Y}_t &= \widetilde{\mathbf{F}}\left(\mathbf{Y}_t;\bm{\beta}\right) \dif t + \widetilde{\bm{\Sigma}}\dif \widetilde{\mathbf{W}}_t, \qquad \mathbf{Y}_0 = \mathbf{y}_0. \label{eq:SDE}
\end{align}
The extended random process $\widetilde{\mathbf{W}}_t$ takes values in $\mathbb{R}^{2d}$ such that its first $d$ coordinates do not matter, and the second $d$ coordinates are equal to $\mathbf{W}_t$. This extension is needed only for the consistency of the notation. The tilde $\ \widetilde{} \ $ over an object indicates that it is associated with process $\mathbf{Y}_t$. Specifically, the object is of dimension $2d$ or $2d \times 2d$. 

When it exists, the unique solution of \eqref{eq:SDE} is called a diffusion or diffusion process.  System \eqref{eq:SDE} is usually not fully observed since the velocity $\mathbf{V}_t$ is not observable. Thus, our primary objective is to estimate the underlying drift parameter $\bm{\beta}$ and the diffusion parameter ${\bm{\Sigma}}$, based on discrete observations of either $\mathbf{Y}_t$ (referred to as the complete observation case), or only $\mathbf{X}_t$ (referred to as the partial observation case). If $\bm{\Sigma}$ is of full rank $d$, the diffusion $\mathbf{Y}_t$ is said to be hypoelliptic since the matrix 
\begin{align}
\widetilde{\bm{\Sigma}}\widetilde{\bm{\Sigma}}^\top = \begin{bmatrix}
\bm{0} & \bm{0}\\
\bm{0} & \bm{\Sigma}\bm{\Sigma}^\top
\end{bmatrix} \label{eq:SigmaSigmaT}
\end{align}
is not of full rank, while $\mathbf{Y}_t$ admits a smooth density, see Section \ref{sec:Hypoellipticity}. Thus, \eqref{eq:sdeXV} is a subclass of a larger class of hypoelliptic diffusions. 

\subsection*{Literature review}

Parametric estimation for hypoelliptic diffusions is an active area of research. \cite{DitlevsenSorensen2004} studied discretely observed integrated diffusion processes. They proposed using prediction-based estimating functions, which are suitable for non-Markovian processes and do not require access to the unobserved component. They proved consistency and asymptotic normality of the estimators for the number of observations $N \to \infty$ without any requirements on the sampling interval $h$. Certain moment conditions are needed to obtain results for fixed $h$, which are often difficult to fulfill for nonlinear drift functions. The estimator was applied to paleoclimate data in \cite{ditlevsen2002fast}, similar to the data we analyze in Section \ref{sec:Greenland}.

\cite{Gloter2006} also focused on parametric estimation for discretely observed integrated diffusion processes, introducing a contrast function using the Euler-Maruyama (EM) discretization. He studied the asymptotic properties as $h \to 0$ and $N \to \infty$, under the so-called rapidly increasing experimental design $N h \to \infty$ and $N h^2 \to 0$. To address the ill-conditioned contrast from the EM discretization, he suggested using only the rough coordinates of the SDE. He proposed to recover the unobserved integrated component through the finite difference approximation $(\mathbf{X}_{t_{k+1}} - \mathbf{X}_{t_k})/h$. This approximation makes the estimator biased and requires a correction factor of 3/2 in one of the terms of the contrast function for partial observations. The correction increases the asymptotic variance of the estimator of the diffusion parameter. \cite{SamsonThieullen2012} expanded the ideas and proved the results of \citep{Gloter2006} in more general models. Their focus was also on contrasts using the EM discretization limited to only the rough coordinates. 

\cite{Pokern2009} proposed an It\^o-Taylor expansion, adding a noise term of order $h^{3/2}$ to the smooth component in the numerical scheme. They argued against the use of finite differences for approximating unobserved components. Instead, they suggested using the It\^o-Taylor expansion leading to non-degenerate conditionally Gaussian approximations of the transition density and using Markov Chain Monte Carlo (MCMC) Gibbs samplers for conditionally imputing missing components based on the observations. However, this resulted in a biased estimator of the drift parameter of the rough component.

\cite{ditlevsen2018hypoelliptic} focused on filtering and inference methods for complete and partial observations. They proposed a contrast estimator based on the strong order 1.5 scheme  \citep{KloedenPlaten}, which incorporates noise of order $h^{3/2}$ into the smooth component, similar to \cite{Pokern2009}. Moreover, they retained terms of order $h^2$ in the mean, which removed the bias in the drift parameters noted in \citep{Pokern2009}. They proved consistency and asymptotic normality under complete observations, with the standard rapidly increasing experimental design $N h \to \infty$ and $N h^2 \to 0$. They adopted an unconventional approach using two separate contrast functions, resulting in marginal asymptotic results rather than a joint central limit theorem. The model was limited to a scalar smooth component and a diagonal diffusion coefficient matrix for the rough component.

\cite{melnykova2020parametric} developed a contrast estimator using local linearization (LL) \citep{Ozaki1985StatisticalIO, ShojiOzaki1998, OzakiJimenezLorenz} and compared it to the least-squares estimator. She employed local linearization of the drift function, providing a non-degenerate conditional Gaussian discretization scheme, enabling the construction of a contrast estimator that achieves asymptotic normality under the standard conditions $N h \to \infty$ and $N h^2 \to 0$. She proved a joint central limit theorem, bypassing the need for two separate contrasts as in \cite{ditlevsen2018hypoelliptic}. The models in \cite{ditlevsen2018hypoelliptic} and \cite{melnykova2020parametric} allow for parameters in the smooth component of the drift, in contrast to models based on second-order SDEs. 

Recent work by \cite{gloter2020, gloter2021} introduced adaptive and non-adaptive methods in hypoelliptic diffusion models, proving asymptotic normality in the complete observation regime. Their non-adaptive estimator is based on a higher-order It\^o-Taylor expansion that introduces additional Gaussian noise onto the smooth coordinates along with an appropriate higher-order mean approximation. The resulting estimator was later termed the local (or locally) Gaussian (LG), which is different from LL. The LG estimator can be viewed as an extension of the estimator proposed in \cite{ditlevsen2018hypoelliptic}, with fewer restrictions on the class of models. \cite{gloter2020, gloter2021} found that using the full SDE to create a contrast reduces the asymptotic variance of the estimator of the diffusion parameter compared to methods using only rough coordinates in the case of complete observations. 

The most recent contributions are \cite{iguchi2023.2, iguchi2023.1, iguchi2023.3},  building on the foundation of the LG estimator and focusing on high-frequency regimes addressing limitations in earlier methods. \cite{iguchi2023.1} presented a new closed-form contrast estimator for hypoelliptic SDEs (denoted as Hypo-I) based on Edgeworth-type density expansion and Malliavin calculus that achieves asymptotic normality under the less restrictive condition of $N h^3 \to 0$. \cite{iguchi2023.2} focused on a highly degenerate class of SDEs (denoted as Hypo-II) where smooth coordinates split into further sub-groups and proposed estimators for both complete and partial observation settings. \cite{iguchi2023.3} further refined the conditions for asymptotic normality for both Hypo-I and Hypo-II estimators under a weak design $N h^p \to 0$, for $p \geq 2$. 

The existing methods are generally based on approximations with varying degrees of refinements to correct for possible nonlinearities. This implies that they quickly degrade for highly nonlinear models if the step size increases. In particular, this is the case for Hamiltonian systems. Instead, we propose to use splitting schemes, more precisely, the Strang splitting (SS) scheme. 

Splitting schemes are established techniques initially developed for solving ordinary differential equations (ODEs) and have proven to be effective also for SDEs \citep{Ableidinger2017, BukwarSamsonTamborrinoTubikanec2021, Pilipovic2024}. These schemes yield accurate results in many practical applications since they incorporate nonlinearities in their construction. This makes them particularly suitable for second-order SDEs, where they have been widely used. Early work in dissipative particle dynamics \citep{Shardlow2003, SERRANO2006}, applications to molecular dynamics \citep{Vanden-Eijnden2006, Melchionna2007, leimkuhler2015molecular} and studies on internal particles \citep{PAVLIOTIS2009} all highlight the scheme's versatility. \cite{Burrage2007}, \cite{Bou_Rabee_2010}, and \cite{Abdulle2015} focused on the long-run statistical properties such as invariant measures. \cite{BouRabee2017CayleySF, Bréhier2019} and \cite{Adams2022} used splitting schemes for stochastic partial differential equations (SPDEs). 

\subsection*{Approach}

Despite the extensive use of splitting schemes in different areas, statistical applications have been lacking.  We have recently proposed statistical estimators for elliptic SDEs \citep{Pilipovic2024}. The straightforward and intuitive schemes lead to easy-to-implement estimators, offering an advantage over more numerically intensive and less user-friendly state-of-the-art frequentist-based methods. We use the SS scheme to approximate the transition density between two consecutive observations and derive the pseudo-likelihood function since the exact likelihood function is often unknown or intractable. Then, to estimate parameters, we employ maximum likelihood estimation (MLE). However, two specific statistical problems arise due to hypoellipticity and partial observations. 

First, hypoellipticity leads to degenerate EM transition schemes, which can be addressed by constructing the pseudo-likelihood solely from the rough coordinates of the SDE, referred to as the rough objective function hereafter. The SS technique enables the estimator to incorporate both smooth and rough components (referred to as the full objective function). It is also possible to construct SS estimators using only the rough objective function, raising the question of which estimator performs better. Our results align with \cite{gloter2020, gloter2021} in the complete observation setting, where we find that using the full objective function reduces the asymptotic variance of the diffusion estimator. We found the same results in the simulation study for the LL estimator proposed by \cite{melnykova2020parametric}. 

Second, we suggest to treat the unobserved velocity by approximating it using finite difference methods. While \cite{Gloter2006} and \cite{SamsonThieullen2012} exclusively use forward differences, we also investigate central and backward differences. The forward difference approach leads to a biased estimator unless it is corrected. One of the main contributions of this work is finding suitable corrections of the pseudo-likelihoods for different finite difference approximations such that the SS estimators are asymptotically unbiased. This also ensures consistency of the diffusion parameter estimator at the cost of increasing its asymptotic variance. 

When only partial observations are available, we explore the impact of using the full objective function versus the rough objective function and how different finite differentiation approximations influence parametric inference. We find that the choice of objective function does not affect the asymptotic variance of the estimator. However, our simulation study on the Kramers oscillator suggests that using the full objective function in finite sample setups introduces more bias than using only the rough objective function, which is the opposite of the complete observation setting. Finally, we analyze a paleoclimate ice core dataset from Greenland using Kramers oscillator. 

\subsection*{Main results}

The main contributions of this paper are:
\begin{enumerate}
    \item We extend the SS estimator of \citep{Pilipovic2024} to hypoelliptic models given by second-order SDEs in complete and partial observation settings. Moreover, we find the appropriate correction factors to obtain consistency in the partial observation setting.
    \item When complete observations are available, we show that the asymptotic variance of the estimator of the diffusion parameter is smaller when optimizing the full objective function. In contrast, for partial observations, we show that the asymptotic variance remains unchanged regardless of using the full or rough objective function. 
    \item We discuss the influence on the statistical properties of using the forward difference approximation for imputing the unobserved velocity variables compared to using the backward or the central difference. 
    \item We evaluate the performance of the estimators through a simulation study of a second-order SDE, the Kramers oscillator. Additionally, we show numerically in a finite sample study that the rough objective function for partial observations is more favorable than the full objective function.
    \item Based on a simulation study on the Kramers oscillator, we conclude that the new estimators outperform estimators based on the EM, LG, and LL schemes based on accuracy and computational speed.
    \item We fit the Kramers oscillator to a paleoclimate ice core dataset from Greenland and estimate the average time needed to pass between two metastable states. 
\end{enumerate}

\subsection*{Outline}

In Section \ref{sec:ProblemSetup}, we introduce the class of SDE models, define hypoellipticity, introduce the Kramers oscillator, and explain the SS scheme and its associated estimators. The asymptotic properties of the estimator are established in Section \ref{sec:EstimatiorProperties}. The theoretical results are illustrated in a simulation study on the Kramers Oscillator in Section \ref{sec:Simulations}. Section \ref{sec:Greenland} illustrates our methodology on the Greenland ice core data, while the technical results and the proofs of the main theorems and properties are in Section \ref{sec:AuxiliaryProperties} and Appendix \ref{sec:Appendix}, respectively. In Sections \ref{sec:Discussion} and \ref{sec:Conclusion} we discuss the results and give concluding remarks. 

\subsection*{Notation}

We use capital bold letters for random vectors, vector-valued functions, and matrices, while lowercase bold letters denote deterministic vectors. $\|\cdot\|$ denotes the $L^2$ vector norm in $\mathbb{R}^d$. Superscript $(i)$ on a vector denotes the $i$-th component, while on a matrix it denotes the $i$-th column. Double subscript $ij$ on a matrix denotes the component in the $i$-th row and $j$-th column. The transpose is denoted by $\top$. Operator $\tr (\cdot)$ returns the trace of a matrix and $\det(\cdot)$ the determinant. $\mathbf{I}_d$ denotes the $d$-dimensional identity matrix, while $\bm{0}_{d\times d}$ is a $d$-dimensional zero square matrix. We denote by $[a_i]_{i=1}^d$ a vector with coordinates $a_i$, and by $[b_{ij}]_{i,j=1}^d$ a matrix with coordinates $b_{ij}$, for $i,j= 1,\ldots,d$. For a real-valued function $g: \mathbb{R}^d \to \mathbb{R}$, $\partial_{x^{(i)}} g(\mathbf{x})$ denotes the partial derivative with respect to $x^{(i)}$ and $\partial_{x^{(i)}x^{(j)}}^2 g(\mathbf{x})$ denotes the second partial derivative with respect to $x^{(i)}$ and $x^{(j)}$. The nabla operator $\nabla_{\mathbf{x}}$ denotes the gradient vector of $g$ with respect of $\mathbf{x}$, that is, $\nabla_{\mathbf{x}} g(\mathbf{x}) =[\partial_{x^{(i)}} g(\mathbf{x})]_{i=1}^d $. $\mathbb{H}$  denotes the Hessian matrix of function $g$, $\mathbb{H}_g(\mathbf{x})= [\partial_{x^{(i)}x^{(j)}} g(\mathbf{x})]_{i,j=1}^d$. For a vector-valued function $\mathbf{F}: \mathbb{R}^d \to \mathbb{R}^d$, the differential operator $D_\mathbf{x}$ denotes the Jacobian matrix $D_\mathbf{x} \mathbf{F}(\mathbf{x}) = [\partial_{x^{(i)}} F^{(j)}(\mathbf{x})]_{i,j=1}^d$.  $\mathbf{R}$ represents a vector (or a matrix) valued function defined on $(0, 1) \times \mathbb{R}^d$ (or $(0, 1) \times \mathbb{R}^{d\times d}$), such that, for some constant $C$, $\|\mathbf{R}(a, \mathbf{x})\| < a C (1 + \|\mathbf{x}\|)^C$ for all $a, \mathbf{x}$. When denoted by $R$, it refers to a scalar function. For an open set $A$, the bar $\overline{A}$ indicates closure. We write $\xrightarrow{\mathbb{P}}$ for convergence in probability $\mathbb{P}$.  

\section{Problem setup} \label{sec:ProblemSetup}

Let $\mathbf{Y} = (\mathbf{Y}_t)_{t\geq 0}$ in \eqref{eq:SDE} be defined on a complete probability space $(\Omega, \mathcal{F}, \mathbb{P}_{\bm{\theta}})$ with a complete right-continuous filtration $\mathcal{F} = (\mathcal{F}_t)_{t\geq 0}$, and let the $2d$-dimensional Wiener process $\widetilde{\mathbf{W}}= (\widetilde{\mathbf{W}}_t)_{t \geq 0}$ be adapted to $\mathcal{F}_{t}$. Let $\mathbf{W}_t = [\widetilde{\mathbf{W}}_t^{(i)}]_{i=d+1}^{2d}$. The probability measure $\mathbb{P}_{\bm{\theta}}$ is parameterized by the parameter $\bm{\theta} = \left(\bm{\beta}, \bm{\Sigma}\right)$. Rewrite equation \eqref{eq:SDE} as 
\begin{align}
\dif \mathbf{Y}_t = \widetilde{\mathbf{A}}(\bm{\beta}) (\mathbf{Y}_t - \widetilde{\mathbf{b}}(\bm{\beta}))\dif t + \widetilde{\mathbf{N}}\left(\mathbf{Y}_t; \bm{\beta}\right) \dif t + \widetilde{\bm{\Sigma}} \dif \widetilde{\mathbf{W}}_t, \label{eq:SDEsplitted}
\end{align}
where
\begin{align}
\widetilde{\mathbf{A}}(\bm{\beta}) = \begin{bmatrix}
\bm{0}_{d\times d} & \mathbf{I}_d\\
\mathbf{A}_{\mathbf{x}}(\bm{\beta}) & \mathbf{A}_{\mathbf{v}}(\bm{\beta})
\end{bmatrix}, \quad \widetilde{\mathbf{b}}(\bm{\beta}) = \begin{bmatrix}
\mathbf{b}(\bm{\beta})\\
\bm{0}_d
\end{bmatrix}, \quad \widetilde{\mathbf{N}}(\mathbf{x}, \mathbf{v}; \bm{\beta}) = \begin{bmatrix}
\bm{0}_d\\
\mathbf{N}(\mathbf{x}, \mathbf{v}; \bm{\beta})
\end{bmatrix}. \label{eq:ANtilde}
\end{align}
The function $\mathbf{F}$ in \eqref{eq:sdeXV} is thus split as $\mathbf{F}(\mathbf{x}, \mathbf{v}; \bm{\beta}) = \mathbf{A}_{\mathbf{x}}(\bm{\beta})(\mathbf{x} - \mathbf{b}(\bm{\beta})) + \mathbf{A}_{\mathbf{v}}(\bm{\beta})\mathbf{v} + \mathbf{N}(\mathbf{x}, \mathbf{v}; \bm{\beta})$. 

Let $\overline{\Theta}_\beta \times \overline{\Theta}_\Sigma = \overline{\Theta}$ denote the closure of the parameter space with $\Theta_\beta$ and $\Theta_\Sigma$ being two convex open bounded subsets of $\mathbb{R}^r$ and $\mathbb{R}^{d\times d}$, respectively. The function $\mathbf{N}: \mathbb{R}^{2d} \times \overline{\Theta}_\beta \to \mathbb{R}^d$ is assumed locally Lipschitz; functions $\mathbf{A}_{\mathbf{x}}$ and $\mathbf{A}_{\mathbf{v}}$ are defined on $\overline{\Theta}_\beta$ and take values in $\mathbb{R}^{d \times d}$; and the parameter matrix $\bm{\Sigma}$ takes values in $\mathbb{R}^{d \times d}$ and it has full rank. Matrix $\bm{\Sigma}\bm{\Sigma}^\top$ is assumed to be positive definite on $\overline{\Theta}_\Sigma$, shaping the variance of the rough coordinates. As any square root of $\bm{\Sigma}\bm{\Sigma}^\top$ induces the same distribution, $\bm{\Sigma}$ is identifiable only up to equivalence classes. Hence, estimation of the parameter $\bm{\Sigma}$ means estimation of $\bm{\Sigma}\bm{\Sigma}^\top$. The drift function $\widetilde{\mathbf{F}}$ in \eqref{eq:SDE} is divided into a linear part given by the matrix $\widetilde{\mathbf{A}}$ and a nonlinear part given by $\widetilde{\mathbf{N}}$. 

The true value of the parameter is denoted by $\bm{\theta}_0 = \left(\bm{\beta}_0, \bm{\Sigma}_0\right)$, and we assume that $\bm{\theta}_0 \in \Theta$. When referring to the true parameters, we write $\mathbf{A}_{\mathbf{x},0}$, $\mathbf{A}_{\mathbf{v},0}$, $\mathbf{b}_0$, $\mathbf{N}_0(\mathbf{x})$, $\mathbf{F}_0(\mathbf{x})$ and $\bm{\Sigma}\bm{\Sigma}_0^\top$ instead of $\mathbf{A}_{\mathbf{x}}(\bm{\beta}_0)$, $\mathbf{A}_{\mathbf{v}}(\bm{\beta}_0)$, $\mathbf{b}(\bm{\beta}_0)$, $\mathbf{N}(\mathbf{x}; \bm{\beta}_0)$, $\mathbf{F}(\mathbf{x}; \bm{\beta}_0)$ and $\bm{\Sigma}_0\bm{\Sigma}_0^\top$, respectively. We write $\mathbf{A}_{\mathbf{x}}$, $\mathbf{A}_{\mathbf{v}}$, $\mathbf{b}$, $\mathbf{N}(\mathbf{x})$, $\mathbf{F}(\mathbf{x})$, and $\bm{\Sigma}\bm{\Sigma}^\top$ in short for an arbitrary parameter $\bm{\theta}$.

\subsection{Example: The Kramers oscillator} \label{sec:Kramers}

To analyze the ice core data in Section \ref{sec:Greenland}, we propose a stochastic model of the escape dynamics in metastable systems, the Kramers oscillator \citep{Kramers1940BrownianTransitionState}, formulated initially to model the escape rate of Brownian particles from potential wells. The escape rate is related to the mean first passage time — the time needed for a particle to exceed the potential's local maximum for the first time, starting at a neighboring local minimum. This rate depends on the damping coefficient, noise intensity, temperature, and specific potential features, including the barrier's height and curvature at the minima and maxima. We apply this framework to quantify the rate of transitions between climate states.

Following \cite{ArnoldImkeller2000}, we introduce the Kramers oscillator as the stochastic Duffing oscillator - an example of a second-order SDE and a stochastic damping Hamiltonian system. The Duffing oscillator \citep{duffing1918erzwungene} is a forced nonlinear oscillator featuring a cubic stiffness term. The governing equation is given by
\begin{equation}
\ddot{x}_t + \eta \dot{x}_t + \frac{\dif }{\dif x} U(x_t) = f(t), \quad \text{where} \quad U(x) = -a \frac{x^2}{2} + b \frac{x^4}{4}, \quad \text{with} \quad a, b > 0, \quad \eta \geq 0. \label{eq:Duffing}
\end{equation}
Parameter $\eta$ in \eqref{eq:Duffing} indicates the damping level, $a$ regulates the linear stiffness, and $b$ determines the nonlinear component of the restoring force. Function $f$ represents the driving force and is usually set to $f(t) = \eta \cos(\omega t)$, which introduces deterministic chaos \citep{korsch1999chaos}.

When the driving force is $f(t) = \sqrt{2\eta \mathrm{T}} \xi(t)$, where $\xi(t)$ is white noise, equation \eqref{eq:Duffing} characterizes the stochastic movement of a particle within a bistable potential well, interpreting $\mathrm{T} > 0$ as the temperature of a heat bath. Setting $\sigma = \sqrt{2 \eta \mathrm{T}}$, equation \eqref{eq:Duffing} can be reformulated as an It\^o SDE for variables $X_t$ and $V_t = \dot{X}_t$ as
\begin{equation} \label{eq:KramersSDE}
    \begin{aligned}
\dif X_t &= V_t \dif t, \\
\dif V_t &= \left (-\eta V_t - \frac{\dif}{\dif x} U(X_t) \right) \dif t + \sigma \dif W_t,
\end{aligned}
\end{equation}
where $W_t$ denotes a standard Wiener process. The parameter set of SDE \eqref{eq:KramersSDE} is $\bm{\theta} = \{\eta, a, b, \sigma^2\}$. Figure \ref{fig:sim_data} illustrates an example of the Kramers oscillator trajectory simulated using the EM scheme with true parameters $\eta_0 = 62.5, a_0 = 297, b_0 = 219, \sigma_0^2 = 9125$ and with step size $h = 0.02$.

\begin{figure}
    \centering
    \includegraphics[width = \textwidth]{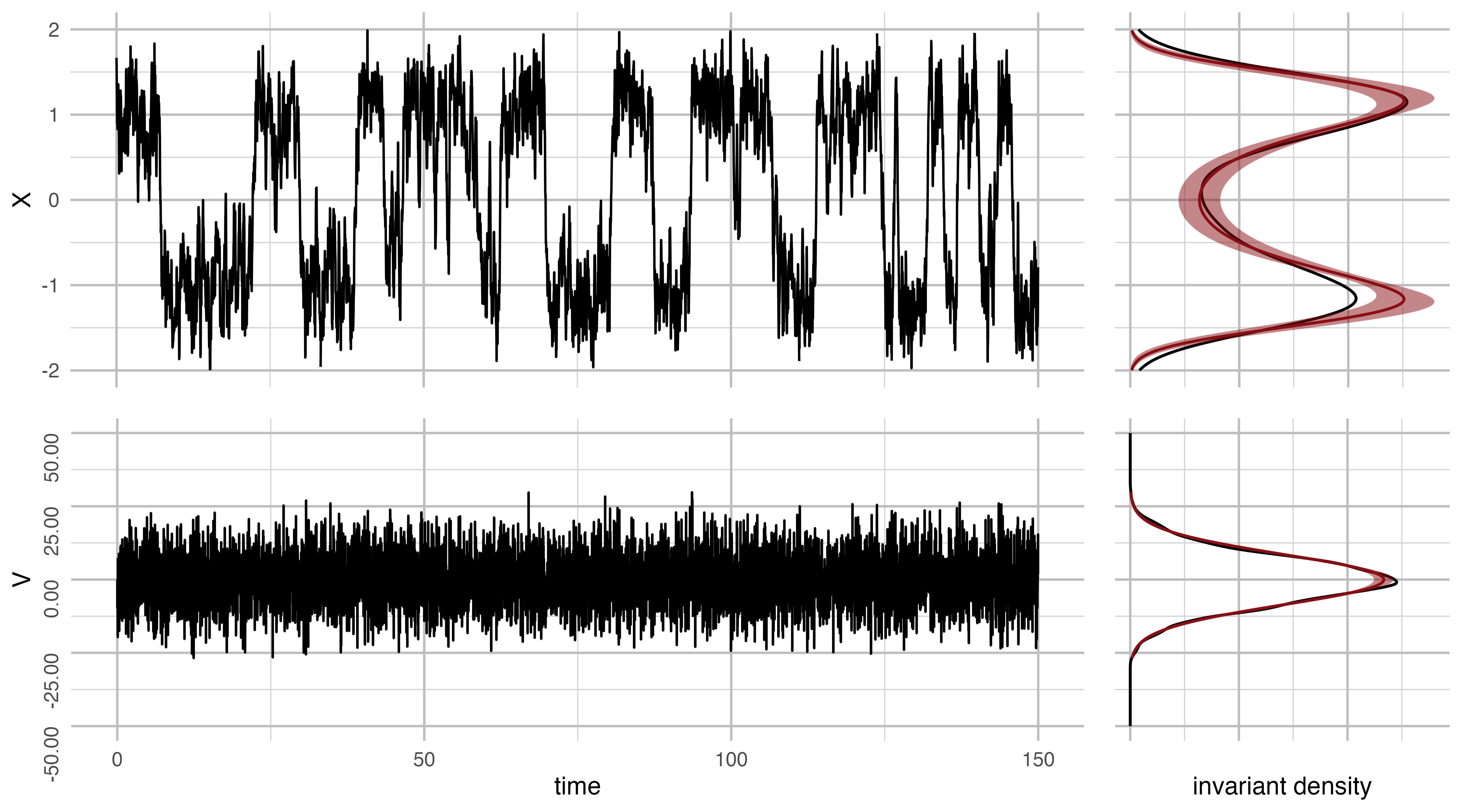}
    \caption{\textbf{Simulated trajectory from the Kramers oscillator.} The trajectory is simulated using Euler-Maruyama scheme with parameters $\eta_0 = 62.5, a_0 = 297, b_0 = 219, \sigma_0^2 = 9125$ and with the step size $h = 0.02$. \textbf{Left:} Trajectories over time of state $X$ (top) and rate of change $V$ (bottom). \textbf{Right:} Empirical densities (black) and estimated invariant densities with confidence intervals (dark red).}
    \label{fig:sim_data}
\end{figure}

The existence and uniqueness of the invariant measure $\nu_0(\dif x, \dif y)$ of \eqref{eq:KramersSDE} is proved in Theorem 3 in \citep{ArnoldImkeller2000}. The invariant measure $\nu_0$ is linked to the invariant density $\pi_0$ through $\nu_0(\dif x, \dif y) = \pi_0(x,v) \dif x \dif y$. Here we write $\pi_0(x,v)$ instead of $\pi(x,v;\bm{\theta}_0)$, and $\pi(x,v)$ instead of $\pi(x,v;\bm{\theta})$. The Fokker-Plank equation for $\pi$ is given by
\begin{equation}
    - v \frac{\partial }{\partial x}\pi(x, v) + \eta \pi(x, v) + \eta v \frac{\partial}{\partial v} \pi(x,v) + \frac{\dif}{\dif x}U(x) \frac{\partial}{\partial v}\pi(x,v) + \frac{\sigma^2}{2} \frac{\partial^2}{\partial v^2}\pi(x,v) = 0. \label{eq:Fokker-Plank}
\end{equation}
The invariant density that solves the Fokker-Plank equation is
\begin{equation}
    \pi(x,v) = C \exp\left(-\frac{2 \eta}{\sigma^2}U(x)\right) \exp\left(-\frac{\eta}{\sigma^2} v^2\right), \label{eq:XVinv}
\end{equation}
where $C$ is the normalizing constant. 

The partial invariant probability of $V_t$ is thus a Gaussian with zero mean and variance $\sigma^2/(2\eta)$. The partial invariant probability of $X_t$ is bimodal driven by the potential $U(x)$:
\begin{equation}
    \pi(x) = C \exp\left(-\frac{2 \eta}{\sigma^2} U(x)\right). \label{eq:Xinv}
\end{equation}
At steady state, for a particle moving in any potential $U(x)$ and driven by random Gaussian noise, the position $x$ and velocity $v$ are independent. This is reflected by the decomposition of the joint density $\pi(x,v)$ into $\pi(x) \pi(v)$.

The Fokker-Plank equation \eqref{eq:Fokker-Plank} can also be used to derive the mean first passage time $\tau$, which is inversely related to Kramers' escape rate $\kappa$  \citep{Kramers1940BrownianTransitionState}:
\begin{equation*}
\tau = \frac{1}{\kappa} \approx \frac{2 \pi }{\left(\sqrt{1 + \frac{\eta^2}{4 \omega^2}} - \frac{\eta}{2\omega}\right)\Omega}  \exp\left(\frac{\Delta U}{\rm T}\right),
\end{equation*}
where $x_\mathrm{barrier} = 0$ is the local maximum of $U(x)$ and $x_\mathrm{well} = \pm \sqrt{a/b}$ are the local minima, $\omega = \sqrt{|U''(x_\mathrm{barrier})|}=\sqrt{a}$, $\Omega = \sqrt{U''(x_\mathrm{well})} = \sqrt{2a}$, and $\Delta U = U(x_\mathrm{barrier}) - U(x_\mathrm{well})=a^2/4b$. The formula assumes strong friction, i.e., an over-damped system ($\eta \gg \omega$), and a small parameter $\mathrm{T}/\Delta U \ll 1$, indicating sufficiently deep potential wells. For the potential defined in \eqref{eq:Duffing}, the mean waiting time $\tau$ is then approximated by
\begin{equation}
\tau \approx \frac{\sqrt{2}\pi}{\sqrt{a + \frac{\eta^2}{4}} - \frac{\eta}{2}} \exp\left(\frac{a^2 \eta}{2b \sigma^2 } \right). \label{eq:tau}
\end{equation}

\subsection{Hypoellipticity} \label{sec:Hypoellipticity}

The SDE \eqref{eq:SDEsplitted} is said to be hypoelliptic if its quadratic diffusion matrix $\widetilde{\bm{\Sigma}}\widetilde{\bm{\Sigma}}^\top$ is not of full rank, while its solutions admit a smooth transition density with respect to the Lebesgue measure. According to H\"ormander's theorem \citep{nualart2006malliavin}, this is fulfilled if the SDE in its Stratonovich form satisfies the weak H\"ormander condition. Since $\bm{\Sigma}$ does not depend on $\mathbf{y}$, the It\^o and Stratonovich forms coincide. 

We begin by recalling the concept of Lie brackets: for smooth vector fields $\bm{f},\bm{g}: \mathbb{R}^{2d} \to \mathbb{R}^{2d}$, the Lie bracket, $[\bm{f}, \bm{g}]$, is defined as
\begin{equation*}
[\bm{f}, \bm{g}] \coloneqq D_{\mathbf{y}}\bm{g}(\mathbf{y}) \bm{f}(\mathbf{y}) - D_{\mathbf{y}}\bm{f}(\mathbf{y}) \bm{g}(\mathbf{y}).
\end{equation*}
We define the set $\mathcal{H}$ of vector fields by initially including $\widetilde{\bm{\Sigma}}^{(i)}$, $i = 1, 2, ..., 2d$, and then recursively adding Lie brackets
\begin{equation*}
H \in \mathcal{H} \Rightarrow [\widetilde{\mathbf{F}}, H], [\widetilde{\bm{\Sigma}}^{(1)}, H], \dots, [\widetilde{\bm{\Sigma}}^{(2d)}, H] \in \mathcal{H}.
\end{equation*}
The weak H\"ormander condition is met if the vectors in $\mathcal{H}$ span $\mathbb{R}^{2d}$ at every point $\mathbf{y} \in \mathbb{R}^{2d}$. The initial vectors span $\{(\mathbf{0}, \mathbf{v}) \in \mathbb{R}^{2d} \mid \mathbf{v} \in \mathbb{R}^d\}$, a $d$-dimensional subspace. We, therefore, need to verify the existence of some $H \in \mathcal{H}$ with a non-zero first element. The first iteration of the system yields
\begin{align*}
[\widetilde{\mathbf{F}}, \widetilde{\bm{\Sigma}}^{(i+d)}] &= -\bm{\Sigma}^{(i)}, && i = 1 ,2, \dots d,\\
[\widetilde{\bm{\Sigma}}^{(i)}, \widetilde{\bm{\Sigma}}^{(j)}] &= \mathbf{0}, && i, j = 1, 2, \dots 2d
\end{align*}
The first $d$ elements of $[\widetilde{\mathbf{F}}, \widetilde{\bm{\Sigma}}^{(i+d)}]$ are linearly independent, which results from the assumption that $\bm{\Sigma}$ has full rank. Thus, the second-order SDE defined in \eqref{eq:SDEsplitted} is always hypoelliptic. 

\subsection{Assumptions} \label{sec:Assumptions}

The following assumptions generalize those presented in \citep{Pilipovic2024}. 

Let $T>0$ be the length of the observed time interval. We assume that \eqref{eq:SDEsplitted} has a unique strong solution $\mathbf{Y} = \{\mathbf{Y}_t \mid t\in[0, T]\}$, adapted to $\mathcal{F} = \{\mathcal{F}_t \mid t\in[0, T]\}$, which follows from the following first two assumptions (Theorem 2 in \citep{Alyushina1988}, Theorem 1 in \citep{Krylov1991}, Theorem 3.5 in \citep{mao2007stochastic}). We need the last three assumptions to prove the properties of the estimators.

\begin{itemize}
\myitem{(A1)} \label{as:NLip} Function ${\mathbf{N}}$ is three times continuously differentiable with respect to both $\mathbf{y}$ and $\bm{\theta}$, i.e., ${\mathbf{N}} \in C^3(\mathbb{R}^{2d} \times \overline{\Theta}_\beta)$. Moreover, it is globally one-sided Lipschitz continuous with respect to $\mathbf{y}$ on $\mathbb{R}^{2d} \times \overline{\Theta}_\beta$. That is, there exists a constant $C > 0$ such that for all $\mathbf{y}_1, \mathbf{y}_2 \in \mathbb{R}^{2d}$,
\begin{equation*}
\left(\mathbf{y}_1 - \mathbf{y}_2\right)^\top\left({\mathbf{N}}(\mathbf{y}_1; \bm{\beta}) - {\mathbf{N}}(\mathbf{y}_2; \bm{\beta})\right) \leq C \|\mathbf{y}_1 - \mathbf{y}_2\|^2.
\end{equation*}
\myitem{(A2)} \label{as:NPoly} Function ${\mathbf{N}}$ exhibits at most polynomial growth in $\mathbf{y}$, uniformly in $\bm{\theta}$. Specifically, there exist constants $C > 0$ and $\chi \geq 1$ such that for all $\mathbf{y}_1, \mathbf{y}_2 \in \mathbb{R}^{2d}$,
\begin{equation*}
\|{\mathbf{N}}\left(\mathbf{y}_1;\bm{\beta}\right) - {\mathbf{N}}\left(\mathbf{y}_2; \bm{\beta}\right) \|^2 \leq C \left(1 + \|\mathbf{y}_1\|^{2\chi - 2} + \| \mathbf{y}_2\|^{2\chi - 2}\right) \| \mathbf{y}_1 - \mathbf{y}_2 \|^2.
\end{equation*}
Additionally, its derivatives exhibit polynomial growth in $\mathbf{y}$, uniformly in $\bm{\theta}$.
\myitem{(A3)} \label{as:Invariant} The solution $\mathbf{Y}$ to SDE \eqref{eq:SDEsplitted} has a unique invariant probability $\nu_0(\dif \mathbf{y})$.
\myitem{(A4)} \label{as:Ergodic} The solution $\mathbf{Y}$ to SDE \eqref{eq:SDEsplitted} satisfies a weak version of the ergodic theorem, namely
\begin{align*}
    \frac{1}{T} \int_0^T \bm{g}(\mathbf{Y}_s; \bm{\theta}) \dif s \xrightarrow[T \to \infty]{a.s.} \nu_0(\mathbf{g(\cdot; \bm{\theta})}),
\end{align*}
for any continuous function $\mathbf{g}$ with polynomial growth. 
\myitem{(A5)} \label{as:DiffusionInv} $\bm{\Sigma}\bm{\Sigma}^\top$ is invertible on $\overline{\Theta}_\Sigma$.
\myitem{(A6)} \label{as:Identifiability} $\bm{\beta}$ is identifiable, that is, if $\mathbf{F}(\mathbf{y}, \bm{\beta}_1) = \mathbf{F}(\mathbf{y}, \bm{\beta}_2)$ for all $\mathbf{y} \in \mathbb{R}^{2d}$, then $\bm{\beta}_1 = \bm{\beta}_2$.
\end{itemize}

Assumption \ref{as:NLip} ensures finiteness of the moments of the solution $\mathbf{X}$ \citep{TretyakovAndZhang}, i.e.,
\begin{align*}
\mathbb{E}[\sup\limits_{t\in [0, T]}\|\mathbf{Y}_t\|^{{2p}}] < C(1 + \|\mathbf{y}_0\|^{{2p}}), \hspace{3ex} \forall {\, p \geq 1}.  
\end{align*}
Assumptions \ref{as:Invariant} and \ref{as:Ergodic} are necessary for the ergodic theorem to ensure convergence in distribution. Assumption \ref{as:DiffusionInv} ensures that the model \eqref{eq:SDEsplitted} is hypoelliptic. Assumption \ref{as:Identifiability} ensures the identifiability of the drift parameter.

\subsection{Strang splitting scheme} \label{sec:StrangSplitting}

Consider the following splitting of \eqref{eq:SDEsplitted}
\begin{align}
    \dif \mathbf{Y}^{[1]}_t &= \widetilde{\mathbf{A}} (\mathbf{Y}^{[1]}_t - \widetilde{\mathbf{b}}) \dif t + \widetilde{\bm{\Sigma}} \dif \widetilde{\mathbf{W}}_t, && \mathbf{Y}^{[1]}_0 = \mathbf{y}_0, \label{eq:SplittingEq1}\\    
    \dif \mathbf{Y}^{[2]}_t &= \widetilde{\mathbf{N}}(\mathbf{Y}^{[2]}_t) \dif t, && \mathbf{Y}^{[2]}_0 = \mathbf{y}_0. \label{eq:SplittingEq2}
\end{align}
There are no assumptions on the choice of $\widetilde{\mathbf{A}}$ and $\widetilde{\mathbf{b}}$ except for the block structure from \eqref{eq:ANtilde} that follows from the second-order SDE assumption. Indeed, we show that the asymptotic results hold for any choice of $\widetilde{\mathbf{A}}_{\mathbf{x}}$, $\widetilde{\mathbf{A}}_{\mathbf{v}}$ and $\mathbf{b}$ in both the complete and the partial observation settings. This extends the results in \cite{Pilipovic2024}, which are also shown to hold in the elliptic complete observation case. While asymptotic results are invariant to the choice of $\widetilde{\mathbf{A}}$ and $\widetilde{\mathbf{b}}$, finite sample properties of the scheme and the corresponding estimators are very different and it is important to choose the splitting wisely. Intuitively, when the process is close to a fixed point of the drift, the linear dynamics are dominating, whereas far from the fixed points, the nonlinearities might be dominating. If the drift has a fixed point $\mathbf{y}^\star$, we therefore suggest setting  $\widetilde{\mathbf{A}} = D_\mathbf{y} \widetilde{\mathbf{F}}(\mathbf{y}^\star)$ and $\widetilde{\mathbf{b}} = \mathbf{y}^\star$. This choice is confirmed in simulations (see \cite{Pilipovic2024} for more details).

Solution of SDE \eqref{eq:SplittingEq1} is an Ornstein–Uhlenbeck (OU) process given by the following $h$-flow
\begin{align}
     \mathbf{Y}_{t_k}^{[1]} &= \Phi_h^{[1]}(\mathbf{Y}_{t_{k-1}}^{[1]}) = \widetilde{\bm{\mu}}_h(\mathbf{Y}_{t_{k-1}}^{[1]}; \bm{\beta}) +  \widetilde{\bm{\varepsilon}}_{h,k}, \label{eq:OU}\\
     \widetilde{\bm{\mu}}_h(\mathbf{y}; \bm{\beta}) &\coloneqq  e^{\widetilde{\mathbf{A}} h} (\mathbf{y} - \widetilde{\mathbf{b}}) + \widetilde{\mathbf{b}},\label{eq:muh}\\
     \widetilde{\bm{\Omega}}_h &= \int_0^h e^{ \widetilde{\mathbf{A}}(h-u)} \widetilde{\bm{\Sigma}} \widetilde{\bm{\Sigma}}^\top e^{ \widetilde{\mathbf{A}}^\top (h-u)} \dif u,\label{eq:Omegah_definition} 
\end{align}
where $\widetilde{\bm{\varepsilon}}_{h,k} \stackrel{i.i.d}{\sim} \mathcal{N}_{2d} (\bm{0}, \widetilde{\bm{\Omega}}_h )$ for $k=1, \ldots , N$. It is useful to rewrite $\widetilde{\bm{\Omega}}_h $ in the following block matrix form
\begin{equation}
\label{eq:Omegahblock}
    \widetilde{\bm{\Omega}}_h = \begin{bmatrix}
        \bm{\Omega}_h^{\mathrm{[SS]}} & \bm{\Omega}_h^{\mathrm{[SR]}}\\
        \bm{\Omega}_h^{\mathrm{[RS]}} & \bm{\Omega}_h^{\mathrm{[RR]}}
    \end{bmatrix}, 
\end{equation}
where $\mathrm{S}$ in the superscript stands for smooth and $\mathrm{R}$ stands for rough. The Schur complement of $\widetilde{\bm{\Omega}}_h$ with respect to $\bm{\Omega}_h^{\mathrm{[RR]}}$ and the determinant of $\widetilde{\bm{\Omega}}_h$ are given by
\begin{align*}
    \bm{\Omega}_h^{\mathrm{[S|R]}} \coloneqq \bm{\Omega}_h^{\mathrm{[SS]}}  - \bm{\Omega}_h^{\mathrm{[SR]}}(\bm{\Omega}_h^{\mathrm{[RR]}})^{-1}\bm{\Omega}_h^{\mathrm{[RS]}},\qquad \det \widetilde{\bm{\Omega}}_h = \det \bm{\Omega}_h^{\mathrm{[RR]}} \det \bm{\Omega}_h^{\mathrm{[S\mid R]}}.
\end{align*}
Assumptions \ref{as:NLip}-\ref{as:NPoly} ensure the existence and uniqueness of the solution of \eqref{eq:SplittingEq2} (Theorem 1.2.17 in \citep{Humphries2002}). Thus, there exists a unique function $\widetilde{\bm{f}}_h : \mathbb{R}^{2d} \times \Theta_\beta \to \mathbb{R}^{2d}$, for $h \geq0$, such that
\begin{equation*}
     \mathbf{Y}_{t_k}^{[2]} = \Phi_h^{[2]}(\mathbf{Y}_{t_{k-1}}^{[2]}) = \widetilde{\bm{f}}_h(\mathbf{Y}_{t_{k-1}}^{[2]}; \bm{\beta}).
\end{equation*}
For all $\bm{\beta} \in \Theta_\beta$, the $h$-flow $\widetilde{\bm{f}}_h$ fulfills the following semi-group properties
\begin{align*}
    \widetilde{\bm{f}}_0 (\mathbf{y}; \bm{\beta}) &= \mathbf{y}, \qquad \widetilde{\bm{f}}_{t+s}(\mathbf{y}; \bm{\beta}) = \widetilde{\bm{f}}_t(\widetilde{\bm{f}}_s(\mathbf{y}; \bm{\beta}); \bm{\beta}), \ \ t, s \geq 0. 
\end{align*}
For  $\mathbf{y}=(\mathbf{x}^\top, \mathbf{v}^\top)^\top$, we have
\begin{align}
     \widetilde{\bm{f}}_h(\mathbf{x}, \mathbf{v}; \bm{\beta}) = \begin{bmatrix}
         \mathbf{x}\\
         \bm{f}_h(\mathbf{x}, \mathbf{v}; \bm{\beta})
     \end{bmatrix}, \label{eq:fhtilde}
\end{align}
where $\bm{f}_h(\mathbf{x}, \mathbf{v}; \bm{\beta})$ is the solution of the ODE with vector field $\mathbf{N}(\mathbf{x}, \mathbf{v}; \bm{\beta})$.

We introduce another assumption to define the pseudo-likelihood based on the splitting scheme.
\begin{itemize}
    \myitem{(A7)} \label{as:fhInv} The inverse function $\widetilde{\bm{f}}_h^{-1}(\mathbf{y}; \bm{\beta})$ is defined asymptotically for all $\mathbf{y}\in \mathbb{R}^{2d}$ and all $\bm{\beta} \in \Theta_\beta$, when $h \to 0$.
\end{itemize}

Then, the inverse of $\tilde{\bm{f}}_h$ can be decomposed as
\begin{align}
    \widetilde{\bm{f}}_h^{-1}(\mathbf{x}, \mathbf{v}; \bm{\beta}) = \begin{bmatrix}
         \mathbf{x}\\
         \bm{f}^{\star -1}_h(\mathbf{x}, \mathbf{v}; \bm{\beta})
     \end{bmatrix}, \label{eq:fhstartildeinv}
\end{align}
where $\bm{f}^{\star -1}_h(\mathbf{x}, \mathbf{v}; \bm{\beta})$ is the rough part of the inverse of $\widetilde{\bm{f}}_h^{-1}$. Note that $\bm{f}^{\star -1}_h$ does not equal $\bm{f}_{h}^{-1}$ since the inverse does not propagate through coordinates when $\bm{f}_h$  depends on $\mathbf{x}$.

We are now ready to define the Strang splitting scheme for model \eqref{eq:SDEsplitted}.
\begin{definition}(Strang splitting) \label{def:splitting}
Let Assumptions \ref{as:NLip}-\ref{as:NPoly} hold. The Strang splitting approximation (SS) of the solution of \eqref{eq:SDEsplitted} is given by
\begin{align}
    \Phi_h^\mathrm{[SS]}({\mathbf{Y}}^\mathrm{[SS]}_{t_{k-1}}) = (\Phi_{h/2}^{[2]} \circ \Phi_h^{[1]} \circ \Phi_{h/2}^{[2]} )({\mathbf{Y}}^\mathrm{[SS]}_{t_{k-1}}) = \widetilde{\bm{f}}_{h/2} (\widetilde{\bm{\mu}}_h(\widetilde{\bm{f}}_{h/2}({\mathbf{Y}}^\mathrm{[SS]}_{t_{k-1}})) +  \widetilde{\bm{\varepsilon}}_{h,k}). \label{eq:StrangSplitting}
\end{align}
\end{definition}

\begin{remark} \label{rmkr:Strang}
The order of composition in the splitting scheme is not unique. Changing the order leads to a sum of 2 independent random variables, one Gaussian and one non-Gaussian, whose likelihood is not trivial. Thus, we only use the splitting \eqref{eq:StrangSplitting}. 
\end{remark}

\subsection{Strang splitting estimators}

In this section, we introduce four estimators based on the SS scheme. We distinguish between estimators based on complete observations (denoted by $\mathrm{C}$ when both $\mathbf{X}$ and $\mathbf{V}$ are observed) and partial observations (denoted by $\mathrm{P}$ when only $\mathbf{X}$ is observed). In applications, we often only have access to partial observations. However, complete observations can occur and are theoretically interesting in their own right. Furthermore, the complete observation estimator is used as a building block for the partial observation case. We distinguish the estimators based on the type of objective function employed. These are the full objective function (denoted by $\mathrm{F}$) based on the full pseudo-likelihood and the rough objective function (denoted by $\mathrm{R}$) based on the velocity-only partial pseudo-likelihood. To decompose the full pseudo-likelihood, we furthermore introduce the conditional pseudo-likelihood based on the smooth component given the rough part (denoted by $\mathrm{S}\mid\mathrm{R}$).

\subsubsection{Complete observations}

Assume we observe the complete sample $\mathbf{Y}_{0:t_N} := (\mathbf{Y}_{t_k})_{k=1}^N$ from \eqref{eq:SDEsplitted} at time steps $0 = t_0 < t_1 < ... < t_N = T$. For notational simplicity, we assume equidistant step size $h = t_k - t_{k-1}$. The SS scheme \eqref{eq:StrangSplitting} is a nonlinear transformation of a Gaussian random variable $\widetilde{\bm{\mu}}_h(\widetilde{\bm{f}}_{h/2}({\mathbf{Y}}^\mathrm{[SS]}_{t_{k-1}})) +  \widetilde{\bm{\varepsilon}}_{h,k}$. We define
\begin{align}
     \widetilde{\mathbf{Z}}_{k, k-1}(\bm{\beta}) &\coloneqq \widetilde{\bm{f}}_{h/2}^{-1}(\mathbf{Y}_{t_k}; \bm{\beta}) - \widetilde{\bm{\mu}}_h(\widetilde{\bm{f}}_{h/2}(\mathbf{Y}_{t_{k-1}};\bm{\beta});\bm{\beta}), \label{eq:Ztktilde}
\end{align}
and apply change of variables to get
\begin{align*}
    p(\mathbf{y}_{t_k}\mid \mathbf{y}_{t_{k-1}}) &= \varphi\left(\widetilde{\bm{f}}_{h/2}^{-1}(\mathbf{y}_{t_k}; \bm{\beta}); \widetilde{\bm{\mu}}_h(\widetilde{\bm{f}}_{h/2}(\mathbf{y}_{t_{k-1}};\bm{\beta});\bm{\beta}), \widetilde{\bm{\Omega}}_h(\bm{\theta})\right)|\det D_\mathbf{y} \widetilde{\bm{f}}_{h/2}^{-1}(\mathbf{y}_{t_k})|\\
    &= \varphi(\widetilde{\mathbf{z}}_{k, k-1}; \mathbf{0}, \widetilde{\bm{\Omega}}_h(\bm{\theta}))|\det D_\mathbf{y} \widetilde{\bm{f}}_{h/2}^{-1}(\mathbf{y}_{t_k})|,
\end{align*}
where $\varphi(\mathbf{x}; \bm{\mu}, \bm{\Sigma})$ is the Gaussian $\mathcal{N}(\bm{\mu}, \bm{\Sigma})$ density evaluated at $\mathbf{x}$.

Using that $-\log |\det D_\mathbf{y} \widetilde{\bm{f}}_{h/2}^{-1}\left(\mathbf{y}; \bm{\beta}\right)| = \log |\det D_\mathbf{y} \widetilde{\bm{f}}_{h/2}\left(\mathbf{y};\bm{\beta}\right)|$ and $\det D_\mathbf{y} \widetilde{\bm{f}}_{h/2}\left(\mathbf{y};\bm{\beta}\right) = \det D_\mathbf{v} \bm{f}_{h/2}\left(\mathbf{y};\bm{\beta}\right)$, together with the Markov property of $\mathbf{Y}_{0:t_N}$, we get the following objective function based on the full pseudo-likelihood
\begin{align}
\mathcal{L}^{\mathrm{[CF]}}(\mathbf{Y}_{0:t_N}; \bm{\theta}) &\coloneqq \sum_{k=1}^N \left(\log\det \widetilde{\bm{\Omega}}_h(\bm{\theta}) +   \widetilde{\mathbf{Z}}_{k, k-1}(\bm{\beta})^\top \widetilde{\bm{\Omega}}_h(\bm{\theta})^{-1}  \widetilde{\mathbf{Z}}_{k, k-1}(\bm{\beta})+ 2 \log |\det D_\mathbf{v}\bm{f}_{h/2}(\mathbf{Y}_{t_k}; \bm{\beta})|\right). \label{eq:L_CF}
\end{align}
Now, split $\widetilde{\mathbf{Z}}_{k, k-1}$ from \eqref{eq:Ztktilde} into the smooth and rough parts $\widetilde{\mathbf{Z}}_{k, k-1} = ((\mathbf{Z}_{k, k-1}^{\mathrm{[S]}})^\top, (\mathbf{Z}_{k, k-1}^{\mathrm{[R]}})^\top)^\top$ defined as
\begin{align}
    \mathbf{Z}_{k, k-1}^{\mathrm{[S]}}(\bm{\beta}) &\coloneqq  [\widetilde{Z}_{k, k-1}^{(i)}(\bm{\beta})]_{i = 1}^{d} = \mathbf{X}_{t_k} - \bm{\mu}_h^{\mathrm{[S]}}(\widetilde{\bm{f}}_{h/2}(\mathbf{Y}_{t_{k-1}}; \bm{\beta});\bm{\beta}), \label{eq:ZtkS} \\
    \mathbf{Z}_{k, k-1}^{\mathrm{[R]}}(\bm{\beta}) &\coloneqq  [\widetilde{Z}_{k, k-1}^{(i)}(\bm{\beta})]_{i = d+1}^{2d} = \bm{f}^{\star -1}_{h/2}(\mathbf{Y}_{t_k}; \bm{\beta}) - \bm{\mu}_h^{\mathrm{[R]}}(\widetilde{\bm{f}}_{h/2}(\mathbf{Y}_{t_{k-1}}; \bm{\beta});\bm{\beta}),\label{eq:ZtkR}
\end{align}
where 
\begin{equation}
    \bm{\mu}_h^{\mathrm{[S]}}(\mathbf{y}; \bm{\beta}) \coloneqq [\widetilde{\mu}^{(i)}_h(\mathbf{y}; \bm{\beta})]_{i = 1}^{d}, \qquad 
   \bm{\mu}_h^{\mathrm{[R]}}(\mathbf{y}; \bm{\beta}) \coloneqq [\widetilde{\mu}^{(i)}_h(\mathbf{y}; \bm{\beta})]_{i = d+1}^{2d}.\label{eq:muh_splitted}
\end{equation}
We also define the following sequence of vectors
\begin{equation}
    \mathbf{Z}_{k, k-1}^{\mathrm{[S\mid R]}}(\bm{\beta}) \coloneqq  \mathbf{Z}_{k, k-1}^{\mathrm{[S]}}(\bm{\beta}) - \bm{\Omega}_h^{\mathrm{[SR]}}(\bm{\Omega}_h^{\mathrm{[RR]}})^{-1}\mathbf{Z}_{k, k-1}^{\mathrm{[R]}}(\bm{\beta}).
\end{equation}
Then, we rewrite the joint density as a product of the marginal and conditional densities
\begin{align*}
    \varphi(\widetilde{\mathbf{z}}_{k, k-1}; \mathbf{0}, \widetilde{\bm{\Omega}}_h) &= \varphi(\mathbf{z}_{k, k-1}^{\mathrm{[R]}}; \mathbf{0}, \bm{\Omega}_h^{\mathrm{[RR]}}) \varphi(\mathbf{z}_{k, k-1}^{\mathrm{[S]}}; \bm{\Omega}_h^{\mathrm{[SR]}}(\bm{\Omega}_h^{\mathrm{[RR]}})^{-1}\mathbf{z}_{k, k-1}^{\mathrm{[R]}}, \bm{\Omega}_h^{\mathrm{[S\mid R]}}). 
\end{align*}
This leads to dividing the full objective function $\mathcal{L}^{\mathrm{[CF]}}$ into a sum of the rough objective function $\mathcal{L}^{\mathrm{[CR]}}(\mathbf{Y}_{0:t_N}; \bm{\theta})$ and the smooth-given-rough objective function $\mathcal{L}^{\mathrm{[CS\mid R]}}(\mathbf{Y}_{0:t_N}; \bm{\theta})$
\begin{align*}
    \mathcal{L}^{\mathrm{[CF]}}(\mathbf{Y}_{0:t_N}; \bm{\theta}) &= \mathcal{L}^{\mathrm{[CR]}}(\mathbf{Y}_{0:t_N}; \bm{\theta}) + \mathcal{L}^{\mathrm{[CS\mid R]}}(\mathbf{Y}_{0:t_N}; \bm{\theta}),
\end{align*}
where 
\begin{align}
&\mathcal{L}^{\mathrm{[CR]}}\left(\mathbf{Y}_{0:t_N}; \bm{\theta}\right) \coloneqq \sum_{k=1}^N\Bigg(\log\det \bm{\Omega}_h^{\mathrm{[RR]}}(\bm{\theta}) +\mathbf{Z}_{k, k-1}^{\mathrm{[R]}}\left(\bm{\beta}\right)^\top \bm{\Omega}_h^{\mathrm{[RR]}}(\bm{\theta})^{-1} \mathbf{Z}_{k, k-1}^{\mathrm{[R]}}\left(\bm{\beta}\right)\notag\\
&\hspace{25ex}+ 2 \log \left|\det D_\mathbf{v}\bm{f}_{h/2}\left(\mathbf{Y}_{t_k}; \bm{\beta}\right)\right|\Bigg), \label{eq:L_CR}\\
&\mathcal{L}^{\mathrm{[CS\mid R]}}\left(\mathbf{Y}_{0:t_N}; \bm{\theta}\right) \coloneqq  \sum_{k=1}^N\left( \log\det \bm{\Omega}_h^{\mathrm{[S\mid R]}}(\bm{\theta}) + \mathbf{Z}_{k, k-1}^{\mathrm{[S\mid R]}}(\bm{\beta})^\top \bm{\Omega}_h^{\mathrm{[S\mid R]}}(\bm{\theta})^{-1} \mathbf{Z}_{k, k-1}^{\mathrm{[S\mid R]}}(\bm{\beta})\right).\label{eq:L_CSR}
\end{align}
The terms containing the drift parameter in $\mathcal{L}^{\mathrm{[CR]}}$ in \eqref{eq:L_CR} are of order $h^{1/2}$, as in the elliptic case, whereas the terms containing the drift parameter in $\mathcal{L}^{\mathrm{[CS\mid R]}}$ in \eqref{eq:L_CSR} are of order $h^{3/2}$. Consequently, under a rapidly increasing experimental design where $N h \to \infty$ and $N h^2 \to 0$, the objective function \eqref{eq:L_CSR} is degenerate for estimating the drift parameter. However, it contributes to the estimation of the diffusion parameter when the full objective function \eqref{eq:L_CF} is used. We show in later sections that employing \eqref{eq:L_CF} results in a lower asymptotic variance for the diffusion parameter making it more efficient in the complete observation scenario.

The estimators based on complete observations are then defined as
\begin{equation*}
    \hat{\bm{\theta}}_N^\mathrm{[obj]}  \coloneqq \argmin_{\bm{\theta}}  \mathcal{L}^\mathrm{[obj]} \left(\mathbf{Y}_{0:t_N}; \bm{\theta}\right), \quad \mathrm{obj} \in\{\mathrm{[CF]}, \mathrm{[CR]}\}. 
\end{equation*}
Although the full objective function is based on twice as many equations as the rough objective function, its implementation complexity, speed, and memory requirements are similar to the rough objective function. Therefore, if the complete observations are available, we recommend using the objective function \eqref{eq:L_CF} based on the full pseudo-likelihood.

\subsubsection{Partial observations}

Assume that we only observe the smooth coordinates $\mathbf{X}_{0:t_N} \coloneqq (\mathbf{X}_{t_k})_{k=0}^N$. The observed process $\mathbf{X}_t$ alone is not a Markov process, although the complete process $\mathbf{Y}_t$ is. To approximate $\mathbf{V}_{t_k}$, we define the backward difference process
\begin{equation}
\Delta_h \mathbf{X}_{t_k} \coloneqq \frac{\mathbf{X}_{t_k} - \mathbf{X}_{t_{k-1}}}{h}. \label{eq:DeltahX}
\end{equation}
From SDE \eqref{eq:sdeXV} it follows that
\begin{equation}
\Delta_h \mathbf{X}_{t_k} = \frac{1}{h}\int_{t_{k-1}}^{t_k} \mathbf{V}_t \dif t. \label{eq:DeltahX_as_V}
\end{equation}
We approximate $\mathbf{V}_{t_k}$ by applying one of the three approximations:
\begin{enumerate}
\item Backward difference $\mathbf{V}_{t_k} \approx \Delta_h \mathbf{X}_{t_k}$;
\item Forward difference $\mathbf{V}_{t_k} \approx \Delta_h \mathbf{X}_{t_{k+1}}$;
\item Central difference $\mathbf{V}_{t_k} \approx \frac{\Delta_h \mathbf{X}_{t_k} + \Delta_h \mathbf{X}_{t_{k+1}}}{2}$.
\end{enumerate}
In our simulation study, the approximation that performs the best is the forward difference, which is also the one used in \cite{Gloter2006} and \cite{SamsonThieullen2012}.

In numerical approximations of ODEs, backward and forward finite differences have the same order of convergence, whereas the central difference has a higher convergence rate, making it a better approximation. However, in the context of stochastic processes, finite differences act as natural smoothers by reducing the variance. Consequently, the central difference $(\mathbf{X}_{t_{k+1}} - \mathbf{X}_{t_{k-1}})/2h$ decreases the true variance more than the forward or backward differences, making it less suitable for statistical purposes. This point is also discussed in Remark \ref{rmrk:Differences}.

Thus, we focus exclusively on forward differences, following \cite{Gloter2006, SamsonThieullen2012}, and all proofs are carried out for this approximation. Similar results hold for the backward difference, with adjustments required in the conditional moments due to filtration issues.

We start by approximating $\widetilde{\mathbf{Z}}$ for the case of partial observations denoted by $\widetilde{\overline{\mathbf{Z}}}$
\begin{align}
     \widetilde{\overline{\mathbf{Z}}}_{k+1, k, k-1}(\bm{\beta}) &\coloneqq \widetilde{\bm{f}}_{h/2}^{-1}(\mathbf{X}_{t_k}, \Delta_h \mathbf{X}_{t_{k+1}}; \bm{\beta}) - \widetilde{\bm{\mu}}_h(\widetilde{\bm{f}}_{h/2}(\mathbf{X}_{t_{k-1}}, \Delta_h \mathbf{X}_{t_{k}};\bm{\beta});\bm{\beta}). \label{eq:Ztkbartilde}
\end{align}
The smooth and rough parts of $\widetilde{\overline{\mathbf{Z}}}$ are thus equal to
\begin{align}
\overline{\mathbf{Z}}_{k, k-1}^{[\mathrm{S}]}(\bm{\beta}) &\coloneqq \mathbf{X}_{t_k} - \bm{\mu}_h^{\mathrm{[S]}}(\widetilde{\bm{f}}_{h/2}(\mathbf{X}_{t_{k-1}}, \Delta_h \mathbf{X}_{t_{k}};\bm{\beta});\bm{\beta}), \label{eq:ZtkSbar}\\
\overline{\mathbf{Z}}_{k+1, k, k-1}^{[\mathrm{R}]}(\bm{\beta}) &\coloneqq \bm{f}^{\star -1}_{h/2}(\mathbf{X}_{t_k}, \Delta_h \mathbf{X}_{t_{k+1}}; \bm{\beta}) - \bm{\mu}_h^{\mathrm{[R]}}(\widetilde{\bm{f}}_{h/2}(\mathbf{X}_{t_{k-1}}, \Delta_h \mathbf{X}_{t_{k}};\bm{\beta});\bm{\beta}), \label{eq:ZtkRbar}
\end{align}
and
\begin{equation}
    \overline{\mathbf{Z}}_{k+1, k, k-1}^{\mathrm{[S\mid R]}}(\bm{\beta}) \coloneqq  \overline{\mathbf{Z}}_{k, k-1}^{[\mathrm{S}]}(\bm{\beta}) - \bm{\Omega}_h^{\mathrm{[SR]}}(\bm{\Omega}_h^{\mathrm{[RR]}})^{-1}\overline{\mathbf{Z}}_{k+1, k, k-1}^{\mathrm{[R]}}(\bm{\beta}).
\end{equation}
Compared to $\mathbf{Z}_{k, k-1}^{\mathrm{[R]}}$ in \eqref{eq:ZtkR}, $\overline{\mathbf{Z}}_{k+1, k, k-1}^{[\mathrm{R}]}$ in \eqref{eq:ZtkRbar} depends on three consecutive data points, with the additional point $\mathbf{X}_{t_{k+1}}$ entering through $\Delta_h \mathbf{X}_{t_{k+1}}$. Furthermore, $\mathbf{X}_{t_k}$ enters both $\bm{f}^{\star -1}_{h/2}$ and $\widetilde{\bm{\mu}}^{\mathrm{[R]}}_h$, making them coupled. This coupling significantly influences later derivations of the estimator's asymptotic properties, in contrast to the elliptic case where the derivations simplify.

Although it may seem straightforward to incorporate $\widetilde{\overline{\mathbf{Z}}}$, $\overline{\mathbf{Z}}_{k, k-1}^{[\mathrm{S}]}$ and $\overline{\mathbf{Z}}_{k, k-1}^{[\mathrm{R}]}$ into the objective functions \eqref{eq:L_CF}, \eqref{eq:L_CR} and \eqref{eq:L_CSR}, it introduces bias in the estimators of the diffusion parameters, as also discussed in \cite{Gloter2006, SamsonThieullen2012}. The bias arises because the finite differences underestimate the covariances. It turns out that both forward and backward differences decrease the variance of the rough part by $2/3$ (see Remark \ref{rmrk:Differences}). Figure \ref{fig:sim_V} illustrates this discrepancy of 2/3 between the variance of the real velocity $V_{t_k}$ from the Kramers oscillator \eqref{eq:KramersSDE} and its forward approximation $\Delta_h X_{t_{k+1}}$.

\begin{figure}
    \centering
    \includegraphics[width = \textwidth]{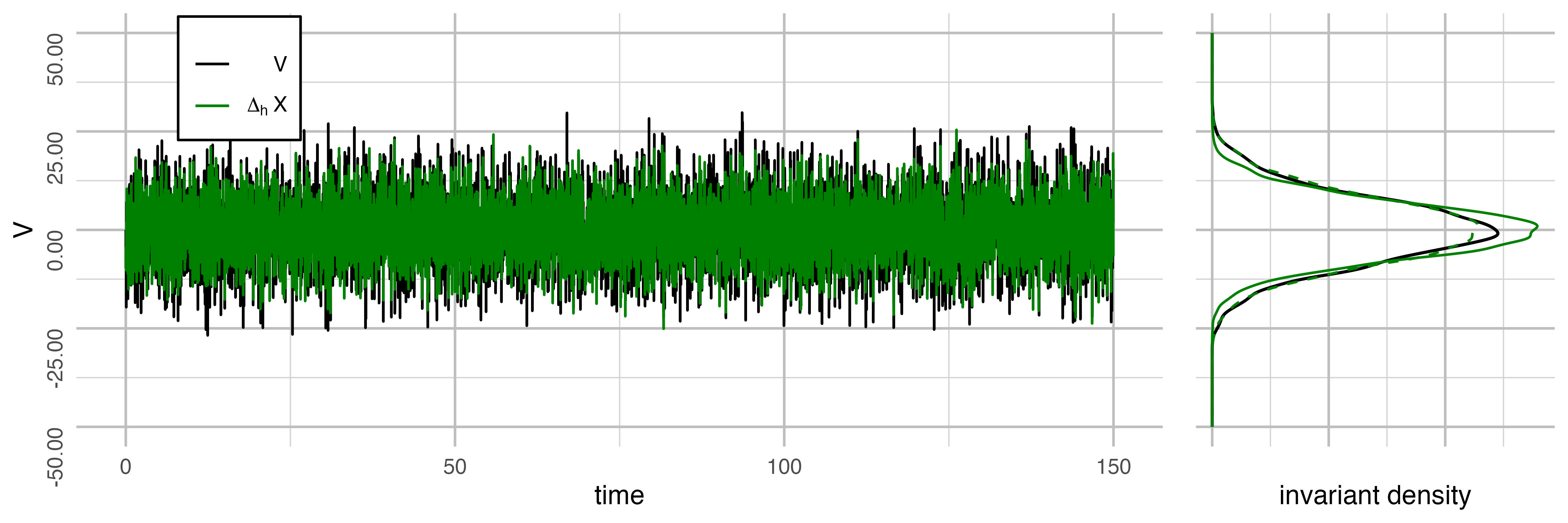}
    \caption{\textbf{Simulated velocity vs. approximated velocity from the Kramers oscillator.} The trajectory is simulated as in Figure \ref{fig:sim_data}. \textbf{Left:} Simulated trajectory over time of velocity $V$ (black) and approximated velocity using forward difference $\Delta_h X$ (green). \textbf{Right:} Empirical densities of $V$ (black), $\Delta_h X$ (green solid) and $\sqrt{3/2}\Delta_h X$ (green dashed).}
    \label{fig:sim_V}
\end{figure}

An obvious way to correct this underestimation would be to multiply $\Delta_h X$ in eq. \eqref{eq:DeltahX} by $\sqrt{3/2}$. However, this approach would result in an inconsistent estimator due to breaking the connection in eq. \eqref{eq:DeltahX_as_V}. Nonetheless, it may be beneficial to multiply the velocity $V$ by $\sqrt{3/2}$ in figures to make results comparable. 

To eliminate this bias, \cite{Gloter2006, SamsonThieullen2012} applied a correction of $2/3$ multiplied to $\log \det$ of the covariance term in the objective functions, which is $\log \det \bm{\Sigma} \bm{\Sigma}^\top$ in the EM discretization. We also need appropriate corrections to the objective functions \eqref{eq:L_CF}, \eqref{eq:L_CR} and \eqref{eq:L_CSR}, however, caution is needed because $\log\det \widetilde{\bm{\Omega}}_h(\bm{\theta})$ depends on both drift and diffusion parameters. To counterbalance this, we also incorporate an adjustment to $h$ in $\bm{\Omega}_h$. Moreover, we add the term $4\log |\det D_\mathbf{v}\bm{f}_{h/2}|$ to the objective function \eqref{eq:L_CSR} to obtain consistency of the drift estimator under partial observations. The detailed derivation of these correction factors will be elaborated in the following sections.

We, thus, propose the following objective functions
\begin{align}
    \mathcal{L}^{\mathrm{[PF]}}(\mathbf{X}_{0:t_N}; \bm{\theta}) &\coloneqq \frac{4}{3}(N-2)\log\det \widetilde{\bm{\Omega}}_{3h/4}(\bm{\theta}) \label{eq:L_PF}\\
    &+\sum_{k=1}^{N-1} \left(\widetilde{\overline{\mathbf{Z}}}_{k+1, k, k-1}(\bm{\beta})^\top \widetilde{\bm{\Omega}}_h(\bm{\theta})^{-1}  \widetilde{\overline{\mathbf{Z}}}_{k+1, k, k-1}(\bm{\beta})+ 6 \log |\det D_\mathbf{v}\bm{f}_{h/2}(\mathbf{X}_{t_k}, \Delta_h \mathbf{X}_{t_{k+1}}; \bm{\beta})|\right),\notag\\
    \mathcal{L}^{\mathrm{[PR]}}\left(\mathbf{X}_{0:t_N}; \bm{\theta}\right) &\coloneqq \frac{2}{3}(N-2) \log\det \bm{\Omega}_{3h/2}^{\mathrm{[RR]}}(\bm{\theta})\label{eq:L_PR}\\
    &\hspace{-2ex} + \sum_{k=1}^{N-1} \left(\overline{\mathbf{Z}}_{k+1, k, k-1}^{[\mathrm{R}]}\left(\bm{\beta}\right)^\top \bm{\Omega}_h^{\mathrm{[RR]}}(\bm{\theta})^{-1} \overline{\mathbf{Z}}_{k+1, k, k-1}^{[\mathrm{R}]}\left(\bm{\beta}\right) + 2 \log \left|\det D_{\mathbf{v}} {\bm{f}}_{h/2}\left(\mathbf{X}_{t_k}, \Delta_h \mathbf{X}_{t_{k+1}}; \bm{\beta}\right)\right|\right), \notag\\
    \mathcal{L}^{\mathrm{[PS\mid R]}}\left(\mathbf{X}_{0:t_N}; \bm{\theta}\right)  &\coloneqq 2(N-2) \log\det \bm{\Omega}_h^{\mathrm{[S\mid R]}}(\bm{\theta}) \label{eq:L_PSR}\\
    &\hspace{-7ex}+\sum_{k=1}^{N-1}  \left(\overline{\mathbf{Z}}_{k+1, k, k-1}^{\mathrm{[S\mid R]}}(\bm{\beta})^\top \bm{\Omega}_h^{\mathrm{[S\mid R]}}(\bm{\theta})^{-1}  \overline{\mathbf{Z}}_{k+1, k, k-1}^{\mathrm{[S\mid R]}}(\bm{\beta}) + 4 \log |\det D_\mathbf{v}\bm{f}_{h/2}(\mathbf{X}_{t_k}, \Delta_h \mathbf{X}_{t_{k+1}}; \bm{\beta})|\right).\notag
\end{align}

\begin{remark}
Due to the correction factors in the objective functions, we now have that \begin{equation}
    \mathcal{L}^{\mathrm{[PF]}}(\mathbf{X}_{0:t_N}; \bm{\theta}) \neq \mathcal{L}^{\mathrm{[PR]}}(\mathbf{X}_{0:t_N}; \bm{\theta}) + \mathcal{L}^{\mathrm{[PS\mid R]}}(\mathbf{X}_{0:t_N}; \bm{\theta}). \label{eq:sum_LP}
\end{equation}
However, when expanding the objective functions \eqref{eq:L_PF}-\eqref{eq:L_PSR} using Taylor series to the lowest relevant order in $h$, their approximations will be equal in \eqref{eq:sum_LP}, as shown in Section \ref{sec:AuxiliaryProperties}.
\end{remark}

\begin{remark} \label{rmrk:Bias_L_PSF}
    As shown in Section \ref{sec:AuxiliaryProperties}, the absence of the usual quadratic term involving $\mathbf{F}(\mathbf{Y}_{t_{k-1}}) - \mathbf{F}_0(\mathbf{Y}_{t_{k-1}})$ in the first order expansion of \eqref{eq:L_PSR} implies that $\mathcal{L}^{\mathrm{[PS\mid R]}}$ is not suitable for estimating the drift parameter. Furthermore, there is a term in the first order expansion of \eqref{eq:L_PSR} that does not vanish, needing an additional correction term of $2h\sum_{k=1}^{N-1} \tr D_\mathbf{v} \mathbf{N}(\mathbf{Y}_{t_k}; \bm{\beta})$ for consistency. This correction is represented as $4\log |\det D_\mathbf{v}\bm{f}_{h/2}|$ in \eqref{eq:L_PSR}. Notably, this term is absent in the complete objective function \eqref{eq:L_CSR}, making this adjustment somewhat artificial and could potentially deviate further from the true negative log-likelihood. Consequently, the objective function based on the full pseudo-likelihood \eqref{eq:L_PF} inherits this characteristic from \eqref{eq:L_PSR}, suggesting that in the partial observation scenario, using only the rough objective function \eqref{eq:L_PR} may be more appropriate.
\end{remark}

The estimators based on the partial sample are then defined as
\begin{equation*}
    \hat{\bm{\theta}}_N^\mathrm{[obj]}  \coloneqq \argmin_{\bm{\theta}}  \mathcal{L}^\mathrm{[obj]}  \left(\mathbf{X}_{0:t_N}; \bm{\theta}\right), \quad \mathrm{obj} \in\{\mathrm{[PF]}, \mathrm{[PR]}\}. 
\end{equation*}
In the partial observation case, the asymptotic variances of the diffusion estimators are identical whether using \eqref{eq:L_PF} or \eqref{eq:L_PR}, in contrast to the complete observation scenario. This variance is shown to be $9/4$ times higher than the variance of $\hat{\bm{\theta}}_N^{\mathrm{[CF]}}$, and $9/8$ times higher than that of $\hat{\bm{\theta}}_N^{\mathrm{[CR]}}$.

The numerical study in Section \ref{sec:Simulations} shows that the estimator based on the rough objective function \eqref{eq:L_PR} is less biased than the one based on the full objective function \eqref{eq:L_PF} in finite sample scenarios with partial observations. A potential reason for this is discussed in Remark \ref{rmrk:Bias_L_PSF}. Therefore, we recommend using the objective function \eqref{eq:L_PR} for partial observations.

\section{Main results} \label{sec:EstimatiorProperties}

This section states the two main results -- consistency and asymptotic normality of all four proposed estimators. The proofs are presented in Appendix \ref{sec:Appendix}.

First, we state the consistency of the estimators in both the complete and partial observation cases. Let $\mathcal{L}^\mathrm{[obj]}$ be one of the objective functions \eqref{eq:L_CF}, \eqref{eq:L_CR}, \eqref{eq:L_PF} or \eqref{eq:L_PR} and $\widehat{\bm{\theta}}_N^\mathrm{[obj]}$ the corresponding estimator. Thus,
\begin{equation*}
    \mathrm{obj} \in\{\mathrm{[CF]}, \mathrm{[CR]}, \mathrm{[PF]}, \mathrm{[PR]}\}.
\end{equation*} 

\begin{theorem}[Consistency of the estimators] \label{thm:Consistency}
Assume \ref{as:NLip}-\ref{as:fhInv}, $h \to 0$, and $Nh \to \infty$. Then, under the complete or partial observation setting, it holds
\begin{align*}
    &\widehat{\bm{\beta}}_N^\mathrm{[obj]} \xrightarrow[]{\mathbb{P}_{\bm{\theta}_0}} \bm{\beta}_0, && \widehat{\bm{\Sigma}\bm{\Sigma}}_N^\mathrm{[obj]}  \xrightarrow[]{\mathbb{P}_{\bm{\theta}_0}} \bm{\Sigma}\bm{\Sigma}^\top_0.
\end{align*}
\end{theorem}

\begin{remark}
We split the full objective function \eqref{eq:L_CF} into the sum of the rough objective function \eqref{eq:L_CR} and the smooth-given-rough objective function \eqref{eq:L_CSR}. Even if \eqref{eq:L_CSR} cannot identify the drift parameter $\bm{\beta}$, it is an important intermediate step in understanding the full objective function \eqref{eq:L_CF}. This can be seen in the proof of Theorem \ref{thm:Consistency}, where we first establish consistency of the diffusion estimator with a convergence rate of $\sqrt{N}$, which is faster than $\sqrt{N h}$, the convergence rate of the drift estimators. Then, under complete observations, we show that
\begin{align}
    \frac{1}{N h} (\mathcal{L}^{\mathrm{[CR]}}(\bm{\beta}, \bm{\sigma}_0)  - \mathcal{L}^{\mathrm{[CR]}}(\bm{\beta}_0, \bm{\sigma}_0)) \xrightarrow[\substack{Nh \to \infty\\ h \to 0}]{\mathbb{P}_{\bm{\theta}_0}} \int(\mathbf{F}_0(\mathbf{y}) - \mathbf{F}(\mathbf{y}))^\top (\bm{\Sigma}\bm{\Sigma}^\top)^{-1} (\mathbf{F}_0(\mathbf{y}) - \mathbf{F}(\mathbf{y})) \dif \nu_0(\mathbf{y}). \label{eq:drift_CR_cons_main}
\end{align}
The right-hand side of \eqref{eq:drift_CR_cons_main} is non-negative, with a unique zero for $\mathbf{F} = \mathbf{F}_0$. Conversely, for objective function \eqref{eq:L_CSR}, it holds
\begin{align}
     \frac{1}{N h} (\mathcal{L}^{\mathrm{[CS\mid R]}}(\bm{\beta}, \bm{\sigma})  - \mathcal{L}^{\mathrm{[CS\mid R]}}(\bm{\beta}_0, \bm{\sigma})) \xrightarrow[\substack{Nh \to \infty\\ h \to 0}]{\mathbb{P}_{\bm{\theta}_0}} 0,  \label{eq:drift_CSR_cons_main}
\end{align}
for all $\bm{\beta} \in \overline{\Theta}_{\beta}$. Hence, the limit in \eqref{eq:drift_CSR_cons_main} does not converge to a function with a unique zero as in \eqref{eq:drift_CR_cons_main}, making the drift parameter unidentifiable. Similar conclusions are drawn in the partial observation case. 
\end{remark}

Now, we state the asymptotic normality of the estimator. First, we need some preliminaries. Let $\rho >0$ and $\mathcal{B}_\rho\left(\bm{\theta}_0\right) = \{\bm{\theta} \in \Theta \mid \left\|\bm{\theta}-\bm{\theta}_0\right\| \leq \rho\}$ be a ball around $\bm{\theta}_0$. Since $\bm{\theta}_0 \in \Theta$, for sufficiently small $\rho >0$, $\mathcal{B}_\rho(\bm{\theta}_0) \in \Theta$. For $\hat{\bm{\theta}}_N^\mathrm{[obj]} \subseteq \mathcal{B}_\rho\left(\bm{\theta}_0\right)$, the mean value theorem yields
\begin{equation}
    \left(\int_0^1 \mathbb{H}_{\mathcal{L}^\mathrm{[obj]}}(\bm{\theta}_0 + t (\hat{\bm{\theta}}_N^\mathrm{[obj]} - \bm{\theta}_0))\dif t\right) (\hat{\bm{\theta}}_N^\mathrm{[obj]} - \bm{\theta}_0) = - \nabla_{\bm{\theta}} \mathcal{L}^\mathrm{[obj]}\left(\bm{\theta}_0\right). \label{eq:AssymptoticNormalityDecomp}
\end{equation}
Define
\begin{align}
    \mathbf{C}_N^\mathrm{[obj]}(\bm{\theta}) \coloneqq 
    \begin{bmatrix} \vspace{1ex}
    \left[ \frac{1}{Nh}\partial_{\beta^{(i_1)}\beta^{(i_2)}}^2 \mathcal{L}^\mathrm{[obj]}(\bm{\theta})\right]_{i_1,i_2 =1}^r & \left[\frac{1}{N\sqrt{h}}\partial_{\beta^{(i)}\sigma^{(j)}}^2 \mathcal{L}^\mathrm{[obj]}(\bm{\theta})\right]_{i=1,j=1}^{r,s}\\ 
    \left[\frac{1}{N\sqrt{h}}\partial_{\sigma^{(j)}\beta^{(i)}}^2 \mathcal{L}^\mathrm{[obj]}(\bm{\theta})\right]_{i=1,j=1}^{r,s} & \left[\frac{1}{N} \partial_{\sigma^{(j_1)} \sigma^{(j_2)}}^2 \mathcal{L}^\mathrm{[obj]} (\bm{\theta})\right]_{j_1, j_2 = 1}^s
    \end{bmatrix}, \label{eq:CN}
\end{align}
\begin{align}
    \mathbf{s}_N^\mathrm{[obj]} \coloneqq \begin{bmatrix}
    \sqrt{Nh} (\hat{\bm{\beta}}_N^\mathrm{[obj]} - \bm{\beta}_0) \vspace{1ex}\\
    \sqrt{N} (\hat{\bm{\sigma}}_N^\mathrm{[obj]} - \bm{\sigma}_0)
    \end{bmatrix},  \qquad
    \bm{\lambda}_N^\mathrm{[obj]} \coloneqq \begin{bmatrix}
    -\dfrac{1}{\sqrt{Nh}} \nabla_{\bm{\beta}} \mathcal{L}^\mathrm{[obj]}(\bm{\theta}_0)\\
    -\dfrac{1}{\sqrt{N}}  \nabla_{\bm{\sigma}} \mathcal{L}^\mathrm{[obj]}(\bm{\theta}_0)
    \end{bmatrix}, \label{eq:sNLN}
\end{align}
and $\mathbf{D}_N^\mathrm{[obj]} \coloneqq \int_0^1 \mathbf{C}_N^\mathrm{[obj]}(\bm{\theta}_0 + t (\hat{\bm{\theta}}_N^\mathrm{[obj]} - \bm{\theta}_0)) \dif t$. Then, \eqref{eq:AssymptoticNormalityDecomp} is equivalent to $\mathbf{D}_N^\mathrm{[obj]} \mathbf{s}_N^\mathrm{[obj]} = \bm{\lambda}_N^\mathrm{[obj]}$. Let
\begin{align}
    &[\mathbf{C}_{\bm{\beta}}(\bm{\theta}_0)]_{i_1,i_2} \coloneqq \int (\partial_{\beta^{(i_1)}} \mathbf{F}_0(\mathbf{y}))^\top (\bm{\Sigma}\bm{\Sigma}^\top_0)^{-1}(\partial_{\beta^{(i_2)}} \mathbf{F}_0(\mathbf{y}))\dif \nu_0(\mathbf{y}), \ 1\leq i_1,i_2 \leq r, \label{eq:Cb}\\
    &[\mathbf{C}_{\bm{\sigma}}(\bm{\theta}_0)]_{j_1,j_2} \coloneqq \tr((\partial_{\sigma^{(j_1)}} \bm{\Sigma}\bm{\Sigma}_0^\top)(\bm{\Sigma}\bm{\Sigma}_0^\top)^{-1}(\partial_{\sigma^{(j_2)}} \bm{\Sigma}\bm{\Sigma}_0^\top)(\bm{\Sigma}\bm{\Sigma}_0^\top)^{-1}), \ 1\leq j_1,j_2 \leq s. \label{eq:Cs}
\end{align}
\begin{theorem} \label{thm:AsymtoticNormality}
    Let assumptions \ref{as:NLip}-\ref{as:fhInv} hold and $\mathrm{obj} \in\{\mathrm{[CF]}, \mathrm{[CR]}, \mathrm{[PF]}, \mathrm{[PR]}\}$. If $h \to 0$, $Nh \to \infty$, and $Nh^2 \to 0$, then 
    \begin{align*}
    \begin{bmatrix}
        \sqrt{Nh} (\hat{\bm{\beta}}_N^{\mathrm{[obj]}} - \bm{\beta}_0)\\
        \sqrt{N} (\hat{\bm{\sigma}}_N^{\mathrm{[obj]}} - \bm{\sigma}_0)
    \end{bmatrix} & \xrightarrow[]{d} \mathcal{N}\left(\bm{0}, \begin{bmatrix}
        \mathbf{C}_{\bm{\beta}}(\bm{\theta}_0)^{-1} & \bm{0}_{r\times s}\\
        \bm{0}_{s\times r} & c^{\mathrm{[obj]}}\mathbf{C}_{\bm{\sigma}}(\bm{\theta}_0)^{-1}
        \end{bmatrix}\right), 
\end{align*}
under $\mathbb{P}_{\bm{\theta}_0}$, where
\begin{equation*}
    c^{\mathrm{[CF]}} = 1, \qquad c^{\mathrm{[CR]}} = 2, \qquad c^{\mathrm{[PF]}} = \frac{9}{4},   \qquad c^{\mathrm{[PR]}} = \frac{9}{4}.
\end{equation*}
\end{theorem}

\section{Simulation study} \label{sec:Simulations}

This Section presents the simulation study of the Kramers oscillator \eqref{eq:KramersSDE}, demonstrating the theoretical aspects and comparing the proposed Strang splitting estimators (denoted by SS) against estimators based on the EM and LL approximations. We chose to compare our proposed estimators to these two because the EM estimator is routinely used in applications, and the LL estimator has shown to be one of the best discrete MLE methods for the elliptic case, see, for example, \citep{ HurnJeismanAndLindsay, López-Pérez2021, Pilipovic2024}. The true parameters are set to $\eta_0 = 6.5, a_0 = 1, b_0 = 0.6$ and $\sigma_0^2 =0.1$. We outline the estimators specifically designed for the Kramers oscillator, explain the simulation procedure, describe the optimization implemented in the \texttt{R} programming language \citep{R}, and then present and interpret the results.

\subsection{Estimators used in the study}

For the Kramers oscillator \eqref{eq:KramersSDE}, the EM approximated transition distribution is
\begin{align*}
    \begin{bmatrix}
        X_{t_k}\\
        V_{t_k} \\
        \end{bmatrix} \mid
        \begin{bmatrix}
        X_{t_{k-1}}\\
        V_{t_{k-1}}\\
        \end{bmatrix} = \begin{bmatrix}
        x\\
        v\\
    \end{bmatrix} &\sim \mathcal{N}\left(
    \begin{bmatrix}
        x + h v\\
        v + h \left (-\eta v + a x - b x^3 \right)\\
    \end{bmatrix}, 
    \begin{bmatrix}
       0 & 0 \\
        0 & h \sigma^2 \\
    \end{bmatrix}\right). 
\end{align*}
The ill-conditioned variance of this discretization restricts us to an estimator based only on the rough objective function. The estimator for complete observations directly follows from the Gaussian distribution. The estimator for partial observations is defined as \citep{SamsonThieullen2012}
\begin{align*}
\widehat{\bm{\theta}}_\mathrm{EM}^\mathrm{[PR]} = \argmin_{\bm{\theta}} \left\{\frac{2}{3}(N-4) \log \sigma^2 + \frac{1}{h \sigma^2}\sum_{k=2}^{N-2} (\Delta_h X_{t_{k+1}} - \Delta_h X_{t_k} - h  (-\eta \Delta_h X_{t_{k-1}} + a X_{t_{k-1}} - b X_{t_{k-1}}^3 ))^2\right\}.
\end{align*}
Note that the EM estimator as defined in \cite{SamsonThieullen2012} uses four consecutive points $X_{t_{k-2}}, X_{t_{k-1}}, X_{t_k}$ and $X_{t_{k+1}}$. This comes from the shift in indices originally introduced to avoid correlation between $\Delta_h X_{t_{k+1}} - \Delta_h X_{t_k}$ and a function of $(X_{t_{k-1}}, \Delta_h X_{t_k})$. We keep this definition in the simulation study but do not apply the shift of indices to other estimators.

To address the issue of the ill-conditioned EM approximation, \citep{gloter2020} introduced the local Gaussian (LG) approximation by adding an additional order of approximation in the mean of the smooth comportment and the smallest necessary orders in each of the four sub-components of the covariance matrix. For the Kramers oscillator \eqref{eq:KramersSDE}, LG yields the following conditional distribution
\begin{align*}
    \begin{bmatrix}
        X_{t_k}\\
        V_{t_k} \\
        \end{bmatrix} \mid
        \begin{bmatrix}
        X_{t_{k-1}}\\
        V_{t_{k-1}}\\
        \end{bmatrix} = \begin{bmatrix}
        x\\
        v\\
    \end{bmatrix} &\sim \mathcal{N}\left(
    \begin{bmatrix}
        x + h v + \frac{h^2}{2}\left(-\eta v + a x - b x^3 \right)\\
        v + h \left(-\eta v + a x - b x^3 \right)\\
    \end{bmatrix}, 
    \begin{bmatrix}
       \frac{h^3}{3} \sigma^2  & \frac{h^2}{2} \sigma^2  \\
        \frac{h^2}{2} \sigma^2 & h \sigma^2 \\
    \end{bmatrix}\right). 
\end{align*}
For the application of LG to second-order SDE, see Section 4.1.4. in \citep{iguchi2023.1}. Note that using only the rough part to form the objective function based on the LG matches using the rough part of the EM approximation.

To our knowledge, the LL estimator has not previously been applied to partial observations. Given the similar theoretical and computational performance of the SS and LL discretizations, we suggest (without formal proof) adjusting the LL objective functions with the same correction factors as used in the SS estimator. The numerical evidence indicates that the LL estimator has the same asymptotic properties as those proved for the SS estimator. We omit the definition of the LL estimator due to its complexity (see \cite{melnykova2020parametric, Pilipovic2024} and accompanying code). However, note that the following covariance matrix has to be computed at each observation to implement the LL estimator
\begin{equation*}
    \widetilde{\bm{\Omega}}_{h, k}^\mathrm{[LL]} \coloneqq \int_0^h e^{ D\widetilde{\mathbf{F}}(\mathbf{Y}_{t_{k-1}})(h-u)} \widetilde{\bm{\Sigma}} \widetilde{\bm{\Sigma}}^\top e^{ D\widetilde{\mathbf{F}}(\mathbf{Y}_{t_{k-1}})^\top (h-u)} \dif u.
\end{equation*}
Matrix $\widetilde{\bm{\Omega}}_{h, k}^\mathrm{[LL]}$ is similar to $\widetilde{\bm{\Omega}}_{h}$ in eq. \eqref{eq:Omegah_definition}, with the key difference being that $\widetilde{\bm{\Omega}}_{h}$ does not depend on data. Although the same code is used to compute both $\widetilde{\bm{\Omega}}_{h}$ and $\widetilde{\bm{\Omega}}_{h, k}^\mathrm{[LL]}$, the computation is considerably more time-consuming for $\widetilde{\bm{\Omega}}_{h, k}^\mathrm{[LL]}$, as shown in the following results. 

To define SS estimators based on the Strang splitting scheme, we first split SDE \eqref{eq:KramersSDE} as follows
\begin{align*}
    \dif \begin{bmatrix}
        X_{t}\\
        V_{t} \\
        \end{bmatrix} =
        \underbrace{\begin{bmatrix}
        0 & 1\\
        -2a & - \eta\\
        \end{bmatrix}}_{\widetilde{\mathbf{A}}} \Bigg(\begin{bmatrix}
        X_t\\
        V_t\\
    \end{bmatrix} - \underbrace{\begin{bmatrix}
        x_\pm^\star\\
        0\\
    \end{bmatrix}}_{\widetilde{\mathbf{b}}}\Bigg)\dif t + 
    \underbrace{\begin{bmatrix}
        0\\
        a X_t - b X_t^3 + 2a (X_t - x_\pm^\star)\\
    \end{bmatrix}}_{\widetilde{\mathbf{N}}(X_t,V_t)}\dif t + \begin{bmatrix}
        0\\
        \sigma\\
    \end{bmatrix}\dif W_t, 
\end{align*}
where $\mathbf{y}_\pm^\star = (x_\pm^\star, v^\star) = (\pm \sqrt{a/b}, 0)$ are the two stable points of the dynamics. Recall that in Section \ref{sec:StrangSplitting}, we suggested splitting SDE such that $\widetilde{\mathbf{A}} = D\widetilde{\mathbf{F}}(y_\pm^\star)$ and $\mathbf{b} = \mathbf{y}_\pm^\star$. Since there are two stable points, we suggest splitting with $x_+^\star$ when $X_t > 0$, and $x_-^\star$ when $X_t < 0$. This splitting follows the guidelines from \cite{Pilipovic2024}. Note that the nonlinear ODE driven by $\widetilde{\mathbf{N}}(x,v)$ has a trivial solution where $x$ is a constant. To obtain SS estimators, we plug in the corresponding components in the objective functions \eqref{eq:L_CF}, \eqref{eq:L_CR}, \eqref{eq:L_PF} and \eqref{eq:L_PR}.

\subsection{Simulation of trajectories}

We simulate a sample path using the EM discretization with a step size of $h^{\mathrm{sim}} = 0.0001$ to ensure good performance. We sub-sample from the path to decrease discretization errors and use a time step $h = 0.1$. The path has $N = 5000$ data points. We repeat the simulations to obtain 500 data sets.

\subsection{Optimization in \texttt{R}}

For optimizing the objective functions, we proceed as in \cite{Pilipovic2024} using the \texttt{R} package \texttt{torch} \citep{Torch}, which allows automatic differentiation. The optimization employs the resilient backpropagation algorithm, \texttt{optim\_rprop}. We use the default hyperparameters and limit the number of optimization iterations to 2000. The convergence criterion is set to a precision of $10^{-5}$ for the difference between estimators in consecutive iterations. The initial parameter values are $(-0.1, -0.1, 0.1, 0.1)$.

\begin{figure}
    \centering
    \includegraphics[width = \textwidth]{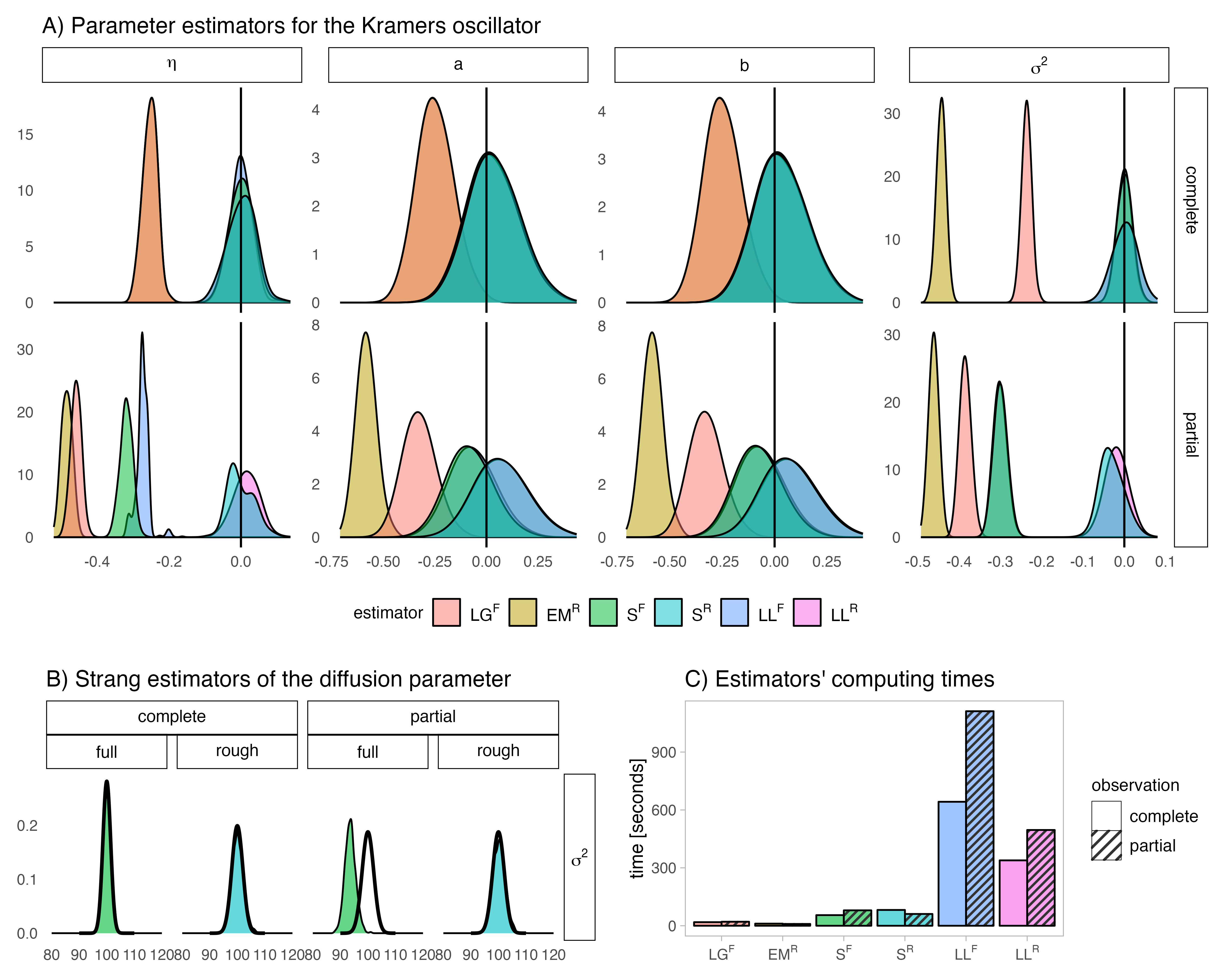}
    \caption[\bf Parameter estimates in a simulation study for the Kramers oscillator, eq. \eqref{eq:KramersSDE}]{
    {\bf Parameter estimates in a simulation study for the Kramers oscillator, eq. \eqref{eq:KramersSDE}}. The color code is the same across all three figures.
    \textbf{A)} Normalized distributions of parameter estimation errors $(\hat{\bm{\theta}}_N - \bm{\theta}_0) \oslash \bm{\theta}_0$ (where $\oslash$ is the element-wise division) in the complete and partial observation cases, based on 500 simulated data sets with $h = 0.1$ and $N = 5000$. Each column corresponds to a different parameter, while the color indicates the type of estimator. Estimators are distinguished by superscripted objective functions ($\mathrm{F}$ for full and $\mathrm{R}$ for rough).
    \textbf{B)} Distribution of $\widehat{\sigma}_N^2$ estimators based on 1000 simulations with $h = 0.02$ and $N = 5000$ across different observation settings (complete or partial) and objective function choices (full or rough) using the Strang splitting scheme. The true value of $\sigma^2$ is set to $\sigma_0^2 = 100$. Theoretical normal densities are overlaid for comparison. Theoretical variances are calculated based on $C_{\sigma^2}(\bm{\theta}_0)$, eq. \eqref{eq:Csigma_Kramers}.
    \textbf{C)} Median computing time in seconds for one estimation of various estimators based on 500 simulations with $h = 0.1$ and $N = 5000$. Shaded color patterns represent times in the partial observation case, while no color pattern indicates times in the complete observation case.}
    \label{fig:final_plot}
\end{figure}

\subsection{Results}

The simulation study results are presented in Figure \ref{fig:final_plot}. Figure \ref{fig:final_plot}A) presents the distributions of the normalized estimators $(\hat{\bm{\theta}}_N - \bm{\theta}_0) \oslash \bm{\theta}_0$ (where $\oslash$ is the element-wise division) in the complete and partial observation cases. 

Under complete observations, the SS and LL estimators perform nearly identically (the turquoise color in the figure comes from overlapping). The variances of the SS and LL estimators of $\sigma^2$ based on the rough objective function are higher than those derived from the full objective function, aligning with theoretical expectations that adding the smooth part in the objective function only influences the variance estimation. Moreover, the EM and LG estimators for the drift parameters also overlap, making the color orange, and they both display notable bias. This aligns with the theoretical findings. 

Interestingly, in the partial observation scenario, Figure \ref{fig:final_plot}A) reveals that estimators employing the full objective function display greater finite sample bias than those based on the rough objective function. Possible reasons for the bias of the SS estimators are discussed in Remark \ref{rmrk:Bias_L_PSF}. However, this bias is considerably reduced for smaller time steps (not shown), thus confirming the theoretical asymptotic results. This observation suggests that the rough objective function is preferable under partial observations due to its lower bias. Backward finite difference approximations of the velocity variables perform similarly to the forward differences and are therefore excluded from the figure for clarity.

We closely examine the variances of the SS estimators of $\sigma^2$ in Figure \ref{fig:final_plot}B). The LL estimators are omitted due to their similarity to the SS estimators and because the computation times for the LL estimators are prohibitive. We opt for $h = 0.02$ and conduct 1000 simulations to align more closely with the asymptotic predictions. Additionally, we set $\sigma_0^2 = 100$ to test different noise levels. On top of the empirical distributions, we overlay theoretical normal densities that match the variances in Theorem \ref{thm:AsymtoticNormality}. The theoretical variance is derived from $C_{\sigma^2}(\bm{\theta}_0)$ in \eqref{eq:Cs}, which for the Kramers oscillator in \eqref{eq:KramersSDE} is
\begin{align}
    C_{\sigma^2}(\bm{\theta}_0) = \frac{1}{\sigma_0^4}. \label{eq:Csigma_Kramers}
\end{align}
Figure \ref{fig:final_plot}B) illustrates that the lowest variance of the diffusion estimator is obtained using the full objective function with complete observations. The second lowest variance is achieved using the rough objective function with complete observations. The largest variance is observed in the partial observation case; however, it remains independent of whether the full or rough objective function is used. Once again, we observe that using the full objective function introduces finite sample bias.

In Figure \ref{fig:final_plot}C), we compare running times calculated using the \texttt{tictoc} package in \texttt{R}. Running times are measured from the start of the optimization step until convergence. The figure depicts the median over 500 repetitions to mitigate the influence of outliers. The EM method is the fastest and is closely followed by the LG and SS estimators. The LL estimators are 10-100 times slower than the SS estimators, depending on whether complete or partial observations are used and whether the full or rough objective function is employed. 

We conclude from the simulation study that the SS estimators perform better than the alternatives when considering both accuracy and speed. Additionally, using the full objective function is preferable under complete observation, whereas in the partial observation regime, it is more effective to base the objective function on the rough components.

\section{Application to Greenland Ice Core Data} \label{sec:Greenland}

The abrupt temperature changes during the ice ages, known as the Dansgaard–Oeschger (DO) events, are essential for understanding the climate \citep{Dansgaard1993}.  These events occurred during the last glacial era spanning approximately the period from 115,000 to 12,000 years before the present and are characterized by rapid warming phases followed by gradual cooling periods, revealing colder (stadial) and warmer (interstadial) climate states \citep{RASMUSSEN201414}. 

Proxy data from Greenland ice cores, particularly stable water isotope composition ($\delta^{18}\text{O}$) and calcium ion concentrations ($\text{Ca}^{2+}$), offer valuable insights into these past climate variations \citep{Boersetal2017, Boersetal2018, Boers2018, ditlevsen2002fast, LohmannDitlevsen, Hassanibesheli2020Reconstructing}.

The $\delta^{18} \text{O}$ ratio, reflecting the relative abundance of $^{18}\text{O}$ and $^{16}\text{O}$ isotopes in the ice, serves as a proxy for paleotemperatures during snow deposition. Conversely, calcium ions, originating from dust deposition, exhibit a strong negative correlation with $\delta^{18} \text{O}$, with higher calcium ion levels indicating colder conditions. Here, we use $\text{Ca}^{2+}$ time series due to its finer temporal resolution.

The data are average measurements within ice segments, with each segment corresponding to a time interval that increases as a function of depth. Since the data represent integrals of the underlying process over time intervals, we should employ a second-order SDE, as discussed in \cite{ditlevsen2002fast, DitlevsenSorensen2004}.

\begin{figure}
    \centering
    \includegraphics[width = \textwidth]{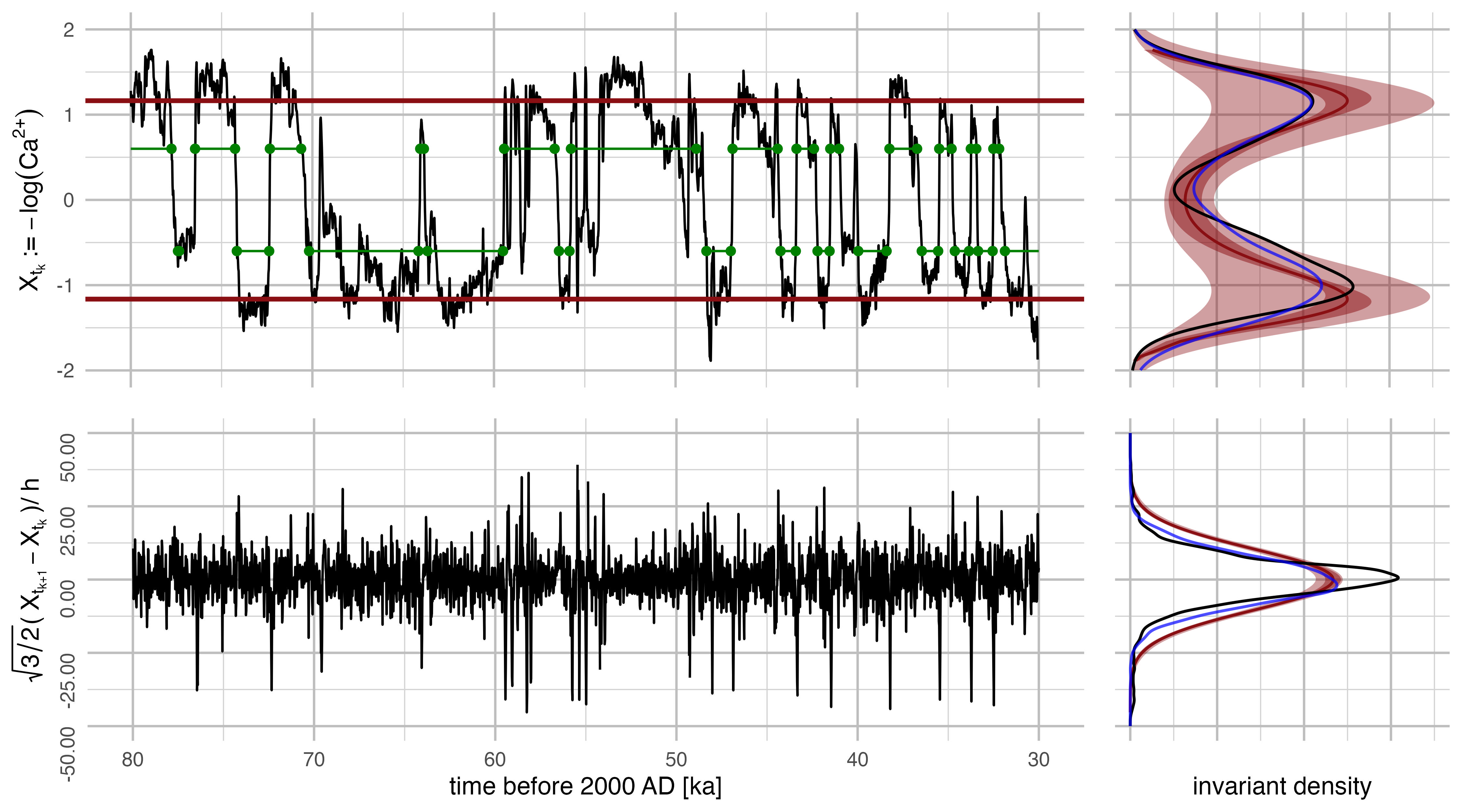}
    \caption{\textbf{Ice core data from Greenland.} \textbf{Left:} Trajectories over time (in 1000 years) of the centered negative logarithm of the $\text{Ca}^{2+}$ measurements (top) and forward difference approximations of its rate of change (bottom). Forward differences are multiplied by $\sqrt{3/2}$ to compensate for underestimating the variance. The two vertical dark red lines represent the estimated stable equilibria of the double-well potential function. Green points denote up- and down-crossings of level $\pm 0.6$, conditioned on having crossed the other level. Green vertical lines indicate empirical occupancy estimates in either of the two metastable states. \textbf{Right:} Empirical densities (black) and estimated invariant densities with confidence intervals (dark red), prediction intervals (light red), and the empirical density of a simulated sample from the estimated model (blue). The empirical density of forward differences is also rescaled by $\sqrt{3/2}$ to match the variance of the theoretical density.}
    \label{fig:ice_core}
\end{figure}

The DO events are abrupt transitions from colder climates (stadials) to approximately 10 degrees warmer climates (interstadials) within a few decades. Although the waiting times between state switches last a couple of thousand years, their spacing exhibits significant variability. The underlying mechanisms driving these changes remain largely elusive, prompting discussions on whether they follow cyclic patterns, result from external forcing, or emerge from noise-induced processes \citep{Boers2018, ditlevsen2007}. We aim to determine if noise-induced transitions of the Kramers oscillator can explain the observed data.

The measurements were conducted at the summit of the Greenland ice sheet as part of the Greenland Icecore Project (GRIP) \citep{GRIP, Andersen2004}. Initially, the data were sampled at 5 cm intervals, resulting in a non-equidistant time series due to ice compression at greater depths, where 5 cm of ice core spans longer time periods. For our analysis, we use a version of the data transformed into a uniformly spaced series through 20-year binning and averaging. This transformation simplifies the analysis and highlights significant climatic trends. Moreover, the binning and averaging suggest that we work with integrated diffusions and that we should model data using second-order SDEs. The dataset is available in the Appendix of \citep{RASMUSSEN201414, Seierstad2014}.

To address the large amplitudes and negative correlation with temperature, we transform the data to the negative logarithm of $\text{Ca}^{2+}$, where higher values of the transformed variable indicate warmer climates at the time of snow deposition. Additionally, we center the transformed measurements around zero. With the 20-year binning, to obtain one point per 20 years, we average across the bins, resulting in a time step of $h = 0.02 \mathrm{ka}$ (kiloanni; $1 \mathrm{ka} = 1000$ years). We imputed a few missing values using the \texttt{na.approx} function from the \texttt{zoo} package. Following  \cite{Hassanibesheli2020Reconstructing}, we analyze a subset of the data with a sufficiently good signal-to-noise ratio. \cite{Hassanibesheli2020Reconstructing} examined the data from $30$ to $60 \mathrm{ka}$ before the present. Here, we extend the analysis to cover $30 \mathrm{ka}$ to $80 \mathrm{ka}$, resulting in a time interval of $T = 50 \mathrm{ka}$ and a sample size of $N = 2500$. We approximate the velocity of the transformed $\text{Ca}^{2+}$ by the forward difference method. Figure \ref{fig:ice_core} illustrates the trajectories and invariant distributions.

We fit the Kramers oscillator to the $-\log \text{Ca}^{2+}$ time series and estimate parameters using the SS estimator. Following Theorem \ref{thm:AsymtoticNormality}, we compute $\mathbf{C}_{\bm{\beta}}(\bm{\theta}_0)$ from \eqref{eq:Cb}. Applying the invariant density $\pi_0(x, v)$ from \eqref{eq:XVinv}, which decouples into $\pi_0(x)$ from \eqref{eq:Xinv} and a Gaussian with zero mean and $\sigma_0^2/(2\eta_0)$ variance, yields 
\begin{align}
    \mathbf{C}_{\bm{\beta}}(\bm{\theta}_0) = \begin{bmatrix}
        \frac{1}{2\eta_0} & 0 & 0\\ \vspace{1ex}
        0 & \frac{1}{\sigma_0^2} \int_{-\infty}^\infty x^2 \pi_0(x) \dif x & -\frac{1}{\sigma_0^2} \int_{-\infty}^\infty x^4 \pi_0(x) \dif x\\
        0 & -\frac{1}{\sigma_0^2} \int_{-\infty}^\infty x^4 \pi_0(x) \dif x & \frac{1}{\sigma_0^2} \int_{-\infty}^\infty x^6 \pi_0(x) \dif x
    \end{bmatrix}. \label{eq:Cbeta_Kramers}
\end{align}
Thus, to obtain $95\%$ confidence intervals (CI) for the estimated parameters, we plug $\widehat{\bm{\theta}}_N$ into \eqref{eq:Csigma_Kramers} and \eqref{eq:Cbeta_Kramers}. Table \ref{tab:est} shows the estimators and confidence intervals. We also calculate the expected waiting time $\tau$, eq. \eqref{eq:tau} of crossing from one state to another, and its confidence interval using the Delta method. 
\begin{table}[ht]
\centering
\begin{tabular}{lccc}
Parameter & Estimate & 95\% CI \\ \hline
$\eta$   & 62.5 & $59.4 - 65.6$ \\
$a$      & 296.7 & $293.6 - 299.8$ \\
$b$      & 219.1 & $156.4 - 281.7$ \\
$\sigma^2$ & 9125 & $8589 - 9662$ \\
$\tau$   & 3.97 & $3.00 - 4.94$ \\ \hline
\end{tabular}
\caption{Estimated parameters of the Kramers oscillator from Greenland ice core data.}
\label{tab:est}
\end{table}

The model fit is assessed in the right panels of Figure \ref{fig:ice_core}. Here, we present the empirical invariant distributions of the two coordinates along with the fitted theoretical invariant distribution and a $95\%$ confidence interval. A prediction interval for the distribution is provided by simulating 1000 datasets from the fitted model, matching the size of the empirical data. We estimate the empirical distributions for each simulated dataset and construct a $95\%$ prediction interval using the pointwise 2.5th and 97.5th percentiles of these estimates. A single example trace is included in blue. To account for the underestimated variance of the forward differences $\Delta_h X$ (as shown in Figure \ref{fig:sim_V}), we scale the trajectory at the bottom of Figure \ref{fig:ice_core}, along with its corresponding empirical density, by $\sqrt{3/2}$.

While the fitted distribution for $-\log \text{Ca}^{2+}$ appears to fit well even with the symmetric Kramers model, the velocity variables are not adequately captured despite the scaling by $\sqrt{3/2}$. This discrepancy is likely caused by extreme values in the data that are not adequately modeled by additive Gaussian noise. Consequently, the model compensates by estimating a large variance. A more suitable noise structure might be obtained using a Student's $t$-distribution, which, however, implies non-additive noise.

We estimate the waiting time between metastable states to be approximately 4000 years. However, this approximation relies on certain assumptions, namely $62.5 \approx \eta \gg \sqrt{a} \approx 17.2$ and $73 \approx \sigma^2/2\eta \ll a^2/4b \approx 100$. Since these assumptions are only vaguely satisfied, the waiting time should be interpreted with caution.

Defining the current state of the process is not straightforward. One method involves identifying successive up- and down-crossings of predefined thresholds within the smoothed data. However, the estimated occupancy time in each state depends on the level of smoothing applied and the distance of crossing thresholds from zero. Using a smoothing technique involving running averages within windows of 11 data points (equivalent to 220 years) and detecting down- and up-crossings of levels $\pm 0.6$, we find an average occupancy time of 4058 years in stadial states and 3550 years in interstadial states. Nevertheless, the actual occupancy times exhibit significant variability, ranging from 60 to 6900 years, with the central $50\%$ of values falling between 665 and 2115 years. This classification of states is depicted in green in Figure \ref{fig:ice_core}. Overall, the estimated mean occupancy time inferred from the Kramers oscillator appears reasonable.

\section{Technical results} \label{sec:AuxiliaryProperties}

In this Section, we present all the technical properties used to derive the main results. 

We start by expanding $\widetilde{\bm{\Omega}}_h$ and its block components $\bm{\Omega}_h^{\mathrm{[RR]}}(\bm{\theta})^{-1}, \bm{\Omega}_h^{\mathrm{[S\mid R]}}(\bm{\theta})^{-1}$, $\log \det \bm{\Omega}_h^{\mathrm{[RR]}}(\bm{\theta}), \log \det \bm{\Omega}_h^{\mathrm{[S\mid R]}}(\bm{\theta})$ and $\log |\det D {\bm{f}}_{h/2}\left(\mathbf{y}; \bm{\beta}\right)|$ when $h$ goes to zero. Then, we expand $\widetilde{\mathbf{Z}}_{k, k-1}(\bm{\beta})$ and $\widetilde{\overline{\mathbf{Z}}}_{k+1, k, k-1}(\bm{\beta})$ around $\mathbf{Y}_{t_{k-1}}$ when $h$ goes to zero. The main tools used are It\^o's lemma, Taylor expansions, and Fubini's theorem. The final result is stated in Propositions \ref{prop:ZtkandZtkbar} and \ref{prop:ZtkS}. The approximations depend on the drift function $\mathbf{F}$, the nonlinear part $\mathbf{N}$, and some correlated sequences of Gaussian random variables. Finally, we obtain approximations of the objective functions \eqref{eq:L_CF}, \eqref{eq:L_CR},  \eqref{eq:L_CSR} and \eqref{eq:L_PF} - \eqref{eq:L_PSR}. Proofs of all the stated propositions and lemmas in this section are in Appendix \ref{sec:Appendix}.

\subsection[Covariance matrix]{Covariance matrix $\widetilde{\bm{\Omega}}_h$}

The covariance matrix $\widetilde{\bm{\Omega}}_h$ is approximated by
\begin{align}
    \widetilde{\bm{\Omega}}_h &= \int_0^h e^{ \widetilde{\mathbf{A}}(h-u)} \widetilde{\bm{\Sigma}} \widetilde{\bm{\Sigma}}^\top e^{ \widetilde{\mathbf{A}}^\top (h-u)} \dif u\notag \\
    &= h \widetilde{\bm{\Sigma}} \widetilde{\bm{\Sigma}}^\top + \frac{h^2}{2} (\widetilde{\mathbf{A}} \widetilde{\bm{\Sigma}} \widetilde{\bm{\Sigma}}^\top + \widetilde{\bm{\Sigma}} \widetilde{\bm{\Sigma}}^\top \widetilde{\mathbf{A}}^\top ) + \frac{h^3}{6}(\widetilde{\mathbf{A}}^2 \widetilde{\bm{\Sigma}} \widetilde{\bm{\Sigma}}^\top  + 2 \widetilde{\mathbf{A}}\widetilde{\bm{\Sigma}} \widetilde{\bm{\Sigma}}^\top \widetilde{\mathbf{A}}^\top  + \widetilde{\bm{\Sigma}} \widetilde{\bm{\Sigma}}^\top  (\widetilde{\mathbf{A}}^2)^\top)\notag \\
    &+ \frac{h^4}{24}(\widetilde{\mathbf{A}}^3 \widetilde{\bm{\Sigma}} \widetilde{\bm{\Sigma}}^\top  + 3 \widetilde{\mathbf{A}}^2\widetilde{\bm{\Sigma}} \widetilde{\bm{\Sigma}}^\top \widetilde{\mathbf{A}}^\top   + 3 \widetilde{\mathbf{A}}\widetilde{\bm{\Sigma}} \widetilde{\bm{\Sigma}}^\top (\widetilde{\mathbf{A}}^2)^\top  + \widetilde{\bm{\Sigma}} \widetilde{\bm{\Sigma}}^\top  (\widetilde{\mathbf{A}}^3)^\top)+ \mathbf{R}(h^5, \mathbf{y}_0). \label{eq:Omegah}
\end{align}
The following lemma approximates each block of $\widetilde{\bm{\Omega}}_h$ up to the first two leading orders of $h$. The result follows directly from equations \eqref{eq:SigmaSigmaT}, \eqref{eq:ANtilde}, and \eqref{eq:Omegah}.
\begin{lemma} \label{lemmma:Omegah}
    The covariance matrix $\widetilde{\bm{\Omega}}_h$ defined in \eqref{eq:Omegah}-\eqref{eq:Omegahblock} approximates block-wise as
    \begin{align*}
        \bm{\Omega}_h^{\mathrm{[SS]}}(\bm{\theta})& = \frac{h^3}{3}\bm{\Sigma}\bm{\Sigma}^\top + \frac{h^4}{8}(\mathbf{A}_{\mathbf{v}}(\bm{\beta})\bm{\Sigma}\bm{\Sigma}^\top + \bm{\Sigma}\bm{\Sigma}^\top\mathbf{A}_{\mathbf{v}}(\bm{\beta})^\top) + \mathbf{R}(h^5, \mathbf{y}_0),\\
        \bm{\Omega}_h^{\mathrm{[SR]}}(\bm{\theta})& = \frac{h^2}{2}\bm{\Sigma}\bm{\Sigma}^\top + \frac{h^3}{6}(\mathbf{A}_{\mathbf{v}}(\bm{\beta})\bm{\Sigma}\bm{\Sigma}^\top + 2\bm{\Sigma}\bm{\Sigma}^\top\mathbf{A}_{\mathbf{v}}(\bm{\beta})^\top) + \mathbf{R}(h^4, \mathbf{y}_0),\\
        \bm{\Omega}_h^{\mathrm{[RS]}}(\bm{\theta})& = \frac{h^2}{2}\bm{\Sigma}\bm{\Sigma}^\top + \frac{h^3}{6}(2\mathbf{A}_{\mathbf{v}}(\bm{\beta})\bm{\Sigma}\bm{\Sigma}^\top + \bm{\Sigma}\bm{\Sigma}^\top\mathbf{A}_{\mathbf{v}}(\bm{\beta})^\top) + \mathbf{R}(h^4, \mathbf{y}_0),\\
        \bm{\Omega}_h^{\mathrm{[RR]}}(\bm{\theta})& = h\bm{\Sigma}\bm{\Sigma}^\top + \frac{h^2}{2}(\mathbf{A}_{\mathbf{v}}(\bm{\beta})\bm{\Sigma}\bm{\Sigma}^\top + \bm{\Sigma}\bm{\Sigma}^\top\mathbf{A}_{\mathbf{v}}(\bm{\beta})^\top) + \mathbf{R}(h^3, \mathbf{y}_0).
    \end{align*}
\end{lemma}
Building on Lemma \ref{lemmma:Omegah}, we calculate products, inverses, and logarithms of the components of $\widetilde{\bm{\Omega}}_h$ in the following lemma.
\begin{lemma} \label{lemmma:OmegahProp}
    For the covariance matrix $\widetilde{\bm{\Omega}}_h$ defined in \eqref{eq:Omegah} it holds:
    \begin{enumerate}
        \myitem{(i)} \label{item:OmegahProp1} $\bm{\Omega}_h^{\mathrm{[RR]}}(\bm{\theta})^{-1} =\frac{1}{h}(\bm{\Sigma}\bm{\Sigma}^\top)^{-1} - \frac{1}{2}((\bm{\Sigma}\bm{\Sigma}^\top)^{-1}\mathbf{A}_{\mathbf{v}}(\bm{\beta}) + \mathbf{A}_{\mathbf{v}}(\bm{\beta})^\top (\bm{\Sigma}\bm{\Sigma}^\top)^{-1})  + \mathbf{R}(h, \mathbf{y}_0)$;
        \myitem{(ii)} \label{item:OmegahProp2} $\bm{\Omega}_h^{\mathrm{[SR]}}(\bm{\theta})\bm{\Omega}_h^{\mathrm{[RR]}}(\bm{\theta})^{-1} = \frac{h}{2}\mathbf{I} - \frac{h^2}{12}(\mathbf{A}_{\mathbf{v}} - \bm{\Sigma}\bm{\Sigma}^\top\mathbf{A}_{\mathbf{v}}(\bm{\beta})^\top(\bm{\Sigma}\bm{\Sigma}^\top)^{-1}) + \mathbf{R}(h^3, \mathbf{y}_0)$;
        \myitem{(iii)} \label{item:OmegahProp3} $\bm{\Omega}_h^{\mathrm{[SR]}}(\bm{\theta})\bm{\Omega}_h^{\mathrm{[RR]}}(\bm{\theta})^{-1}\bm{\Omega}_h^{\mathrm{[RS]}}(\bm{\theta}) = \frac{h^3}{4}\bm{\Sigma}\bm{\Sigma}^\top + \frac{h^4}{8}(\mathbf{A}_{\mathbf{v}}(\bm{\beta})\bm{\Sigma}\bm{\Sigma}^\top + \bm{\Sigma}\bm{\Sigma}^\top\mathbf{A}_{\mathbf{v}}(\bm{\beta})^\top) + \mathbf{R}(h^5, \mathbf{y}_0)$;
        \myitem{(iv)} \label{item:OmegahProp4} $\bm{\Omega}_h^{\mathrm{[S\mid R]}}(\bm{\theta}) = \frac{h^3}{12}\bm{\Sigma}\bm{\Sigma}^\top + \mathbf{R}(h^5, \mathbf{y}_0)$;
        \myitem{(v)} \label{item:OmegahProp5} $\log \det \bm{\Omega}_h^{\mathrm{[RR]}}(\bm{\theta}) =d \log h +  \log \det \bm{\Sigma}\bm{\Sigma}^\top + h \tr \mathbf{A}_{\mathbf{v}}(\bm{\beta}) + {R}(h^2, \mathbf{y}_0)$;
        \myitem{(vi)} \label{item:OmegahProp6} $\log \det \bm{\Omega}_h^{\mathrm{[S\mid R]}}(\bm{\theta}) = 3d \log h + \log \det \bm{\Sigma}\bm{\Sigma}^\top + {R}(h^2, \mathbf{y}_0)$;
        \myitem{(vii)} \label{item:OmegahProp7} $\log \det \widetilde{\bm{\Omega}}_h(\bm{\theta}) = 4d \log h + 2\log \det \bm{\Sigma}\bm{\Sigma}^\top + h \tr \mathbf{A}_{\mathbf{v}}(\bm{\beta}) + {R}(h^2, \mathbf{y}_0)$. 
    \end{enumerate}
\end{lemma}
\begin{remark}
    We adjusted the objective functions for partial observations using the term $c \log \det \bm{\Omega}_{h/c}^{\mathrm{[\cdot]}}$, where $c$ is a correction constant. This adjustment keeps the term $h \tr \mathbf{A}_\mathbf{v}(\bm{\beta})$ in \ref{item:OmegahProp5}-\ref{item:OmegahProp7} constant, not affecting the asymptotic distribution of the drift parameter. There is no $h^4$-term in $\bm{\Omega}_h^{\mathrm{[S\mid R]}}(\bm{\theta})$ which simplifies the approximation of $\bm{\Omega}_h^{\mathrm{[S\mid R]}}(\bm{\theta})^{-1}$ and $\log \det \bm{\Omega}_h^{\mathrm{[S\mid R]}}(\bm{\theta})$. 
\end{remark}

\subsection[Nonlinear solution]{Nonlinear solution $\widetilde{\bm{f}}_h$}

We now state a useful proposition for the nonlinear solution $\widetilde{\bm{f}}_h$ (Section 1.8 in \citep{SolvingODEI}). 
\begin{proposition} \label{prop:fh} Let Assumptions \ref{as:NLip}, \ref{as:NPoly} and \ref{as:fhInv} hold. When $h \to 0$, the $h$-flow of  \eqref{eq:SplittingEq2} approximates as
\begin{align}
    \widetilde{\bm{f}}_h(\mathbf{y}) &= \mathbf{y} + h \widetilde{\mathbf{N}}(\mathbf{y}) + \frac{h^2}{2} (D_\mathbf{y} \widetilde{\mathbf{N}}(\mathbf{y}))\widetilde{\mathbf{N}}(\mathbf{y}) + \mathbf{R}(h^3, \mathbf{y}), \label{eq:fhtildeapprox}\\
    \widetilde{\bm{f}}_h^{-1}(\mathbf{y}) &= \mathbf{y} - h \widetilde{\mathbf{N}}(\mathbf{y}) + \frac{h^2}{2} (D_\mathbf{y} \widetilde{\mathbf{N}}(\mathbf{y}))\widetilde{\mathbf{N}}(\mathbf{y}) + \mathbf{R}(h^3, \mathbf{y}).  \label{eq:fhtildeinvapprox}
\end{align}
\end{proposition}
Applying the previous proposition on \eqref{eq:fhtilde} and \eqref{eq:fhstartildeinv}, we get
\begin{align}
    \bm{f}_h(\mathbf{y}) &= \mathbf{v} + h \mathbf{N}(\mathbf{y}) + \frac{h^2}{2} (D_\mathbf{v} \mathbf{N}(\mathbf{y}))\mathbf{N}(\mathbf{y}) + \mathbf{R}(h^3, \mathbf{y}),\label{eq:fh}\\
    \bm{f}^{\star -1}_h(\mathbf{y}) &= \mathbf{v} - h \mathbf{N}(\mathbf{y}) + \frac{h^2}{2} (D_\mathbf{v} \mathbf{N}(\mathbf{y}))\mathbf{N}(\mathbf{y}) + \mathbf{R}(h^3, \mathbf{y}).\label{eq:fhstarinv}
\end{align}

The following lemma approximates $\log |\det D {\bm{f}}_{h/2}\left(\mathbf{y}; \bm{\beta}\right)|$ in the objective functions and connects it with Lemma \ref{lemmma:OmegahProp}.
\begin{lemma} \label{lemma:logdetfh}
Let $\widetilde{\bm{f}}_h$ be the function defined in \eqref{eq:fhtilde}. It holds
\begin{align*}
        2 \log |\det D {\bm{f}}_{h/2}\left(\mathbf{Y}_{t_k}; \bm{\beta}\right)| &=  h \tr D_\mathbf{v} {\mathbf{N}}(\mathbf{Y}_{t_{k-1}}; \bm{\beta})  + {R}(h^2, \mathbf{Y}_{t_{k-1}}),\\
        2 \log |\det D {\bm{f}}_{h/2}\left(\mathbf{X}_{t_k}, \Delta_h \mathbf{X}_{t_{k+1}}; \bm{\beta}\right)| &=  h \tr D_\mathbf{v} {\mathbf{N}}(\mathbf{Y}_{t_{k-1}}; \bm{\beta}) + {R}(h^{3/2}, \mathbf{Y}_{t_{k-1}}). 
    \end{align*}
\end{lemma}
An immediate consequence of the previous lemma and that $D_\mathbf{v}\mathbf{F} (\mathbf{y}; \bm{\beta}) = \mathbf{A}_{\mathbf{v}}(\bm{\beta})  + D_\mathbf{v}\mathbf{N} (\mathbf{y}; \bm{\beta})$ is
    \begin{align*}
    \log \det \bm{\Omega}_h^{\mathrm{[RR]}}(\bm{\theta}) + 2\log |\det D {\bm{f}}_{h/2}\left(\mathbf{Y}_{t_k}; \bm{\beta}\right)| &= \log \det h \bm{\Sigma}\bm{\Sigma}^\top + h\tr D_\mathbf{v} \mathbf{F}(\mathbf{Y}_{t_{k-1}}; \bm{\beta}) + R(h^{3/2}, \mathbf{Y}_{t_{k-1}}).
\end{align*} 
The same equality holds when $\mathbf{Y}_{t_k}$ is approximated by $(\mathbf{X}_{t_k}, \Delta_h \mathbf{X}_{t_{k+1}})$. 

The following lemma expands function $\bm{\mu}_h(\widetilde{\bm{f}}_{h/2}(\mathbf{y}))$ up to the lowest necessary order of $h$.
\begin{lemma} \label{lemma:mu}
For the functions $\widetilde{\bm{f}}_{h}$ in \eqref{eq:fhtilde} and $\widetilde{\bm{\mu}}_h$ in \eqref{eq:muh_splitted}, it holds
\begin{align}
    \bm{\mu}_h^{\mathrm{[S]}}(\widetilde{\bm{f}}_{h/2}(\mathbf{y})) &= \mathbf{x} + h\mathbf{v} + \frac{h^2}{2} \mathbf{F}(\mathbf{y}) + \mathbf{R}(h^3,\mathbf{y}), \label{eq:muhsapprox}\\
    \bm{\mu}_h^{\mathrm{[R]}}(\widetilde{\bm{f}}_{h/2}(\mathbf{y})) &= \mathbf{v} + h(\mathbf{F}(\mathbf{y}) - \frac{1}{2} \mathbf{N}(\mathbf{y})) + \mathbf{R}(h^2,\mathbf{y}). \label{eq:muhrapprox}
\end{align}
\end{lemma}

\subsection[Random variables]{Random variables $\widetilde{\mathbf{Z}}_{k, k-1}$ and $\widetilde{\overline{\mathbf{Z}}}_{k+1, k, k-1}$}

To approximate the random variables $\mathbf{Z}_{k, k-1}^{\mathrm{[S]}}(\bm{\beta}), \mathbf{Z}_{k, k-1}^{\mathrm{[R]}}(\bm{\beta})$, $\overline{\mathbf{Z}}_{k, k-1}^{[\mathrm{S}]}(\bm{\beta})$, and $\overline{\mathbf{Z}}_{k+1, k, k-1}^{[\mathrm{R}]}(\bm{\beta})$ around $\mathbf{Y}_{t_{k-1}}$, we start by defining the following random sequences  
\begin{align}    
    \bm{\eta}_{k-1} &\coloneqq \frac{1}{h^{1/2}}\int_{t_{k-1}}^{t_k} \dif \mathbf{W}_t, \label{eq:eta}\\
    \bm{\xi}_{k-1} & \coloneqq  \frac{1}{h^{3/2}} \int_{t_{k-1}}^{t_k} (t - t_{k-1})\dif \mathbf{W}_t, &&\bm{\xi}_k'  \coloneqq \frac{1}{h^{3/2}} \int_{t_k}^{t_{k+1}} (t_{k+1} - t)\dif \mathbf{W}_t, \label{eq:xi}\\
    \bm{\zeta}_{k-1} &\coloneqq \frac{1}{h^{5/2}}\int_{t_{k-1}}^{t_k} (t - t_{k-1})^2\dif \mathbf{W}_t, && \bm{\zeta}_{k}' \coloneqq \frac{1}{h^{5/2}}\int_{t_k}^{t_{k+1}} (t_{k+1} - t)^2\dif \mathbf{W}_t\label{eq:zeta}.
\end{align}
The random variables \eqref{eq:eta}-\eqref{eq:zeta} are Gaussian with mean zero. Moreover, at time $t_k$ they are $\mathcal{F}_{t_{k+1}}$ measurable and independent of $\mathcal{F}_{t_k}$. The following linear combinations of \eqref{eq:eta}-\eqref{eq:zeta} appear in the expansions in the partial observation case
\begin{align}
    \mathbf{U}_{k, k-1} & \coloneqq \bm{\xi}_{k}' + \bm{\xi}_{k-1}, \label{eq:Uk}\\
    \mathbf{Q}_{k, k-1} & \coloneqq \bm{\zeta}_{k}' + 2\bm{\eta}_{k-1} - \bm{\zeta}_{k-1}. \label{eq:Tk}
\end{align}
It is not hard to check that $\bm{\xi}_{k}' + \bm{\eta}_{k-1} - \bm{\xi}_{k-1}' = \mathbf{U}_{k, k-1}$. This alternative representation of $\mathbf{U}_{k, k-1}$ will be used later in proofs. 

The It\^{o} isometry yields
\begin{align}
    \mathbb{E}_{\bm{\theta}_0}[\bm{\eta}_{k-1}\bm{\eta}_{k-1}^\top \mid \mathcal{F}_{t_{k-1}}] &=  \mathbf{I}, && \mathbb{E}_{\bm{\theta}_0}[\bm{\eta}_{k-1}\bm{\xi}_{k-1}^\top \mid \mathcal{F}_{t_{k-1}}] =\mathbb{E}_{\bm{\theta}_0}[\bm{\eta}_{k-1}\bm{\xi}_{k-1}'^\top \mid \mathcal{F}_{t_{k-1}}] =  \frac{1}{2}\mathbf{I}, \label{eq:etaxi} \\
    \mathbb{E}_{\bm{\theta}_0}[\bm{\xi}_{k-1}\bm{\xi}_{k-1}'^\top \mid \mathcal{F}_{t_{k-1}}] &= \frac{1}{6}\mathbf{I}, && \mathbb{E}_{\bm{\theta}_0}[\bm{\xi}_{k-1}\bm{\xi}_{k-1}^\top \mid \mathcal{F}_{t_{k-1}}] =\mathbb{E}_{\bm{\theta}_0}[\bm{\xi}_k'\bm{\xi}_k'^\top \mid \mathcal{F}_{t_{k-1}}] =  \frac{1}{3}\mathbf{I},  \label{eq:xi'xi'}\\
    \mathbb{E}_{\bm{\theta}_0}[\mathbf{U}_{k, k-1}\mathbf{U}_{k, k-1}^\top \mid \mathcal{F}_{t_{k-1}}] &= \frac{2}{3}\mathbf{I}, && \mathbb{E}_{\bm{\theta}_0}[\mathbf{U}_{k, k-1}(\mathbf{U}_{k, k-1} + 2 \bm{\xi}_{k-1}')^\top \mid \mathcal{F}_{t_{k-1}}] = \mathbf{I}. \label{eq:UU}
\end{align}
The covariances of other combinations of the random variables \eqref{eq:eta}-\eqref{eq:zeta} are not needed for the proofs. However, we need some fourth moments calculated in Appendix \ref{sec:Appendix} to derive asymptotic properties.

The following two propositions are the last building blocks for approximating the objective functions \eqref{eq:L_CR}-\eqref{eq:L_CSR} and \eqref{eq:L_PR}-\eqref{eq:L_PSR}. 

\begin{proposition} \label{prop:ZtkandZtkbar} The random variables $\widetilde{\mathbf{Z}}_{k, k-1}(\bm{\beta})$ in \eqref{eq:Ztktilde} and $\widetilde{\overline{\mathbf{Z}}}_{k+1, k, k-1}(\bm{\beta})$ in \eqref{eq:Ztkbartilde} are approximated as
\begin{align*}
    \mathbf{Z}_{k, k-1}^{\mathrm{[S]}}(\bm{\beta}) &= h^{3/2} \bm{\Sigma}_0 \bm{\xi}'_{k-1} + \frac{h^2}{2} (\mathbf{F}_0(\mathbf{Y}_{t_{k-1}}) - \mathbf{F}(\mathbf{Y}_{t_{k-1}})) + \frac{h^{5/2}}{2} D_\mathbf{v}\mathbf{F}_0(\mathbf{Y}_{t_{k-1}}) \bm{\Sigma}_0 \bm{\zeta}'_{k-1} + \mathbf{R}(h^3,\mathbf{Y}_{t_{k-1}}),\\
    \mathbf{Z}_{k, k-1}^{\mathrm{[R]}}(\bm{\beta}) &= h^{1/2} \bm{\Sigma}_0 \bm{\eta}_{k-1} + h (\mathbf{F}_0(\mathbf{Y}_{t_{k-1}}) - \mathbf{F}(\mathbf{Y}_{t_{k-1}})) - \frac{h^{3/2}}{2}D_\mathbf{v}\mathbf{N}(\mathbf{Y}_{t_{k-1}}) \bm{\Sigma}_0 \bm{\eta}_{k-1}\\
    &+h^{3/2} D_\mathbf{v} \mathbf{F}_0(\mathbf{Y}_{t_{k-1}}) \bm{\Sigma}_0 \bm{\xi}'_{k-1} + \mathbf{R}(h^2,\mathbf{Y}_{t_{k-1}}),\\ 
    \overline{\mathbf{Z}}_{k, k-1}^{[\mathrm{S}]}(\bm{\beta}) &= - \frac{h^2}{2}\mathbf{F}(\mathbf{Y}_{t_{k-1}}) - \frac{h^{5/2}}{2}D_\mathbf{v}\mathbf{F}(\mathbf{Y}_{t_{k-1}})\bm{\Sigma}_0 \bm{\xi}'_{k-1} + \mathbf{R}(h^3,\mathbf{Y}_{t_{k-1}}),\\
    \overline{\mathbf{Z}}_{k+1, k, k-1}^{[\mathrm{R}]}(\bm{\beta}) &= h^{1/2}\bm{\Sigma}_0 \mathbf{U}_{k, k-1} +  h(\mathbf{F}_0(\mathbf{Y}_{t_{k-1}}) - \mathbf{F}(\mathbf{Y}_{t_{k-1}}))- \frac{h^{3/2}}{2}D_\mathbf{v}\mathbf{N}(\mathbf{Y}_{t_{k-1}})\bm{\Sigma}_0 \mathbf{U}_{k, k-1}\notag\\
    &-h^{3/2} D_\mathbf{v} \mathbf{F}(\mathbf{Y}_{t_{k-1}})\bm{\Sigma}_0 \bm{\xi}_{k-1}'+ \frac{h^{3/2}}{2} D_\mathbf{v}\mathbf{F}_0(\mathbf{Y}_{t_{k-1}}) \bm{\Sigma}_0 \mathbf{Q}_{k, k-1}   + \mathbf{R}(h^2,\mathbf{Y}_{t_{k-1}}). 
\end{align*}
\end{proposition}

\begin{remark}\label{rmrk:Differences}
    Proposition \ref{prop:ZtkandZtkbar} yields
    \begin{align*}
        \mathbb{E}_{\bm{\theta}_0}[ \mathbf{Z}_{k, k-1}^{\mathrm{[R]}}(\bm{\beta}) \mathbf{Z}_{k, k-1}^{\mathrm{[R]}}(\bm{\beta})^\top \mid \mathbf{Y}_{t_{k-1}}] &= h\bm{\Sigma}\bm{\Sigma}_0^\top + \mathbf{R}(h^2, \mathbf{Y}_{t_{k-1}}) = \bm{\Omega}_h^\mathrm{[RR]} + \mathbf{R}(h^2, \mathbf{Y}_{t_{k-1}}),\\
        \mathbb{E}_{\bm{\theta}_0}[\overline{\mathbf{Z}}_{k+1, k, k-1}^{\mathrm{[R]}}(\bm{\beta}) \overline{\mathbf{Z}}_{k+1, k, k-1}^{\mathrm{[R]}}(\bm{\beta})^\top \mid \mathbf{Y}_{t_{k-1}}] &=  \frac{2}{3}h\bm{\Sigma}\bm{\Sigma}_0^\top + \mathbf{R}(h^2, \mathbf{Y}_{t_{k-1}}) = \frac{2}{3}\bm{\Omega}_h^\mathrm{[RR]} + \mathbf{R}(h^2, \mathbf{Y}_{t_{k-1}}).
    \end{align*}
    Thus, the correction factor $2/3$ in \eqref{eq:L_PR} compensates for the underestimation of the covariance of $\overline{\mathbf{Z}}_{k+1, k, k-1}^{\mathrm{[R]}}(\bm{\beta})$. Similarly, it can be shown that the same underestimation happens when using the backward difference. On the other hand, when using the central difference, it can be shown that 
    \begin{align*}
        \mathbb{E}_{\bm{\theta}_0}[\overline{\mathbf{Z}}_{k+1, k, k-1}^{\mathrm{[R]}, central}(\bm{\beta}) \overline{\mathbf{Z}}_{k+1, k, k-1}^{\mathrm{[R]}, central}(\bm{\beta})^\top \mid \mathbf{Y}_{t_{k-1}}] &=  \frac{5}{12}h\bm{\Sigma}\bm{\Sigma}_0^\top + \mathbf{R}(h^2, \mathbf{Y}_{t_{k-1}}),
    \end{align*}
    which is a larger deviation from $\bm{\Omega}_h^\mathrm{[RR]}$, yielding a larger correcting factor and larger asymptotic variance of the diffusion parameter estimator. 
\end{remark}

\begin{proposition}
\label{prop:ZtkS} Let $\widetilde{\mathbf{Z}}_{k, k-1}(\bm{\beta})$ and $\widetilde{\overline{\mathbf{Z}}}_{k+1, k, k-1}(\bm{\beta})$ be defined as in \eqref{eq:Ztktilde} and \eqref{eq:Ztkbartilde}, respectively. Then
\begin{align*}
    \mathbf{Z}_{k, k-1}^{\mathrm{[S\mid R]}}(\bm{\beta}) &= -\frac{h^{3/2}}{2} \bm{\Sigma}_0 (\bm{\eta}_{k-1} - 2\bm{\xi}'_{k-1}) +\frac{h^{5/2}}{12}(\mathbf{A}_{\mathbf{v}} - \bm{\Sigma}\bm{\Sigma}^\top\mathbf{A}_{\mathbf{v}}(\bm{\beta})^\top(\bm{\Sigma}\bm{\Sigma}^\top)^{-1})\bm{\Sigma}_0\bm{\eta}_{k-1} \\
    &+\frac{h^{5/2}}{4}  D_\mathbf{v} \mathbf{N}(\mathbf{Y}_{t_{k-1}})\bm{\Sigma}_0\bm{\eta}_{k-1} -\frac{h^{5/2}}{2} D_\mathbf{v} \mathbf{F}_0(\mathbf{Y}_{t_{k-1}})\bm{\Sigma}_0(\bm{\xi}'_{k-1}-\bm{\zeta}'_{k-1} )+ \mathbf{R}(h^3,\mathbf{Y}_{t_{k-1}}),\\
    \overline{\mathbf{Z}}_{k+1, k, k-1}^{\mathrm{[S\mid R]}}(\bm{\beta}) &= - \frac{h^{3/2}}{2}\bm{\Sigma}_0 \mathbf{U}_{k, k-1} - \frac{h^2}{2}\mathbf{F}_0(\mathbf{Y}_{t_{k-1}})+\frac{h^{5/2}}{12}(\mathbf{A}_{\mathbf{v}} - \bm{\Sigma}\bm{\Sigma}^\top\mathbf{A}_{\mathbf{v}}(\bm{\beta})^\top(\bm{\Sigma}\bm{\Sigma}^\top)^{-1})\bm{\Sigma}_0\mathbf{U}_{k, k-1}\\
    &+ \frac{h^{5/2}}{4}  D_\mathbf{v} \mathbf{N}(\mathbf{Y}_{t_{k-1}})\bm{\Sigma}_0\mathbf{U}_{k, k-1}-\frac{h^{5/2}}{4} D_\mathbf{v} \mathbf{F}_0(\mathbf{Y}_{t_{k-1}})\bm{\Sigma}_0 \mathbf{Q}_{k, k-1} + \mathbf{R}(h^3,\mathbf{Y}_{t_{k-1}}).
\end{align*}
\end{proposition} 

\subsection{Objective functions} 

Starting with the complete observation case, we approximate the objective functions $\mathcal{L}^{\mathrm{[CR]}}$ \eqref{eq:L_CR} and $\mathcal{L}^{\mathrm{[CS\mid R]}}$ \eqref{eq:L_CSR} up to order $R(h^{3/2}, \mathbf{Y}_{t_{k-1}})$ by $\mathcal{L}_N^{\mathrm{[CR]}}$ \eqref{eq:asymptotic_L_CR} and $\mathcal{L}_N^{\mathrm{[CS\mid R]}}$ \eqref{eq:asymptotic_L_CSR}. This approximation is enough to prove the asymptotic properties of the estimators $\hat{\bm{\theta}}_N^{\mathrm{[CR]}}$ and $\hat{\bm{\theta}}_N^{\mathrm{[CS\mid R]}}$. After omitting the terms of order $R(h, \mathbf{Y}_{t_{k-1}})$ that do not depend on $\bm{\beta}$, we obtain the following approximations
\begin{align}
     \mathcal{L}_N^{\mathrm{[CR]}}\left(\mathbf{Y}_{0:t_N}; \bm{\theta}\right) &\coloneqq (N-1) \log \det \bm{\Sigma}\bm{\Sigma}^\top  + \sum_{k=1}^{N}\bm{\eta}_{k-1}^\top \bm{\Sigma}_0^\top (\bm{\Sigma}\bm{\Sigma}^\top)^{-1}\bm{\Sigma}_0 \bm{\eta}_{k-1}\label{eq:asymptotic_L_CR}\\
     &+ 2\sqrt{h}\sum_{k=1}^{N}  \bm{\eta}_{k-1}^\top \bm{\Sigma}_0^\top  (\bm{\Sigma}\bm{\Sigma}^\top)^{-1}(\mathbf{F}(\mathbf{Y}_{t_{k-1}}; \bm{\beta}_0) - \mathbf{F}(\mathbf{Y}_{t_{k-1}}; \bm{\beta}))\notag\\ 
     &+ h\sum_{k=1}^{N} (\mathbf{F}(\mathbf{Y}_{t_{k-1}}; \bm{\beta}_0) - \mathbf{F}(\mathbf{Y}_{t_{k-1}}; \bm{\beta}))^\top (\bm{\Sigma}\bm{\Sigma}^\top)^{-1}(\mathbf{F}(\mathbf{Y}_{t_{k-1}}; \bm{\beta}_0) - \mathbf{F}(\mathbf{Y}_{t_{k-1}}; \bm{\beta}))\notag\\
     &- h\sum_{k=1}^{N} \bm{\eta}_{k-1}^\top \bm{\Sigma}_0^\top D_\mathbf{v} \mathbf{F}(\mathbf{Y}_{t_{k-1}}; \bm{\beta})^\top(\bm{\Sigma}\bm{\Sigma}^\top)^{-1}\bm{\Sigma}_0 \bm{\eta}_{k-1}   +  h\sum_{k=1}^{N} \tr D_\mathbf{v} \mathbf{F}(\mathbf{Y}_{t_k}; \bm{\beta}), \notag\\ 
     \mathcal{L}_N^{\mathrm{[CS\mid R]}}\left(\mathbf{Y}_{0:t_N}; \bm{\theta}\right) &\coloneqq (N-1) \log \det \bm{\Sigma}\bm{\Sigma}^\top  + 3\sum_{k=1}^{N}(\bm{\eta}_{k-1} - 2\bm{\xi}'_{k-1})^\top \bm{\Sigma}_0^\top (\bm{\Sigma}\bm{\Sigma}^\top)^{-1}\bm{\Sigma}_0 (\bm{\eta}_{k-1} - 2\bm{\xi}'_{k-1})\label{eq:asymptotic_L_CSR}\\
     &- 3 h\sum_{k=1}^{N}  (\bm{\eta}_{k-1} - 2\bm{\xi}'_{k-1})^\top \bm{\Sigma}_0^\top (\bm{\Sigma}\bm{\Sigma}^\top)^{-1}D_\mathbf{v} \mathbf{N}(\mathbf{Y}_{t_{k-1}}; \bm{\beta})\bm{\Sigma}_0 \bm{\eta}_{k-1}\notag\\
     &-  h\sum_{k=1}^{N}  (\bm{\eta}_{k-1} - 2\bm{\xi}'_{k-1})^\top \bm{\Sigma}_0^\top (\bm{\Sigma}\bm{\Sigma}^\top)^{-1}(\mathbf{A}_{\mathbf{v}}(\bm{\beta}) - \bm{\Sigma}\bm{\Sigma}^\top\mathbf{A}_{\mathbf{v}}(\bm{\beta})^\top(\bm{\Sigma}\bm{\Sigma}^\top)^{-1})\bm{\Sigma}_0 \bm{\eta}_{k-1} \notag\\
     \mathcal{L}_N^{\mathrm{[CF]}}\left(\mathbf{Y}_{0:t_N}; \bm{\theta}\right) &\coloneqq \mathcal{L}_N^{\mathrm{[CR]}}\left(\mathbf{Y}_{0:t_N}; \bm{\theta}\right) + \mathcal{L}_N^{\mathrm{[CS\mid R]}}\left(\mathbf{Y}_{0:t_N}; \bm{\theta}\right). \label{eq:asymptotic_L_CF}
\end{align} 
The two last sums in \eqref{eq:asymptotic_L_CSR} converge to zero because $\mathbb{E}_{\bm{\theta}_0}[(\bm{\eta}_{k-1} - 2 \bm{\xi}'_{k-1})\bm{\eta}_{k-1}^\top | \mathcal{F}_{t_{k-1}}] = \mathbf{0}$. Moreover, \eqref{eq:asymptotic_L_CSR} lacks the quadratic form of $\mathbf{F}(\mathbf{Y}_{t_{k-1}}) - \mathbf{F}_0(\mathbf{Y}_{t_{k-1}})$, which is crucial for the asymptotic variance of the drift estimator. This implies that the objective function $\mathcal{L}_N^{\mathrm{[CS\mid R]}}$ is not suitable for estimating the drift parameter. Conversely, \eqref{eq:asymptotic_L_CSR} provides a correct and consistent estimator of the diffusion parameter, indicating that the full objective function (the sum of $\mathcal{L}_N^{\mathrm{[CR]}}$ and $\mathcal{L}_N^{\mathrm{[CS\mid R]}}$) consistently estimates $\bm{\theta}$. 

Similarly, the approximated objective functions in the partial observation case are
\begin{align}
     \mathcal{L}_N^{\mathrm{[PR]}}&\left(\mathbf{Y}_{0:t_N}; \bm{\theta}\right) = \frac{2}{3}(N-2) \log \det \bm{\Sigma}\bm{\Sigma}^\top  + \sum_{k=1}^{N-1}\mathbf{U}_{k, k-1}^\top \bm{\Sigma}_0^\top (\bm{\Sigma}\bm{\Sigma}^\top)^{-1}\bm{\Sigma}_0 \mathbf{U}_{k, k-1}\label{eq:asymptotic_L_PR}\\
     &+ 2\sqrt{h}\sum_{k=1}^{N}  \mathbf{U}_{k, k-1}^\top \bm{\Sigma}_0^\top  (\bm{\Sigma}\bm{\Sigma}^\top)^{-1}(\mathbf{F}(\mathbf{Y}_{t_{k-1}}; \bm{\beta}_0) - \mathbf{F}(\mathbf{Y}_{t_{k-1}}; \bm{\beta}))\notag\\
     &+ h\sum_{k=1}^{N-1} (\mathbf{F}(\mathbf{Y}_{t_{k-1}}; \bm{\beta}_0) - \mathbf{F}(\mathbf{Y}_{t_{k-1}}; \bm{\beta}))^\top (\bm{\Sigma}\bm{\Sigma}^\top)^{-1}(\mathbf{F}(\mathbf{Y}_{t_{k-1}}; \bm{\beta}_0) - \mathbf{F}(\mathbf{Y}_{t_{k-1}}; \bm{\beta}))\notag\\
     &- h\sum_{k=1}^{N-1} (\mathbf{U}_{k, k-1} + 2 \bm{\xi}_{k-1}')^\top \bm{\Sigma}_0^\top D_\mathbf{v} \mathbf{F}(\mathbf{Y}_{t_{k-1}}; \bm{\beta})^\top(\bm{\Sigma}\bm{\Sigma}^\top)^{-1}\bm{\Sigma}_0 \mathbf{U}_{k, k-1}  +  h\sum_{k=1}^{N-1} \tr D_\mathbf{v} \mathbf{F}(\mathbf{Y}_{t_k}; \bm{\beta}),\notag\\ 
     \mathcal{L}_N^{\mathrm{[PS\mid R]}}&\left(\mathbf{Y}_{0:t_N}; \bm{\theta}\right) = 2(N-2) \log \det \bm{\Sigma}\bm{\Sigma}^\top  + 3\sum_{k=1}^{N-1}\mathbf{U}_{k, k-1}^\top \bm{\Sigma}_0^\top (\bm{\Sigma}\bm{\Sigma}^\top)^{-1}\bm{\Sigma}_0 \mathbf{U}_{k, k-1}\label{eq:asymptotic_L_PSR}\\
     &+6\sqrt{h}\sum_{k=1}^{N}  \mathbf{U}_{k, k-1}^\top \bm{\Sigma}_0^\top  (\bm{\Sigma}\bm{\Sigma}^\top)^{-1}\mathbf{F}(\mathbf{Y}_{t_{k-1}}; \bm{\beta}_0)\notag\\
     &- 3 h\sum_{k=1}^{N-1} \mathbf{U}_{k, k-1}^\top \bm{\Sigma}_0^\top D_\mathbf{v} \mathbf{N}(\mathbf{Y}_{t_{k-1}}; \bm{\beta})^\top(\bm{\Sigma}\bm{\Sigma}^\top)^{-1}\bm{\Sigma}_0 \mathbf{U}_{k, k-1}+  2 h\sum_{k=1}^{N-1} \tr D_\mathbf{v} \mathbf{N}(\mathbf{Y}_{t_k}; \bm{\beta}),\notag\\
     \mathcal{L}_N^{\mathrm{[PF]}}&\left(\mathbf{Y}_{0:t_N}; \bm{\theta}\right) = \mathcal{L}_N^{\mathrm{[PR]}}\left(\mathbf{Y}_{0:t_N}; \bm{\theta}\right) + \mathcal{L}_N^{\mathrm{[PS\mid R]}}\left(\mathbf{Y}_{0:t_N}; \bm{\theta}\right). \label{eq:asymptotic_L_PF}
\end{align}
Now, the term with $\mathbf{A}_{\mathbf{v}}(\bm{\beta}) - \bm{\Sigma}\bm{\Sigma}^\top\mathbf{A}_{\mathbf{v}}(\bm{\beta})^\top(\bm{\Sigma}\bm{\Sigma}^\top)^{-1}$ vanishes because 
\begin{align*}
     \tr(\bm{\Sigma}_0 \mathbf{U}_{k, k-1} \mathbf{U}_{k, k-1}^\top \bm{\Sigma}_0^\top (\bm{\Sigma}\bm{\Sigma}^\top)^{-1}(\mathbf{A}_{\mathbf{v}}(\bm{\beta}) - \bm{\Sigma}\bm{\Sigma}^\top\mathbf{A}_{\mathbf{v}}(\bm{\beta})^\top(\bm{\Sigma}\bm{\Sigma}^\top)^{-1})) = 0
\end{align*}
due to the symmetry of the matrices and the trace cyclic property. 

Even though the partial observation objective function $\mathcal{L}^{\mathrm{[PR]}}\left(\mathbf{X}_{0:t_N}; \bm{\theta}\right)$ from \eqref{eq:L_PR} depends only on $\mathbf{X}_{0:t_N}$, we could approximate it with $\mathcal{L}_N^{\mathrm{[PR]}}\left(\mathbf{Y}_{0:t_N}; \bm{\theta}\right)$ from \eqref{eq:asymptotic_L_PR}. This is useful for proving the asymptotic normality of the estimator since its asymptotic distribution will depend on the invariant probability $\nu_0$ defined for the solution $\mathbf{Y}$.  

As mentioned in Remark \ref{rmrk:Bias_L_PSF}, the absence of the quadratic form $\mathbf{F}(\mathbf{Y}_{t_{k-1}}) - \mathbf{F}_0(\mathbf{Y}_{t_{k-1}})$ in \eqref{eq:asymptotic_L_PSR} implies that $\mathcal{L}_N^{\mathrm{[PS\mid R]}}$ is not suitable for estimating the drift parameter. Additionally, the penultimate term in \eqref{eq:asymptotic_L_PSR} does not vanish, needing an additional correction term of $2h\sum_{k=1}^{N-1} \tr D_\mathbf{v} \mathbf{N}(\mathbf{Y}_{t_k}; \bm{\beta})$ for consistency. This correction is represented as $4\log |\det D_\mathbf{v}\bm{f}_{h/2}|$ in \eqref{eq:L_PSR}. 

\section{Discussion} \label{sec:Discussion}

While our focus in this paper has been primarily confined to second-order SDEs with no parameters in the smooth components, we are confident that our findings can be extended to encompass models featuring parameters in the drift of the smooth coordinates. 

Extending the current methodology to SDEs with a non-constant diffusion matrix is possible, but one has to be careful. The assumption of a constant diffusion matrix is necessary if the goal is to obtain an Ornstein-Uhlenbeck process in the splitting. The OU process provides a Gaussian transition density, a building block for creating objective functions based on splitting schemes. Splitting strategies beyond the OU process need a different way of obtaining the transition distribution explicitly or by approximation. The development of this theory is outside the scope of this paper.

There are three considerations of the Strang splitting scheme.

First, when choosing a splitting scheme, one must consider the trade-off between accuracy and complexity. The Strang splitting scheme is simple enough to yield a tractable transition density while maintaining high accuracy. Splittings with fewer steps (e.g., Lie-Trotter splitting has two steps compared to Strang with three steps) are too simple and do not have satisfying finite sample properties, as shown in \citep{Pilipovic2024}. Conversely, splitting schemes with more than one stochastic splitting step lose the property of having an explicit transition density due to the convolution of different distributions, as discussed in Remark \ref{rmkr:Strang}. 

Second, splitting the original drift $\mathbf{F}$ into more vector fields is possible, and the proposed method would work as long as there is an explicit transition density. One way of accomplishing this is constraining to only one stochastic linear splitting in the middle of the composition. Under this constraint, splitting the nonlinear part of $\mathbf{F}$ into more deterministic vector fields will yield an explicit pseudo-likelihood. This can be useful if the nonlinear part does not have an explicit solution, but a further splitting yielding explicit solutions can be found. 

Third, the symmetry in the flow composition is built into the Strang splitting \eqref{eq:StrangSplitting}. Composing the flows as $\Phi_{\theta h}^{[2]} \circ \Phi_h^{[1]} \circ \Phi_{(1- \theta)h}^{[2]} $, for $\theta \neq 1/2$ would be possible, however, it will hardly improve the performance. For a detailed review of the splitting schemes in deterministic differential equations, see \citep{Blanes_Casas_Murua_2024}.

\section{Conclusion} \label{sec:Conclusion}

Many fundamental laws of physics and chemistry are formulated as second-order differential equations. Moreover, this model class is important for understanding complex dynamical systems in various applied fields such as ecology and economics. Extending these deterministic models to stochastic second-order differential equations represents a natural generalization, allowing for incorporating uncertainties and variability inherent in real-world systems. While various statistical methods for analyzing data generated from such stochastic models exist, this paper presents a new, more intuitive, computationally efficient, and easily implemented approach, offering improvements in practical applications.

In this study, we propose estimating model parameters using a recently developed SS estimator for SDEs. This estimator has demonstrated well-behaved finite sample results with relatively large sample time steps, particularly in handling highly nonlinear models. We adjust the estimator to the partial observation setting and employ either the full objective function or only the rough objective function. For all four obtained estimators, we establish the consistency and asymptotic normality.

An important contribution of this work is finding the right correction terms in the objective functions to account for the approximation using finite differences of the unobserved rough component. These corrections are needed to remove the bias of parameter estimators when only partial observations are available.

We compared the proposed estimators with state-of-the-art frequentist methods based on three approximations: EM, LG, and LL. The Kramers oscillator simulation study shows that the SS estimators outperform EM and LG in terms of accuracy while maintaining comparable computation times. Furthermore, the SS estimators achieve similar accuracy as the LL estimators but are significantly faster, placing the proposed estimator at an advantageous balance of accuracy and speed among the methods tested.

Applying the SS estimator to a historical paleoclimate dataset derived from Greenland ice cores demonstrates that our method can be applied to real-world data. Although the Kramers oscillator is a simplified model that may not fully capture the complex underlying dynamics, for example, due to being restricted to additive noise, it still enables us to examine transitions between the metastable states within the data. This study underscores both the practical value and limitations of our approach, highlighting the need for future development of models that incorporate non-additive noise.

The proposed method can be extended to more general state-dependent noise types by letting the stochastic part of the splitting be another stochastic process than the OU process. However, the transition density will no longer be known. A possibility is to approximate the transition density by tractable Gaussian approximations using the correct moments \citep{Kessler1997}, which are available for linear drift and second-order polynomials in the squared diffusion matrix, the so-called Pearson diffusions \citep{FormanSorensen2008}. Then, the Strang splitting estimators are directly applicable.

\section*{Acknowledgement}

This work has received funding from the European Union's Horizon 2020 research and innovation program under the Marie Skłodowska-Curie grant agreement No 956107, "Economic Policy in Complex Environments (EPOC)"; and Novo Nordisk Foundation NNF20OC0062958.

\section*{Code Availability}

The \texttt{R} code used to generate the figures and produce the results in this paper is available at \cite{Pilipovic2025} (\url{https://doi.org/10.17894/UCPH.B40994DC-CF4B-4C2B-A3FA-3B49F9896AFB}). Instructions for running the code are provided in the repository.


\clearpage

\setcounter{section}{0}
\setcounter{equation}{0}

\renewcommand{\thesection}{A\arabic{section}}
\renewcommand{\theequation}{A\arabic{equation}}

\section{Appendix} \label{sec:Appendix}

\setcounter{page}{1}

This section provides proofs for all the nontrivial lemmas, propositions, and theorems presented in Sections \ref{sec:EstimatiorProperties} and \ref{sec:AuxiliaryProperties}. The majority of these proofs, especially those in Section \ref{sec:AuxiliaryProperties}, heavily rely on It\^o or Taylor expansions in $h$ around $\mathbf{Y}_{t_{k-1}}$. Additionally, we frequently employ Fubini's theorem as a helpful tool. Our initial focus is on the results from Section \ref{sec:AuxiliaryProperties}, as they constitute technical auxiliary properties essential for understanding the main results outlined in Section \ref{sec:EstimatiorProperties}.

\subsection{Proofs of results from Section  \ref{sec:AuxiliaryProperties}}

\begin{proof}[Proof of Lemma \ref{lemmma:OmegahProp}]
To prove \ref{item:OmegahProp1}, calculate
\begin{align*}
    \bm{\Omega}_h^{\mathrm{[RR]}}(\bm{\theta})^{-1} &= \frac{1}{h}(\bm{\Sigma}\bm{\Sigma}^\top)^{-1}(\mathbf{I} + \frac{h}{2}(\mathbf{A}_{\mathbf{v}}(\bm{\beta}) + \bm{\Sigma}\bm{\Sigma}^\top \mathbf{A}_{\mathbf{v}}(\bm{\beta})^\top (\bm{\Sigma}\bm{\Sigma}^\top)^{-1})^{-1}) + \mathbf{R}(h, \mathbf{y}_0)\notag\\
    &= \frac{1}{h}(\bm{\Sigma}\bm{\Sigma}^\top)^{-1}(\mathbf{I} - \frac{h}{2}(\mathbf{A}_{\mathbf{v}}(\bm{\beta}) + \bm{\Sigma}\bm{\Sigma}^\top \mathbf{A}_{\mathbf{v}}(\bm{\beta})^\top (\bm{\Sigma}\bm{\Sigma}^\top)^{-1})  + \mathbf{R}(h, \mathbf{y}_0)\notag\\
    &=  \frac{1}{h}(\bm{\Sigma}\bm{\Sigma}^\top)^{-1} - \frac{1}{2}((\bm{\Sigma}\bm{\Sigma}^\top)^{-1}\mathbf{A}_{\mathbf{v}}(\bm{\beta}) + \mathbf{A}_{\mathbf{v}}(\bm{\beta})^\top (\bm{\Sigma}\bm{\Sigma}^\top)^{-1})  + \mathbf{R}(h, \mathbf{y}_0).
\end{align*}
Proof of \ref{item:OmegahProp2}:
\begin{align*}
    \bm{\Omega}_h^{\mathrm{[SR]}}(\bm{\theta})\bm{\Omega}_h^{\mathrm{[RR]}}(\bm{\theta})^{-1} &= (\frac{h^2}{2}(\bm{\Sigma}\bm{\Sigma}^\top) + \frac{h^3}{6}(\mathbf{A}_{\mathbf{v}}(\bm{\beta})(\bm{\Sigma}\bm{\Sigma}^\top) + 2\bm{\Sigma}\bm{\Sigma}^\top \mathbf{A}_{\mathbf{v}}(\bm{\beta})^\top ))\frac{1}{h}(\bm{\Sigma}\bm{\Sigma}^\top)^{-1}\\
    &- \frac{h^2}{4}(\bm{\Sigma}\bm{\Sigma}^\top)((\bm{\Sigma}\bm{\Sigma}^\top)^{-1}\mathbf{A}_{\mathbf{v}}(\bm{\beta}) + \mathbf{A}_{\mathbf{v}}(\bm{\beta})^\top (\bm{\Sigma}\bm{\Sigma}^\top)^{-1})) + \mathbf{R}(h^3, \mathbf{y}_0)\\
    &= \frac{h}{2}\mathbf{I} - \frac{h^2}{12}(\mathbf{A}_{\mathbf{v}} - \bm{\Sigma}\bm{\Sigma}^\top\mathbf{A}_{\mathbf{v}}(\bm{\beta})^\top(\bm{\Sigma}\bm{\Sigma}^\top)^{-1}) + \mathbf{R}(h^3, \mathbf{y}_0).
\end{align*}
To prove \ref{item:OmegahProp3}, use the previous result to obtain
\begin{align*}
&\bm{\Omega}_h^{\mathrm{[SR]}}(\bm{\theta})\bm{\Omega}_h^{\mathrm{[RR]}}(\bm{\theta})^{-1}\bm{\Omega}_h^{\mathrm{[RS]}}(\bm{\theta})\\
&= (\frac{h}{2}\mathbf{I} - \frac{h^2}{12}(\mathbf{A}_{\mathbf{v}} - \bm{\Sigma}\bm{\Sigma}^\top\mathbf{A}_{\mathbf{v}}(\bm{\beta})^\top(\bm{\Sigma}\bm{\Sigma}^\top)^{-1}))(\frac{h^2}{2}(\bm{\Sigma}\bm{\Sigma}^\top) + \frac{h^3}{6}(2\mathbf{A}_{\mathbf{v}}(\bm{\beta})(\bm{\Sigma}\bm{\Sigma}^\top) + \bm{\Sigma}\bm{\Sigma}^\top \mathbf{A}_{\mathbf{v}}(\bm{\beta})^\top ))+ \mathbf{R}(h^5, \mathbf{y}_0)\\
&= \frac{h^3}{4}\bm{\Sigma}\bm{\Sigma}^\top + \frac{h^4}{8}(\mathbf{A}_{\mathbf{v}}(\bm{\beta})\bm{\Sigma}\bm{\Sigma}^\top + \bm{\Sigma}\bm{\Sigma}^\top\mathbf{A}_{\mathbf{v}}(\bm{\beta})^\top) + \mathbf{R}(h^5, \mathbf{y}_0).
\end{align*}
Proof of \ref{item:OmegahProp4} follows from \ref{item:OmegahProp3} and Lemma \ref{lemmma:Omegah}. To prove \ref{item:OmegahProp5}, approximate the log-determinant as
\begin{align}
    \log \det \bm{\Omega}_h^{\mathrm{[RR]}}(\bm{\theta}) &=\log \det \left(h \bm{\Sigma}\bm{\Sigma}^\top + \frac{h^2}{2}(\mathbf{A}_{\mathbf{v}}(\bm{\beta})\bm{\Sigma}\bm{\Sigma}^\top + \bm{\Sigma}\bm{\Sigma}^\top \mathbf{A}_{\mathbf{v}}(\bm{\beta})^\top)\right) + {R}(h^2, \mathbf{y}_0)\notag\\
    &= d \log h + \log \det \bm{\Sigma}\bm{\Sigma}^\top + \log \det \left(\mathbf{I} + \frac{h}{2}(\mathbf{A}_{\mathbf{v}}(\bm{\beta}) + \bm{\Sigma}\bm{\Sigma}^\top \mathbf{A}_{\mathbf{v}}(\bm{\beta})^\top (\bm{\Sigma}\bm{\Sigma}^\top)^{-1})\right) + {R}(h^2, \mathbf{y}_0)\notag\\
    &=d \log h + \log \det \bm{\Sigma}\bm{\Sigma}^\top + \frac{h}{2}\tr(\mathbf{A}_{\mathbf{v}}(\bm{\beta}) + \bm{\Sigma}\bm{\Sigma}^\top \mathbf{A}_{\mathbf{v}}(\bm{\beta})^\top (\bm{\Sigma}\bm{\Sigma}^\top)^{-1}) + {R}(h^2, \mathbf{y}_0)\notag\\
    &= d \log h +\log \det \bm{\Sigma}\bm{\Sigma}^\top + h \tr \mathbf{A}_{\mathbf{v}}(\bm{\beta}) + {R}(h^2, \mathbf{y}_0).
\end{align} 
To prove \ref{item:OmegahProp6}, repeat the previous reasoning on \ref{item:OmegahProp4}. The proof of \ref{item:OmegahProp7} follows from $\det \widetilde{\bm{\Omega}}_h = \det \bm{\Omega}_h^{\mathrm{[RR]}} \det \bm{\Omega}_h^{\mathrm{[S\mid R]}}$ and properties \ref{item:OmegahProp5} and \ref{item:OmegahProp6}. 
\end{proof}

\begin{proof}[Proof of Lemma \ref{lemma:logdetfh}]
Using the same approximation as in the previous proof of \ref{item:OmegahProp5}, we obtain
\begin{align}
    2\log |\det D {\bm{f}}_{h/2}\left(\mathbf{y}; \bm{\beta}\right)| &= 2\log | \det (\mathbf{I} + \frac{h}{2} D {\mathbf{N}}(\mathbf{y}; \bm{\beta}))| + {R}(h^2, \mathbf{y})\notag\\
     &= 2\log | 1 + \frac{h}{2} \tr D {\mathbf{N}}(\mathbf{y};\bm{\beta}) | + {R}(h^2, \mathbf{y})\notag\\
     &= h \tr D {\mathbf{N}}(\mathbf{y}; \bm{\beta})  + {R}(h^2, \mathbf{y}) = h \tr D_\mathbf{v} {\mathbf{N}}(\mathbf{y}; \bm{\beta})  + {R}(h^2, \mathbf{y}). \label{eq:logdetterm}
\end{align}
In complete observation, put $\mathbf{Y}_{t_k}$ instead of $\mathbf{y}$ and use It\^o's lemma on $\mathbf{N}(\mathbf{Y}_{t_k})$ as in \eqref{eq:ItoF}. In partial observation, put $(\mathbf{X}_{t_k}, \Delta_h \mathbf{X}_{t_{k+1}})$ instead of $\mathbf{y}$ and approximate $\mathbf{N}(\mathbf{X}_{t_k}, \Delta_h \mathbf{X}_{t_{k+1}})$ as in \eqref{eq:TaylorF}.
\end{proof}

\begin{proof}[Proof of Lemma \ref{lemma:mu}]
We use definition \eqref{eq:ANtilde} and approximation \eqref{eq:fhtildeapprox}, and plug them in \eqref{eq:muh} to obtain
    \begin{align*}
        \widetilde{\bm{\mu}}_h(\widetilde{\bm{f}}_{h/2}(\mathbf{y})) &= e^{\widetilde{\mathbf{A}}h} (\widetilde{\boldsymbol{f}}_{h/2}(\mathbf{y}) - \widetilde{\mathbf{b}}) + \widetilde{\mathbf{b}}\notag\\
        &=\left(\mathbf{I}_{2d} + h \widetilde{\mathbf{A}} + \frac{h^2}{2}\widetilde{\mathbf{A}}^2 + \mathbf{R}(h^3,\mathbf{y})\right)(\widetilde{\boldsymbol{f}}_{h/2}(\mathbf{y}) - \widetilde{\mathbf{b}}) + \widetilde{\mathbf{b}}\notag\\
        &= \begin{bmatrix}
            \mathbf{I}_d + \frac{h^2}{2} \mathbf{A}_{\mathbf{x}} + \mathbf{R}(h^3,\mathbf{y}) & h \mathbf{I}_d + \frac{h^2}{2}\mathbf{A}_{\mathbf{v}}  +\mathbf{R}(h^3,\mathbf{y}) \vspace{1ex}\\
            h \mathbf{A}_{\mathbf{x}} + \mathbf{R}(h^2,\mathbf{y}) & \mathbf{I}_d + h \mathbf{A}_{\mathbf{v}} + \mathbf{R}(h^2,\mathbf{y})
        \end{bmatrix} \begin{bmatrix}
            \mathbf{x} - \mathbf{b}\\
            \mathbf{v}  + \frac{h}{2} \mathbf{N}(\mathbf{y}) +  \mathbf{R}(h^2,\mathbf{y})
        \end{bmatrix} + \begin{bmatrix}
            \mathbf{b}\\
            \mathbf{0}
        \end{bmatrix}\notag\\
        &=\begin{bmatrix}
            \mathbf{x} + h\mathbf{v} + \frac{h^2}{2} \mathbf{F}(\mathbf{y}) +  \mathbf{R}(h^3,\mathbf{y})\vspace{1ex}\\
            \mathbf{v} + h(\mathbf{F}(\mathbf{y}) - \frac{1}{2} \mathbf{N}(\mathbf{y})) + \mathbf{R}(h^2,\mathbf{y})
        \end{bmatrix}.
    \end{align*}
    This concludes the proof.
\end{proof}

To prove Proposition \ref{prop:ZtkandZtkbar}, we need the following lemma that provides expansion of $\Delta_h \mathbf{X}_{t_{k+1}} - \Delta_h \mathbf{X}_{t_k}$.
\begin{lemma} \label{lemma:DeltaX}
For process $\Delta_h \mathbf{X}_{t_{k+1}}$ \eqref{eq:DeltahX} it holds
\begin{align}
    \Delta_h \mathbf{X}_{t_{k+1}} - \Delta_h \mathbf{X}_{t_k} &= \sqrt{h}\bm{\Sigma}_0 \mathbf{U}_{k, k-1} +  h  \mathbf{F}_0(\mathbf{Y}_{t_{k-1}}) + \frac{h^{3/2}}{2} D_\mathbf{v}\mathbf{F}_0(\mathbf{Y}_{t_{k-1}}) \bm{\Sigma}_0 \mathbf{Q}_{k, k-1} + \mathbf{R}(h^2, \mathbf{Y}_{t_{k-1}}),\label{eq:DiffDeltahX}\\
    \Delta_h \mathbf{X}_{t_k} - \mathbf{V}_{t_{k-1}} &= \sqrt{h}\bm{\Sigma}_0 \bm{\xi}_{k-1}' +  \frac{h}{2} \mathbf{F}_0(\mathbf{Y}_{t_{k-1}} )  + \frac{h^{3/2}}{2} D_\mathbf{v}\mathbf{F}_0(\mathbf{Y}_{t_{k-1}}) \bm{\Sigma}_0 \bm{\zeta}'_{k-1} + \mathbf{R}(h^2,\mathbf{Y}_{t_{k-1}}). \label{eq:DiffDeltahXV}
\end{align}
\end{lemma}

\begin{proof}[Proof of Lemma \ref{lemma:DeltaX}] 
Proof of  \eqref{eq:DiffDeltahX}. Equation \eqref{eq:sdeXV} in integral form and \eqref{eq:DeltahX} yield
\begin{align*}
     \Delta_h \mathbf{X}_{t_{k+1}} - \Delta_h \mathbf{X}_{t_k} &= \frac{1}{h}\int_{t_k}^{t_{k+1}} \mathbf{V}_t \dif t - \frac{1}{h}\int_{t_{k-1}}^{t_k} \mathbf{V}_t \dif t = \frac{1}{h}\int_{t_k}^{t_{k+1}} (\mathbf{V}_t - \mathbf{V}_{t_k}) \dif t + \frac{1}{h}\int_{t_{k-1}}^{t_k} (\mathbf{V}_{t_k} - \mathbf{V}_t) \dif t \notag\\
     &=\frac{1}{h}\int_{t_k}^{t_{k+1}} \int_{t_k}^t \mathbf{F}_0(\mathbf{Y}_s) \dif s \dif t + \frac{1}{h}\bm{\Sigma}_0 \int_{t_k}^{t_{k+1}} \int_{t_k}^t \dif \mathbf{W}_s \dif t\notag\\
     &+ \frac{1}{h}\int_{t_{k-1}}^{t_k} \int_t^{t_k} \mathbf{F}_0(\mathbf{Y}_s) \dif s \dif t + \frac{1}{h}\bm{\Sigma}_0 \int_{t_{k-1}}^{t_k} \int_t^{t_k} \dif \mathbf{W}_s \dif t.
\end{align*}
Apply Fubini's theorem on double integrals to obtain
\begin{align}
     \Delta_h \mathbf{X}_{t_{k+1}} - \Delta_h \mathbf{X}_{t_k} &= \frac{1}{h}\int_{t_k}^{t_{k+1}}  (t_{k+1} - t) \mathbf{F}_0(\mathbf{Y}_t) \dif t +  \frac{1}{h}\int_{t_{k-1}}^{t_k}  (t - t_{k-1}) \mathbf{F}_0(\mathbf{Y}_t) \dif t + \sqrt{h}\bm{\Sigma}_0 \mathbf{U}_{k, k-1}. \label{eq:DeltaDeltaX_mid}
\end{align}
Applying It\^o's lemma on $\mathbf{F}_0(\mathbf{Y}_t)$ yields the following approximation
\begin{align}
    \mathbf{F}_0(\mathbf{Y}_t) &= \mathbf{F}_0(\mathbf{Y}_{t_k}) + D_\mathbf{v} \mathbf{F}_0(\mathbf{Y}_{t_k}) \bm{\Sigma}_0 \int_{t_k}^t \dif \mathbf{W}_s + \mathbf{R}(h, \mathbf{Y}_{t_{k}}). \label{eq:ItoF}
\end{align}
Plugging \eqref{eq:ItoF} into \eqref{eq:DeltaDeltaX_mid} gives
\begin{align*}
     \Delta_h\mathbf{X}_{t_{k+1}} - \Delta_h\mathbf{X}_{t_k} &= h\mathbf{F}_0(\mathbf{Y}_{t_k})+ \frac{1}{h}D_\mathbf{v}\mathbf{F}_0(\mathbf{Y}_{t_k})\bm{\Sigma}_0\int_{t_k}^{t_{k+1}}\int_{t_k}^t (t_{k+1} - t)\dif \mathbf{W}_s \dif t \\
     &- \frac{1}{h}D_\mathbf{v}\mathbf{F}_0(\mathbf{Y}_{t_k})\bm{\Sigma}_0\int_{t_{k-1}}^{t_k}\int_t^{t_k} (t - t_{k-1})\dif \mathbf{W}_s\dif t + \sqrt{h}\bm{\Sigma}_0 \mathbf{U}_{k, k-1} + \mathbf{R}(h^2, \mathbf{Y}_{t_{k-1}}).
\end{align*}
Once again, apply Fubini's theorem on the double integrals to get
\begin{align*}
    \int_{t_k}^{t_{k+1}}\int_{t_k}^t (t_{k+1} - t)\dif \mathbf{W}_s \dif t = \frac{1}{2} h^{5/2}\bm{\zeta}_k',\qquad \int_{t_{k-1}}^{t_k}\int_t^{t_k} (t - t_{k-1})\dif \mathbf{W}_s\dif t = \frac{1}{2} h^{5/2} \bm{\zeta}_{k-1}.
\end{align*}
So far, we have
\begin{align*}
     \Delta_h\mathbf{X}_{t_{k+1}} - \Delta_h\mathbf{X}_{t_k} &= \sqrt{h}\bm{\Sigma}_0 \mathbf{U}_{k, k-1} + h\mathbf{F}_0(\mathbf{Y}_{t_k}) + \frac{h^{3/2}}{2}D_\mathbf{v}\mathbf{F}_0(\mathbf{Y}_{t_k})\bm{\Sigma}_0(\bm{\zeta}'_k - \bm{\zeta}_{k-1}) + \mathbf{R}(h^2, \mathbf{Y}_{t_{k-1}}).
\end{align*}
To conclude the proof, use It\^o's lemma to get $\mathbf{F}_0(\mathbf{Y}_{t_k}) = \mathbf{F}_0(\mathbf{Y}_{t_{k-1}}) + \sqrt{h} D_\mathbf{v} \mathbf{F}_0(\mathbf{Y}_{t_{k-1}}) \bm{\Sigma}_0 \bm{\eta}_{k-1} +  \mathbf{R}(h, \mathbf{Y}_{t_{k-1}})$.

Proof of \eqref{eq:DiffDeltahXV}. As before, start with \eqref{eq:DeltahX} and use \eqref{eq:ItoF} to get
\begin{align*}
    \Delta_h \mathbf{X}_{t_k} - \mathbf{V}_{t_{k-1}} &= \frac{1}{h}\int_{t_{k-1}}^{t_k}  (t_k - t) \mathbf{F}_0(\mathbf{Y}_t) \dif t + \frac{1}{h}\bm{\Sigma}_0 \int_{t_{k-1}}^{t_k}  (t_k - t) \dif \mathbf{W}_t \\
    &= \sqrt{h}\bm{\Sigma}_0 \bm{\xi}_{k-1}' + \frac{h}{2} \mathbf{F}_0(\mathbf{Y}_{t_{k-1}}) + \frac{1}{h} D_\mathbf{v}\mathbf{F}_0(\mathbf{Y}_{t_k})\bm{\Sigma}_0  \int_{t_{k-1}}^{t_k}\int_{t_{k-1}}^t (t_k - t)\dif \mathbf{W}_s \dif t + \mathbf{R}(h^2,\mathbf{Y}_{t_{k-1}}).
\end{align*}
This concludes the proof.
\end{proof}

\begin{proof}[Proof of Proposition \ref{prop:ZtkandZtkbar}]
The expansion of $\mathbf{Z}_{k, k-1}^{\mathrm{[S]}}$ follows directly from Lemma \ref{lemma:mu} and \ref{lemma:DeltaX}. Indeed, it holds
\begin{align*}
    \mathbf{Z}_{k, k-1}^{\mathrm{[S]}}(\bm{\beta}) &= \mathbf{X}_{t_k} - \mathbf{X}_{t_{k-1}} - h \mathbf{V}_{t_{k-1}} - \frac{h^2}{2} \mathbf{F}(\mathbf{Y}_{t_{k-1}})  + \mathbf{R}(h^3,\mathbf{Y}_{t_{k-1}}) \\
    &= h (\Delta_h \mathbf{X}_{t_k} - \mathbf{V}_{t_{k-1}} ) - \frac{h^2}{2} \mathbf{F}(\mathbf{Y}_{t_{k-1}})  + \mathbf{R}(h^3,\mathbf{Y}_{t_{k-1}}).
\end{align*}
To expand $ \mathbf{Z}_{k, k-1}^{\mathrm{[R]}}$, we use definition \eqref{eq:ZtkR} and approximations \eqref{eq:fhstarinv} and \eqref{eq:muhrapprox}, as follows
\begin{align*}
    \mathbf{Z}_{k, k-1}^{\mathrm{[R]}}(\bm{\beta}) &= \mathbf{V}_{t_k} - \mathbf{V}_{t_{k-1}} - h \mathbf{F}(\mathbf{Y}_{t_{k-1}}) - \frac{h}{2}(\mathbf{N}(\mathbf{Y}_{t_{k}}) - \mathbf{N}(\mathbf{Y}_{t_{k-1}})) + \mathbf{R}(h^2, \mathbf{Y}_{t_{k-1}})\notag\\
    &=\bm{\Sigma}_0 \int_{t_{k-1}}^{t_k} \dif \mathbf{W}_t + \int_{t_{k-1}}^{t_k} \mathbf{F}_0(\mathbf{Y}_t)\dif t  - h \mathbf{F}(\mathbf{Y}_{t_{k-1}}) - \frac{h}{2}(\mathbf{N}(\mathbf{Y}_{t_{k}}) - \mathbf{N}(\mathbf{Y}_{t_{k-1}})) + \mathbf{R}(h^2, \mathbf{Y}_{t_{k-1}})\notag\\
    &= \sqrt{h} \bm{\Sigma}_0 \bm{\eta}_{k-1} + h (\mathbf{F}_0(\mathbf{Y}_{t_{k-1}}) - \mathbf{F}(\mathbf{Y}_{t_{k-1}})) - \frac{h}{2}(\mathbf{N}(\mathbf{Y}_{t_{k}}) - \mathbf{N}(\mathbf{Y}_{t_{k-1}}))\\
    &+D_\mathbf{v} \mathbf{F}_0(\mathbf{Y}_{t_{k-1}}) \bm{\Sigma}_0 \int_{t_{k-1}}^{t_k} \int_{t_{k-1}}^t \dif \mathbf{W}_s \dif t + \mathbf{R}(h^2,\mathbf{Y}_{t_{k-1}}).
\end{align*}
In the last line, we used It\^o's lemma on $\mathbf{F}_0(\mathbf{Y}_t)$ as in \eqref{eq:ItoF}. Again, apply It\^o's lemma on $\mathbf{N}(\mathbf{Y}_{t_k})$ to get
\begin{align*}
    \mathbf{Z}_{k, k-1}^{\mathrm{[R]}}(\bm{\beta}) &= \sqrt{h} \bm{\Sigma}_0 \bm{\eta}_{k-1} + h (\mathbf{F}_0(\mathbf{Y}_{t_{k-1}}) - \mathbf{F}(\mathbf{Y}_{t_{k-1}})) - \frac{h^{3/2}}{2}D_\mathbf{v}\mathbf{N}(\mathbf{Y}_{t_{k-1}}) \bm{\Sigma}_0 \bm{\eta}_{k-1}\\
    &+h^{3/2} D_\mathbf{v} \mathbf{F}_0(\mathbf{Y}_{t_{k-1}}) \bm{\Sigma}_0 \bm{\xi}'_{k-1} + \mathbf{R}(h^2,\mathbf{Y}_{t_{k-1}}).
\end{align*}
The expansion of $\mathbf{Z}_{k, k-1}^{\mathrm{[S]}}$ follows from definition \eqref{eq:ZtkSbar} and plugging $(\mathbf{X}_{t_{k-1}}, \Delta_h \mathbf{X}_{t_k})$ in approximation \eqref{eq:muhrapprox}
\begin{align*}
    \overline{\mathbf{Z}}_{k+1, k, k-1}^{[\mathrm{S}]}(\bm{\beta}) &= \mathbf{X}_{t_k} - \mathbf{X}_{t_{k-1}} - h \Delta_h \mathbf{X}_{t_k} - \frac{h^2}{2}  \mathbf{F}(\mathbf{X}_{t_{k-1}}, \Delta_h \mathbf{X}_{t_k}) + \mathbf{R}(h^3,\mathbf{Y}_{t_{k-1}})\\ 
    &= - \frac{h^2}{2}  \mathbf{F}(\mathbf{X}_{t_{k-1}}, \Delta_h \mathbf{X}_{t_k}) + \mathbf{R}(h^3,\mathbf{Y}_{t_{k-1}}). 
\end{align*}
Use Taylor's formula on $\mathbf{F}(\mathbf{X}_{t_{k-1}}, \Delta_h \mathbf{X}_{t_k})$ to get 
\begin{align}
    \mathbf{F}(\mathbf{X}_{t_{k-1}}, \Delta_h \mathbf{X}_{t_k}) &= \mathbf{F}(\mathbf{Y}_{t_{k-1}}) + D_\mathbf{v} \mathbf{F}(\mathbf{Y}_{t_{k-1}}) (\Delta_h \mathbf{X}_{t_k} - \mathbf{V}_{t_{k-1}}) + \mathbf{R}(h^2,\mathbf{Y}_{t_{k-1}}). \label{eq:TaylorF}
\end{align}
Now, the rest follows from Lemma \ref{lemma:DeltaX}.

Finally, to expand $\overline{\mathbf{Z}}_{k+1, k, k-1}^{[\mathrm{R}]}$, start with definition \eqref{eq:ZtkRbar} and approximations \eqref{eq:fhstarinv}, and \eqref{eq:muhrapprox}
\begin{align*}
    \overline{\mathbf{Z}}_{k+1, k, k-1}^{[\mathrm{R}]}(\bm{\beta}) &=\Delta_h \mathbf{X}_{t_{k+1}} - \Delta_h \mathbf{X}_{t_k} - h  \mathbf{F}(\mathbf{X}_{t_{k-1}}, \Delta_h \mathbf{X}_{t_k}) \notag\\
    &- \frac{h}{2}(\mathbf{N}(\mathbf{X}_{t_k}, \Delta_h \mathbf{X}_{t_{k+1}}) - \mathbf{N}(\mathbf{X}_{t_{k-1}}, \Delta_h \mathbf{X}_{t_k})) + \mathbf{R}(h^2,\mathbf{Y}_{t_{k-1}}).
\end{align*}
Lemma \ref{lemma:DeltaX} yields
\begin{align*}
    \overline{\mathbf{Z}}_{k+1, k, k-1}^{[\mathrm{R}]}(\bm{\beta}) &=  \sqrt{h}\bm{\Sigma}_0 \mathbf{U}_{k, k-1} +  h(\mathbf{F}_0(\mathbf{Y}_{t_{k-1}}) - \mathbf{F}(\mathbf{X}_{t_{k-1}}, \Delta_h \mathbf{X}_{t_k}))\\
    &- \frac{h}{2}(\mathbf{N}(\mathbf{X}_{t_k}, \Delta_h \mathbf{X}_{t_{k+1}}) - \mathbf{N}(\mathbf{X}_{t_{k-1}}, \Delta_h \mathbf{X}_{t_k})) + \frac{h^{3/2}}{2} D_\mathbf{v}\mathbf{F}_0(\mathbf{Y}_{t_{k-1}}) \bm{\Sigma}_0 \mathbf{Q}_{k, k-1} + \mathbf{R}(h^2, \mathbf{Y}_{t_{k-1}}).
\end{align*}
Apply Taylor's formula on $\mathbf{F}(\mathbf{X}_{t_{k-1}}, \Delta_h \mathbf{X}_{t_k})$, $\mathbf{N}(\mathbf{X}_{t_{k}}, \Delta_h \mathbf{X}_{t_{k+1}})$, and $\mathbf{N}(\mathbf{X}_{t_{k-1}}, \Delta_h \mathbf{X}_{t_k})$, to get
\begin{align*}
    \overline{\mathbf{Z}}_{k+1, k, k-1}^{[\mathrm{R}]}(\bm{\beta}) &=  \sqrt{h}\bm{\Sigma}_0 \mathbf{U}_{k, k-1} +  h(\mathbf{F}_0(\mathbf{Y}_{t_{k-1}}) - \mathbf{F}(\mathbf{Y}_{t_{k-1}}))- \frac{h}{2}(\mathbf{N}(\mathbf{Y}_{t_k}) - \mathbf{N}(\mathbf{Y}_{t_{k-1}})) \notag\\
    &+ \frac{h^{3/2}}{2} D_\mathbf{v}\mathbf{F}_0(\mathbf{Y}_{t_{k-1}}) \bm{\Sigma}_0 \mathbf{Q}_{k, k-1}  -h^{3/2} D_\mathbf{v} \mathbf{F}(\mathbf{Y}_{t_{k-1}})\bm{\Sigma}_0 \bm{\xi}_{k-1}'  - \frac{h^{3/2}}{2} D_\mathbf{v} \mathbf{N}(\mathbf{Y}_{t_k})\bm{\Sigma}_0 \bm{\xi}_{k}' \notag\\
    &+ \frac{h^{3/2}}{2} D_\mathbf{v} \mathbf{N}(\mathbf{Y}_{t_{k-1}}))\bm{\Sigma}_0 \bm{\xi}_{k-1}' + \mathbf{R}(h^2,\mathbf{Y}_{t_{k-1}}).
\end{align*}
Finally, applying It\^o's lemma on $\mathbf{N}(\mathbf{Y}_{t_k})$ yields 
\begin{align*}
    \overline{\mathbf{Z}}_{k+1, k, k-1}^{[\mathrm{R}]}(\bm{\beta}) &=  \sqrt{h}\bm{\Sigma}_0 \mathbf{U}_{k, k-1} +  h(\mathbf{F}_0(\mathbf{Y}_{t_{k-1}}) - \mathbf{F}(\mathbf{Y}_{t_{k-1}}))- \frac{h^{3/2}}{2}D_\mathbf{v}\mathbf{N}(\mathbf{Y}_{t_{k-1}})\bm{\Sigma}_0 \bm{\eta}_{k-1} \notag\\
    &+ \frac{h^{3/2}}{2} D_\mathbf{v}\mathbf{F}_0(\mathbf{Y}_{t_{k-1}}) \bm{\Sigma}_0 \mathbf{Q}_{k, k-1}  -h^{3/2} D_\mathbf{v} \mathbf{F}(\mathbf{Y}_{t_{k-1}})\bm{\Sigma}_0 \bm{\xi}_{k-1}'  - \frac{h^{3/2}}{2} D_\mathbf{v} \mathbf{N}(\mathbf{Y}_{t_{k-1}})\bm{\Sigma}_0 \bm{\xi}_{k}' \notag\\
    &+ \frac{h^{3/2}}{2} D_\mathbf{v} \mathbf{N}(\mathbf{Y}_{t_{k-1}}))\bm{\Sigma}_0 \bm{\xi}_{k-1}' + \mathbf{R}(h^2,\mathbf{Y}_{t_{k-1}})\notag\\
    &= \sqrt{h}\bm{\Sigma}_0 \mathbf{U}_{k, k-1} +  h(\mathbf{F}_0(\mathbf{Y}_{t_{k-1}}) - \mathbf{F}(\mathbf{Y}_{t_{k-1}}))- \frac{h^{3/2}}{2}D_\mathbf{v}\mathbf{N}(\mathbf{Y}_{t_{k-1}})\bm{\Sigma}_0 \mathbf{U}_{k, k-1}\notag\\
    &+ \frac{h^{3/2}}{2} D_\mathbf{v}\mathbf{F}_0(\mathbf{Y}_{t_{k-1}}) \bm{\Sigma}_0 \mathbf{Q}_{k, k-1}  -h^{3/2} D_\mathbf{v} \mathbf{F}(\mathbf{Y}_{t_{k-1}})\bm{\Sigma}_0 \bm{\xi}_{k-1}' + \mathbf{R}(h^2,\mathbf{Y}_{t_{k-1}}).
\end{align*}
This concludes the proof. 
\end{proof}

\begin{proof}[Proof of Proposition \ref{prop:ZtkS}]
Combining Proposition \ref{prop:ZtkandZtkbar} and property \ref{item:OmegahProp2} of Lemma \ref{lemmma:OmegahProp} yields
\begin{align*}
        \mathbf{Z}_{k, k-1}^{\mathrm{[S\mid R]}}(\bm{\beta})&=\mathbf{Z}_{k, k-1}^{\mathrm{[S]}}(\bm{\beta}) -  \bm{\Omega}_h^{\mathrm{[SR]}}(\bm{\theta})\bm{\Omega}_h^{\mathrm{[RR]}}(\bm{\theta})^{-1} \mathbf{Z}_{k, k-1}^{\mathrm{[R]}}(\bm{\beta})\\
        &= h^{3/2} \bm{\Sigma}_0 \bm{\xi}'_{k-1} + \frac{h^2}{2} (\mathbf{F}_0(\mathbf{Y}_{t_{k-1}}) - \mathbf{F}(\mathbf{Y}_{t_{k-1}})) + \frac{h^{5/2}}{2} D_\mathbf{v}\mathbf{F}_0(\mathbf{Y}_{t_{k-1}}) \bm{\Sigma}_0 \bm{\zeta}'_{k-1}\\
        &-\left(\frac{h}{2}\mathbf{I} - \frac{h^2}{12}(\mathbf{A}_{\mathbf{v}} - \bm{\Sigma}\bm{\Sigma}^\top\mathbf{A}_{\mathbf{v}}(\bm{\beta})^\top(\bm{\Sigma}\bm{\Sigma}^\top)^{-1})\right)\left(h^{1/2} \bm{\Sigma}_0 \bm{\eta}_{k-1} + h (\mathbf{F}_0(\mathbf{Y}_{t_{k-1}}) - \mathbf{F}(\mathbf{Y}_{t_{k-1}}))\right.\\
        &\left. \hspace{20ex} - \frac{h^{3/2}}{2}D_\mathbf{v}\mathbf{N}(\mathbf{Y}_{t_{k-1}}) \bm{\Sigma}_0 \bm{\eta}_{k-1} +h^{3/2} D_\mathbf{v} \mathbf{F}_0(\mathbf{Y}_{t_{k-1}}) \bm{\Sigma}_0 \bm{\xi}'_{k-1})\right)  + \mathbf{R}(h^3,\mathbf{Y}_{t_{k-1}})\\
        &=h^{3/2} \bm{\Sigma}_0 \bm{\xi}'_{k-1} + \frac{h^2}{2} (\mathbf{F}_0(\mathbf{Y}_{t_{k-1}}) - \mathbf{F}(\mathbf{Y}_{t_{k-1}})) + \frac{h^{5/2}}{2} D_\mathbf{v}\mathbf{F}_0(\mathbf{Y}_{t_{k-1}}) \bm{\Sigma}_0 \bm{\zeta}'_{k-1}\\
        &-h^{3/2} \bm{\Sigma}_0 \bm{\eta}_{k-1}- \frac{h^2}{2} (\mathbf{F}_0(\mathbf{Y}_{t_{k-1}}) - \mathbf{F}(\mathbf{Y}_{t_{k-1}})) + \frac{h^{5/2}}{4} D_\mathbf{v}\mathbf{N}(\mathbf{Y}_{t_{k-1}}) \bm{\Sigma}_0 \bm{\eta}_{k-1} \\
        &-  \frac{h^{5/2}}{2} D_\mathbf{v}\mathbf{F}_0(\mathbf{Y}_{t_{k-1}}) \bm{\Sigma}_0 \bm{\xi}'_{k-1} + \frac{h^{5/2}}{12}(\mathbf{A}_{\mathbf{v}} - \bm{\Sigma}\bm{\Sigma}^\top\mathbf{A}_{\mathbf{v}}(\bm{\beta})^\top(\bm{\Sigma}\bm{\Sigma}^\top)^{-1})\bm{\Sigma}_0 \bm{\eta}_{k-1} +  \mathbf{R}(h^3,\mathbf{Y}_{t_{k-1}}).
\end{align*}
Additionally
\begin{align*}
        \overline{\mathbf{Z}}_{k+1, k, k-1}^{\mathrm{[S\mid R]}}(\bm{\beta})&=\overline{\mathbf{Z}}_{k, k-1}^{[\mathrm{S}]}(\bm{\beta}) -  \bm{\Omega}_h^{\mathrm{[SR]}}(\bm{\theta})\bm{\Omega}_h^{\mathrm{[RR]}}(\bm{\theta})^{-1}\overline{\mathbf{Z}}_{k+1, k, k-1}^{[\mathrm{R}]}(\bm{\beta})\\
        &= - \frac{h^2}{2}\mathbf{F}(\mathbf{Y}_{t_{k-1}}) - \frac{h^{5/2}}{2}D_\mathbf{v}\mathbf{F}(\mathbf{Y}_{t_{k-1}})\bm{\Sigma}_0 \bm{\xi}'_{k-1} -\left(\frac{h}{2}\mathbf{I} - \frac{h^2}{12}(\mathbf{A}_{\mathbf{v}} - \bm{\Sigma}\bm{\Sigma}^\top\mathbf{A}_{\mathbf{v}}(\bm{\beta})^\top(\bm{\Sigma}\bm{\Sigma}^\top)^{-1})\right)\\
        &\cdot\Bigg(h^{1/2}\bm{\Sigma}_0 \mathbf{U}_{k, k-1} +  h(\mathbf{F}_0(\mathbf{Y}_{t_{k-1}}) - \mathbf{F}(\mathbf{Y}_{t_{k-1}}))  - \frac{h^{3/2}}{2}D_\mathbf{v}\mathbf{N}(\mathbf{Y}_{t_{k-1}})\bm{\Sigma}_0 \mathbf{U}_{k, k-1}\\
        &\hspace{5ex}+ \frac{h^{3/2}}{2} D_\mathbf{v}\mathbf{F}_0(\mathbf{Y}_{t_{k-1}}) \bm{\Sigma}_0 \mathbf{Q}_{k, k-1}  -h^{3/2} D_\mathbf{v} \mathbf{F}(\mathbf{Y}_{t_{k-1}})\bm{\Sigma}_0 \bm{\xi}_{k-1}'\Bigg)  + \mathbf{R}(h^3,\mathbf{Y}_{t_{k-1}})\\
        &=- \frac{h^2}{2}\mathbf{F}(\mathbf{Y}_{t_{k-1}}) - \frac{h^{5/2}}{2}D_\mathbf{v}\mathbf{F}(\mathbf{Y}_{t_{k-1}})\bm{\Sigma}_0 \bm{\xi}'_{k-1} -h^{3/2} \bm{\Sigma}_0 \mathbf{U}_{k, k-1}- \frac{h^2}{2} (\mathbf{F}_0(\mathbf{Y}_{t_{k-1}}) - \mathbf{F}(\mathbf{Y}_{t_{k-1}}))  \\
        &+ \frac{h^{5/2}}{4} D_\mathbf{v}\mathbf{N}(\mathbf{Y}_{t_{k-1}}) \bm{\Sigma}_0  \mathbf{U}_{k, k-1} -  \frac{h^{5/2}}{4} D_\mathbf{v}\mathbf{F}_0(\mathbf{Y}_{t_{k-1}}) \bm{\Sigma}_0  \mathbf{Q}_{k, k-1} +  \frac{h^{5/2}}{2} D_\mathbf{v}\mathbf{F}_0(\mathbf{Y}_{t_{k-1}}) \bm{\Sigma}_0 \bm{\xi}'_{k-1}\\
        &+ \frac{h^{5/2}}{12}(\mathbf{A}_{\mathbf{v}} - \bm{\Sigma}\bm{\Sigma}^\top\mathbf{A}_{\mathbf{v}}(\bm{\beta})^\top(\bm{\Sigma}\bm{\Sigma}^\top)^{-1})\bm{\Sigma}_0 \mathbf{U}_{k, k-1} +  \mathbf{R}(h^3,\mathbf{Y}_{t_{k-1}}).
\end{align*}
\end{proof}

\subsection{Proofs from Section \ref{sec:EstimatiorProperties}}

Before we start the proofs, we state the following ergodic property, which is proved by \cite{Kessler1997} in case of complete observations, while \cite{SamsonThieullen2012} proved for both complete and partial observation.
\begin{lemma}[Proposition 4 in \citep{SamsonThieullen2012}] \label{lemma:Kessler}
Let Assumptions \ref{as:NLip}, \ref{as:NPoly} and \ref{as:Invariant} hold, and let $\mathbf{Y}$ be the solution to \eqref{eq:SDE}. Let $g: \mathbb{R}^{2d} \times \Theta \to \mathbb{R}$ be a differentiable function with respect to $\mathbf{y}$ and $\bm{\theta}$ with derivatives of polynomial growth in $\mathbf{y}$, uniformly in $\bm{\theta}$.  If $h \to 0$ and $Nh \to \infty$, then
\begin{align}
    \frac{1}{N-1}\sum_{k=1}^N g\left(\mathbf{Y}_{t_k}; \bm{\theta}\right) &\xrightarrow[\substack{Nh \to \infty\\ h \to 0}]{\mathbb{P}_{\bm{\theta}_0}} \int g\left(\mathbf{y}; \bm{\theta}\right) \dif \nu_0(\mathbf{y}), \\
    \frac{1}{N-2}\sum_{k=1}^{N-1} g\left(\mathbf{X}_{t_k}, \Delta_h \mathbf{X}_{t_k}; \bm{\theta}\right) &\xrightarrow[\substack{Nh \to \infty\\ h \to 0}]{\mathbb{P}_{\bm{\theta}_0}} \int g\left(\mathbf{y}; \bm{\theta}\right) \dif \nu_0(\mathbf{y}), 
\end{align}
uniformly in $\bm{\theta}$.
\end{lemma}

\subsection{Proof of consistency}

In the following sections, we use superscripts $\mathrm{[C\cdot]}$ and $\mathrm{[P\cdot]}$ to refer to any objective function in the complete and partial observation case, respectively. Likewise, $\mathrm{[\cdot F]}$, $\mathrm{[\cdot R]}$ and $\mathrm{[\cdot S\mid R]}$ denotes an objective function based on the full, rough partial and conditional smooth-given-rough pseudo-likelihood, respectively, regardless of the observation case. 

\begin{proof}[Proof of Theorem \ref{thm:Consistency}]
The proof of the consistency of the estimators follows a similar path as in Theorem 5.1 of \citep{Pilipovic2024}. With $\bm{\sigma} \coloneqq \vech(\bm{\Sigma}\bm{\Sigma}^\top) = ([\bm{\Sigma}\bm{\Sigma}^\top]_{11},[\bm{\Sigma}\bm{\Sigma}^\top]_{12}, [\bm{\Sigma}\bm{\Sigma}^\top]_{22}, ..., [\bm{\Sigma}\bm{\Sigma}^\top]_{1d}, ..., [\bm{\Sigma}\bm{\Sigma}^\top]_{dd})$, we half-vectorize $\bm{\Sigma}\bm{\Sigma}^\top$ to avoid working with tensors when computing derivatives with respect to $\bm{\Sigma}\bm{\Sigma}^\top$. Since $\bm{\Sigma}\bm{\Sigma}^\top$ is a symmetric $d\times d$ matrix, $\bm{\sigma}$ is of dimension $s = d(d+1)/2$. For a diagonal matrix, instead of a half-vectorization, we use $\bm{\sigma} \coloneqq \diag(\bm{\Sigma}\bm{\Sigma}^\top)$ and $s=d$ in that case. 

We start by finding the limits in $\mathbb{P}_{\bm{\theta}_0}$ of  
\begin{align}
    \frac{1}{N-1}  \mathcal{L}_N^{\mathrm{[C\cdot]}} (\bm{\beta}, \bm{\sigma}) \quad \text{and} \quad \frac{1}{N-2}  \mathcal{L}_N^{\mathrm{[P\cdot]}} (\bm{\beta}, \bm{\sigma}), \label{eq:LikConsGamma}
\end{align}
for $N h \to \infty$, $h \to 0$, uniformly in $\bm{\theta}$. We apply Lemma 9 in \citep{GenonCatalot&Jacod} to prove the convergence and use Proposition A1 in \citep{Gloter2006} to prove the uniform convergence. For more detailed derivations, see proofs in \citep{Pilipovic2024}. Taking the expectations of \eqref{eq:asymptotic_L_CR}- \eqref{eq:asymptotic_L_PF}, we conclude that
\begin{align*}
    \frac{1}{N-1}  \mathcal{L}_N^{\mathrm{[CR]}} (\bm{\beta}, \bm{\sigma}) &\to \log \det (\bm{\Sigma}\bm{\Sigma}^\top) + \tr((\bm{\Sigma}\bm{\Sigma}^\top)^{-1}\bm{\Sigma}\bm{\Sigma}_0^\top), \\
    \frac{1}{N-1}  \mathcal{L}_N^{\mathrm{[CS\mid R]}} (\bm{\beta}, \bm{\sigma}) &\to \log \det (\bm{\Sigma}\bm{\Sigma}^\top) + \tr((\bm{\Sigma}\bm{\Sigma}^\top)^{-1}\bm{\Sigma}\bm{\Sigma}_0^\top), \\
    \frac{1}{N-1}  \mathcal{L}_N^{\mathrm{[CF]}} (\bm{\beta}, \bm{\sigma}) &\to 2\log \det (\bm{\Sigma}\bm{\Sigma}^\top) + 2\tr((\bm{\Sigma}\bm{\Sigma}^\top)^{-1}\bm{\Sigma}\bm{\Sigma}_0^\top), \\
    \frac{1}{N-2}  \mathcal{L}_N^{\mathrm{[PR]}} (\bm{\beta}, \bm{\sigma}) &\to \frac{2}{3}\log \det (\bm{\Sigma}\bm{\Sigma}^\top) + \frac{2}{3}\tr((\bm{\Sigma}\bm{\Sigma}^\top)^{-1}\bm{\Sigma}\bm{\Sigma}_0^\top),\\
    \frac{1}{N-2}  \mathcal{L}_N^{\mathrm{[PS\mid R]}} (\bm{\beta}, \bm{\sigma}) &\to 2\log \det (\bm{\Sigma}\bm{\Sigma}^\top) + 2\tr((\bm{\Sigma}\bm{\Sigma}^\top)^{-1}\bm{\Sigma}\bm{\Sigma}_0^\top),\\
    \frac{1}{N-2}  \mathcal{L}_N^{\mathrm{[PF]}} (\bm{\beta}, \bm{\sigma}) &\to \frac{8}{3}\log \det (\bm{\Sigma}\bm{\Sigma}^\top) + \frac{8}{3}\tr((\bm{\Sigma}\bm{\Sigma}^\top)^{-1}\bm{\Sigma}\bm{\Sigma}_0^\top),
\end{align*}
in $\mathbb{P}_{\bm{\theta}_0}$, for $N h \to \infty$, $h \to 0$, uniformly in $\bm{\theta}$. From here, the rest of the proof of consistency for $\widehat{\bm{\Sigma}\bm{\Sigma}}_N^{\top\mathrm{[C\cdot]}}$ and $\widehat{\bm{\Sigma}\bm{\Sigma}}_N^{\top\mathrm{[P\cdot]}}$ is the same as in \citep{Pilipovic2024}. The coefficients in front of $\log\det$ terms in the partial observation setup correspond to the correcting factors in definitions of objective functions \eqref{eq:L_PF}-\eqref{eq:L_PSR}. They are needed to match coefficients in front of $\tr$ terms that come from the forward difference's under- or over-estimation of the noise effects. 

To prove the consistency of the drift estimators $\widehat{\bm{\beta}}_N^{\mathrm{[CR]}}$ and $\widehat{\bm{\beta}}_N^{\mathrm{[PR]}}$, we start by finding the limits in $\mathbb{P}_{\bm{\theta}_0}$ of  
\begin{align}
     \frac{1}{(N-1)h} ( \mathcal{L}_N^{\mathrm{[CR]}}(\bm{\beta}, \bm{\sigma})  - \mathcal{L}_N^{\mathrm{[CR]}}(\bm{\beta}_0, \bm{\sigma})) \quad \text{and} \quad \frac{1}{(N-2)h} ( \mathcal{L}_N^{\mathrm{[PR]}}(\bm{\beta}, \bm{\sigma})  - \mathcal{L}_N^{\mathrm{[PR]}}(\bm{\beta}_0, \bm{\sigma})), \label{eq:LikConsBeta}
\end{align}
for $N h \to \infty$, $h \to 0$, uniformly in $\bm{\theta}$. Starting with expressions \eqref{eq:asymptotic_L_CR} and \eqref{eq:asymptotic_L_PR} we get 
\begin{align*}
    &\frac{1}{(N-1)h} ( \mathcal{L}_N^{\mathrm{[CR]}}(\bm{\beta}, \bm{\sigma})  - \mathcal{L}_N^{\mathrm{[CR]}}(\bm{\beta}_0, \bm{\sigma})) = \frac{2}{(N-1)\sqrt{h}}\sum_{k=1}^{N}  \bm{\eta}_{k-1}^\top \bm{\Sigma}_0^\top  (\bm{\Sigma}\bm{\Sigma}^\top)^{-1}(\mathbf{F}_0(\mathbf{Y}_{t_{k-1}}) - \mathbf{F}(\mathbf{Y}_{t_{k-1}}))\\
    &\hspace{30ex}+ \frac{1}{N-1}\sum_{k=1}^{N} (\mathbf{F}_0(\mathbf{Y}_{t_{k-1}}) - \mathbf{F}(\mathbf{Y}_{t_{k-1}}))^\top (\bm{\Sigma}\bm{\Sigma}^\top)^{-1}(\mathbf{F}_0(\mathbf{Y}_{t_{k-1}}) - \mathbf{F}(\mathbf{Y}_{t_{k-1}}))\notag\\
    &\hspace{30ex}- \frac{1}{N-1}\sum_{k=1}^{N}  \bm{\eta}_{k-1}^\top \bm{\Sigma}_0^\top D_\mathbf{v} (\mathbf{F}(\mathbf{Y}_{t_{k-1}}) - \mathbf{F}_0(\mathbf{Y}_{t_{k-1}}))^\top(\bm{\Sigma}\bm{\Sigma}^\top)^{-1}\bm{\Sigma}_0 \bm{\eta}_{k-1}\\
    &\hspace{30ex}+\frac{1}{N-1}\sum_{k=1}^{N} \tr D_\mathbf{v} (\mathbf{F}(\mathbf{Y}_{t_k}) - \mathbf{F}_0(\mathbf{Y}_{t_k})),\\
    &\frac{1}{(N-2)h} ( \mathcal{L}_N^{\mathrm{[PR]}}(\bm{\beta}, \bm{\sigma})  - \mathcal{L}_N^{\mathrm{[PR]}}(\bm{\beta}_0, \bm{\sigma})) = \frac{2}{(N-2)\sqrt{h}}\sum_{k=1}^{N-1}  \mathbf{U}_{k, k-1}^\top \bm{\Sigma}_0^\top  (\bm{\Sigma}\bm{\Sigma}^\top)^{-1}(\mathbf{F}_0(\mathbf{Y}_{t_{k-1}}) - \mathbf{F}(\mathbf{Y}_{t_{k-1}}))\notag\\
    &\hspace{20ex}+ \frac{1}{N-2}\sum_{k=1}^{N-1} (\mathbf{F}_0(\mathbf{Y}_{t_{k-1}}) - \mathbf{F}(\mathbf{Y}_{t_{k-1}}))^\top (\bm{\Sigma}\bm{\Sigma}^\top)^{-1}(\mathbf{F}_0(\mathbf{Y}_{t_{k-1}}) - \mathbf{F}(\mathbf{Y}_{t_{k-1}}))\notag\\
    &\hspace{20ex}- \frac{1}{N-2}\sum_{k=1}^{N-1} (\mathbf{U}_{k, k-1}+ 2 \bm{\xi}_{k-1}')^\top \bm{\Sigma}_0^\top D_\mathbf{v} (\mathbf{F}(\mathbf{Y}_{t_{k-1}}) - \mathbf{F}_0(\mathbf{Y}_{t_{k-1}}))^\top(\bm{\Sigma}\bm{\Sigma}^\top)^{-1}\bm{\Sigma}_0 \mathbf{U}_{k, k-1}\\
    &\hspace{20ex}+\frac{1}{N-2}\sum_{k=1}^{N-1} \tr D_\mathbf{v} (\mathbf{F}(\mathbf{Y}_{t_k}) - \mathbf{F}_0(\mathbf{Y}_{t_k})).
\end{align*}
To prove the convergence in probability of the previous two sequences, we use Lemma \ref{lemma:Kessler} and Lemma 9 in \citep{GenonCatalot&Jacod}. To apply Lemma 9 from \citep{GenonCatalot&Jacod}, we need to show that the sum of expectations converges to a certain value while the sum of covariances converges to zero. Here, we only show the former. Moreover, standard tools like Proposition A1 in \citep{Gloter2006} or Lemma 3.1 in \citep{Yoshida1990} can be used to prove uniform convergence. Thus, we look at the expectation to find the limits of these sequences. We use  the known covariances \eqref{eq:etaxi} and \eqref{eq:UU} to get
\begin{align}
    \frac{1}{N h} (\mathcal{L}_N^{\mathrm{[\cdot R]}}(\bm{\beta}, \bm{\sigma})  - \mathcal{L}_N^{\mathrm{[\cdot R]}}(\bm{\beta}_0, \bm{\sigma})) \xrightarrow[\substack{Nh \to \infty\\ h \to 0}]{\mathbb{P}_{\bm{\theta}_0}} &\int(\mathbf{F}_0(\mathbf{y}) - \mathbf{F}(\mathbf{y}))^\top (\bm{\Sigma}\bm{\Sigma}^\top)^{-1} (\mathbf{F}_0(\mathbf{y}) - \mathbf{F}(\mathbf{y})) \dif \nu_0(\mathbf{y})\notag\\
    +&\int \tr (D_\mathbf{v} \left(\mathbf{F}_0\left(\mathbf{y}\right)-\mathbf{F}\left(\mathbf{y}\right)\right)(\bm{\Sigma}\bm{\Sigma}_0^\top(\bm{\Sigma}\bm{\Sigma}^\top)^{-1} - \mathbf{I}))\dif \nu_0(\mathbf{y}). \label{eq:drift_R_cons}
\end{align}
Thus, the consistency of the drift estimator in the partial case coincides with the complete case when using rough objective functions. This is because the right-hand side of \eqref{eq:drift_R_cons} is non-negative when $\bm{\Sigma}\bm{\Sigma}^\top = \bm{\Sigma}\bm{\Sigma}_0^\top$, and the left-hand side is non-positive, following the definition of the likelihood. The remainder of the proof is analogous to that in \cite{Pilipovic2024} and is therefore not repeated here.

Here, we also illustrate why the objective functions based on the conditional pseudo-likelihood of smooth given rough coordinates do not provide identifiable drift estimators. Starting with the complete observations objective function \eqref{eq:asymptotic_L_CSR} and using that $\mathbb{E}_{\bm{\theta}_0}[(\bm{\eta}_{k-1} - 2 \bm{\xi}'_{k-1})\bm{\eta}_{k-1}^\top | \mathcal{F}_{t_{k-1}}] = \mathbf{0}$, we conclude
\begin{align}
     \frac{1}{N h} (\mathcal{L}_N^{\mathrm{[CS\mid R]}}(\bm{\beta}, \bm{\sigma})  - \mathcal{L}_N^{\mathrm{[CS\mid R]}}(\bm{\beta}_0, \bm{\sigma})) \xrightarrow[\substack{Nh \to \infty\\ h \to 0}]{\mathbb{P}_{\bm{\theta}_0}} 0.\label{eq:drift_CSR_cons}
\end{align}
In the partial observation case, we need to add the term $4\log |\det D_\mathbf{v}\bm{f}_{h/2}|$ in \eqref{eq:L_CSR}. Due to this correction, we obtain the consistency of the drift estimator from the following derivations and the fact that the diffusion estimator converges faster
\begin{align*}
    \frac{1}{(N-2)h} ( \mathcal{L}_N^{\mathrm{[PS\mid R]}}(\bm{\beta}, \bm{\sigma})  - \mathcal{L}_N^{\mathrm{[PS\mid R]}}(\bm{\beta}_0, \bm{\sigma})) &=\frac{2}{N-2}\sum_{k=1}^{N-1} \tr D_\mathbf{v} (\mathbf{N}(\mathbf{Y}_{t_k}) - \mathbf{N}_0(\mathbf{Y}_{t_k}))\\
    &\hspace{-12ex}+ \frac{3}{N-2}\sum_{k=1}^{N-1} \tr((\bm{\Sigma}\bm{\Sigma}^\top)^{-1} \bm{\Sigma}_0 \mathbf{U}_{k, k-1}\mathbf{U}_{k, k-1}^\top \bm{\Sigma}_0^\top D_\mathbf{v} (\mathbf{N}_0(\mathbf{Y}_{t_{k-1}}) -\mathbf{N}(\mathbf{Y}_{t_{k-1}}))^\top)\\
    &\hspace{-12ex}\xrightarrow[\substack{Nh \to \infty\\ h \to 0}]{\mathbb{P}_{\bm{\theta}_0}} 2\int \tr (D_\mathbf{v} \left(\mathbf{N}_0\left(\mathbf{y}\right)-\mathbf{N}\left(\mathbf{y}\right)\right)(\bm{\Sigma}\bm{\Sigma}_0^\top(\bm{\Sigma}\bm{\Sigma}^\top)^{-1} - \mathbf{I}))\dif \nu_0(\mathbf{y}). 
\end{align*}
Finally, the consistency of the estimators based on the full objective functions follows from the previous proofs, \eqref{eq:asymptotic_L_CF}, and \eqref{eq:asymptotic_L_PF}. That concludes the proof of consistency.
\end{proof}

\begin{proof}[Proof of Theorem \ref{thm:AsymtoticNormality}]
Here, we only outline the proof. According to Theorem 1 in \citep{Kessler1997} or Theorem 1 in \citep{MSorensenMUchidaSmallDiffusions}, Lemmas \ref{lemma:AsymptoticNormality1} and \ref{lemma:LnConvergence} below are enough for establishing asymptotic normality of $\hat{\bm{\theta}}_N$. For more details, see the proof of Theorem 1 in \citep{MSorensenMUchidaSmallDiffusions}.

\begin{lemma} \label{lemma:AsymptoticNormality1}
Let $\mathbf{C}_N(\bm{\theta}_0)$ be as defined in \eqref{eq:CN}. For $h \to 0$ and $Nh \to \infty$, it holds
\begin{align*}
        &\mathbf{C}_N^\mathrm{[CR]}(\bm{\theta}_0) \xrightarrow[]{\mathbb{P}_{\bm{\theta}_0}} \begin{bmatrix}
        2\mathbf{C}_{\bm{\beta}}(\bm{\theta}_0) & \bm{0}_{r\times s}\\
        \bm{0}_{s\times r} & \mathbf{C}_{\bm{\sigma}}(\bm{\theta}_0)
        \end{bmatrix}, && \mathbf{C}_N^\mathrm{[PR]}(\bm{\theta}_0) \xrightarrow[]{\mathbb{P}_{\bm{\theta}_0}} \begin{bmatrix}
        2\mathbf{C}_{\bm{\beta}}(\bm{\theta}_0) & \bm{0}_{r\times s}\\
        \bm{0}_{s\times r} & \frac{2}{3}\mathbf{C}_{\bm{\sigma}}(\bm{\theta}_0)
        \end{bmatrix}, \\
        &\mathbf{C}_N^\mathrm{[CF]}(\bm{\theta}_0) \xrightarrow[]{\mathbb{P}_{\bm{\theta}_0}} \begin{bmatrix}
        2\mathbf{C}_{\bm{\beta}}(\bm{\theta}_0) & \bm{0}_{r\times s}\\
        \bm{0}_{s\times r} & 2\mathbf{C}_{\bm{\sigma}}(\bm{\theta}_0)
        \end{bmatrix}, && \mathbf{C}_N^\mathrm{[PF]}(\bm{\theta}_0) \xrightarrow[]{\mathbb{P}_{\bm{\theta}_0}} \begin{bmatrix}
        2\mathbf{C}_{\bm{\beta}}(\bm{\theta}_0) & \bm{0}_{r\times s}\\
        \bm{0}_{s\times r} & \frac{8}{3}\mathbf{C}_{\bm{\sigma}}(\bm{\theta}_0)
        \end{bmatrix}.
\end{align*}
Moreover, let $\rho_N$ be a sequence such that $\rho_N \to 0$, then in all cases it holds
\begin{align*}
    \sup_{\|\bm{\theta}\| \leq \rho_N} \|\mathbf{C}_N^\mathrm{[obj]}(\bm{\theta}_0 + \bm{\theta}) -\mathbf{C}_N^\mathrm{[obj]}(\bm{\theta}_0) &\|\xrightarrow[]{\mathbb{P}_{\bm{\theta}_0}} 0.
\end{align*}
\end{lemma}

\begin{lemma} \label{lemma:LnConvergence}
Let $\bm{\lambda}_{N}$ be as defined \eqref{eq:sNLN}. For $h \to 0$, $Nh \to \infty$ and $Nh^2 \to 0$, it holds
\begin{align*}
        &\bm{\lambda}_N^\mathrm{[CR]}  \xrightarrow[]{d} \mathcal{N}\left(\bm{0}, \begin{bmatrix}
        4\mathbf{C}_{\bm{\beta}}(\bm{\theta}_0) & \bm{0}_{r\times s}\\
        \bm{0}_{s\times r} & 2\mathbf{C}_{\bm{\sigma}}(\bm{\theta}_0)
        \end{bmatrix}\right),        
        &&\bm{\lambda}_N^\mathrm{[PR]}  \xrightarrow[]{d}  \mathcal{N}\left(\bm{0}, \begin{bmatrix}
        4\mathbf{C}_{\bm{\beta}}(\bm{\theta}_0) & \bm{0}_{r\times s}\\
        \bm{0}_{s\times r} & \mathbf{C}_{\bm{\sigma}}(\bm{\theta}_0)
        \end{bmatrix}\right),\\
        &\bm{\lambda}_N^\mathrm{[CF]}  \xrightarrow[]{d} \mathcal{N}\left(\bm{0}, \begin{bmatrix}
        4\mathbf{C}_{\bm{\beta}}(\bm{\theta}_0) & \bm{0}_{r\times s}\\
        \bm{0}_{s\times r} & 4\mathbf{C}_{\bm{\sigma}}(\bm{\theta}_0)
        \end{bmatrix}\right),        
        &&\bm{\lambda}_N^\mathrm{[PF]}  \xrightarrow[]{d}  \mathcal{N}\left(\bm{0}, \begin{bmatrix}
        4\mathbf{C}_{\bm{\beta}}(\bm{\theta}_0) & \bm{0}_{r\times s}\\
        \bm{0}_{s\times r} & 16\mathbf{C}_{\bm{\sigma}}(\bm{\theta}_0)
        \end{bmatrix}\right),
\end{align*}
under $\mathbb{P}_{\bm{\theta}_0}$.
\end{lemma}

Now, the two previous lemmas suggest
\begin{align*}
    \mathbf{s}_N^\mathrm{[obj]} = (\mathbf{D}_n^\mathrm{[obj]})^{-1} \bm{\lambda}_N^\mathrm{[obj]} \xrightarrow[]{d} \mathbf{C}_N^\mathrm{[obj]}(\bm{\theta}_0)^{-1}\bm{\lambda}_N^\mathrm{[obj]}.
\end{align*}
The previous line is not entirely formal, but it gives the intuition. For more details on formally deriving the result, see Section 7.4 in \cite{Pilipovic2024} or proof of Theorem 1 in \cite{MSorensenMUchidaSmallDiffusions}. 
\end{proof}

In the following proofs, we recurrently use the fact that $\partial_{\theta^{(i)}} R(h^p, \mathbf{Y}_{t_{k-1}}) = R(h^p, \mathbf{Y}_{t_{k-1}})$, for any $p$ and $i$.

\begin{proof}[Proof of Lemma \ref{lemma:AsymptoticNormality1}]
    We start by proving the first part of the lemma for both complete and partial cases using the rough objective functions \eqref{eq:asymptotic_L_CR} and \eqref{eq:asymptotic_L_PR}. First, we find their second derivatives with respect to $\bm{\beta}$
    \begin{align*}
        \frac{1}{(N-1) h} \partial_{\beta^{(i_1)}\beta^{(i_2)}}^2\mathcal{L}_N^{\mathrm{[CR]}}\left(\mathbf{Y}_{0:t_N}; \bm{\theta}\right) &=\frac{2}{(N-1)\sqrt{h}} \sum_{k=1}^{N}  \bm{\eta}_{k-1}^\top \bm{\Sigma}_0^\top  (\bm{\Sigma}\bm{\Sigma}^\top)^{-1}\partial_{\beta^{(i_1)}\beta^{(i_2)}}^2 \mathbf{F}(\mathbf{Y}_{t_{k-1}}; \bm{\beta}) \notag\\
        &\hspace{-25ex}+ \frac{1}{N-1}\sum_{k=1}^{N} \tr D_\mathbf{v}  \partial_{\beta^{(i_1)}\beta^{(i_2)}}^2 \mathbf{F}(\mathbf{Y}_{t_k}; \bm{\beta})+ \frac{2}{N-1}\sum_{k=1}^{N} \partial_{\beta^{(i_1)}}\mathbf{F}(\mathbf{Y}_{t_{k-1}}; \bm{\beta})^\top (\bm{\Sigma}\bm{\Sigma}^\top)^{-1}\partial_{\beta^{(i_2)}}\mathbf{F}(\mathbf{Y}_{t_{k-1}}; \bm{\beta})\notag\\
        &\hspace{-25ex}- \frac{2}{N-1}\sum_{k=1}^{N} \partial_{\beta^{(i_1)}\beta^{(i_2)}}^2\mathbf{F}(\mathbf{Y}_{t_{k-1}}; \bm{\beta})^\top (\bm{\Sigma}\bm{\Sigma}^\top)^{-1}(\mathbf{F}(\mathbf{Y}_{t_{k-1}}; \bm{\beta}_0) - \mathbf{F}(\mathbf{Y}_{t_{k-1}}; \bm{\beta}))\notag\\
        &\hspace{-25ex}- \frac{1}{N-1}\sum_{k=1}^{N} \bm{\eta}_{k-1}^\top \bm{\Sigma}_0^\top D_\mathbf{v} \partial_{\beta^{(i_1)}\beta^{(i_2)}}^2\mathbf{F}(\mathbf{Y}_{t_{k-1}}; \bm{\beta})^\top(\bm{\Sigma}\bm{\Sigma}^\top)^{-1}\bm{\Sigma}_0 \bm{\eta}_{k-1} , \\
        \frac{1}{(N-2) h} \partial_{\beta^{(i_1)}\beta^{(i_2)}}^2\mathcal{L}_N^{\mathrm{[PR]}}\left(\mathbf{Y}_{0:t_N}; \bm{\theta}\right) &=\frac{2}{(N-2)\sqrt{h}} \sum_{k=1}^{N-1}  \mathbf{U}_{k, k-1}^\top \bm{\Sigma}_0^\top  (\bm{\Sigma}\bm{\Sigma}^\top)^{-1}\partial_{\beta^{(i_1)}\beta^{(i_2)}}^2 \mathbf{F}(\mathbf{Y}_{t_{k-1}}; \bm{\beta}) \notag\\
        &\hspace{-25ex}+\frac{1}{N-2}\sum_{k=1}^{N-1} \tr D_\mathbf{v} \partial_{\beta^{(i_1)}\beta^{(i_2)}}^2\mathbf{F}(\mathbf{Y}_{t_k}; \bm{\beta})+ \frac{2}{N-2}\sum_{k=1}^{N-1} \partial_{\beta^{(i_1)}}\mathbf{F}(\mathbf{Y}_{t_{k-1}}; \bm{\beta})^\top (\bm{\Sigma}\bm{\Sigma}^\top)^{-1}\partial_{\beta^{(i_2)}}\mathbf{F}(\mathbf{Y}_{t_{k-1}}; \bm{\beta})\notag\\
        &\hspace{-25ex}- \frac{2}{N-2}\sum_{k=1}^{N-1} \partial_{\beta^{(i_1)}\beta^{(i_2)}}^2\mathbf{F}(\mathbf{Y}_{t_{k-1}}; \bm{\beta})^\top (\bm{\Sigma}\bm{\Sigma}^\top)^{-1}(\mathbf{F}(\mathbf{Y}_{t_{k-1}}; \bm{\beta}_0) - \mathbf{F}(\mathbf{Y}_{t_{k-1}}; \bm{\beta}))\notag\\
        &\hspace{-25ex}- \frac{1}{N-2}\sum_{k=1}^{N-1} (\mathbf{U}_{k, k-1} + 2 \bm{\xi}_{k-1}')^\top \bm{\Sigma}_0^\top D_\mathbf{v} \partial_{\beta^{(i_1)}\beta^{(i_2)}}^2\mathbf{F}(\mathbf{Y}_{t_{k-1}}; \bm{\beta})^\top(\bm{\Sigma}\bm{\Sigma}^\top)^{-1}\bm{\Sigma}_0 \mathbf{U}_{k, k-1}.
    \end{align*}
As in the proof of consistency, it holds
\begin{align*}
    \frac{1}{N h} \partial_{\beta^{(i_1)}\beta^{(i_2)}}^2\mathcal{L}_N^{\mathrm{[\cdot R]}} \left(\mathbf{Y}_{0:t_N}; \bm{\theta}\right) \Big|_{\bm{\theta} = \bm{\theta}_0} &\xrightarrow[]{\mathbb{P}_{\bm{\theta}_0}} 2 \int  \partial_{\beta^{(i_1)}}\mathbf{F}_0(\mathbf{y}) ^\top (\bm{\Sigma}\bm{\Sigma}^\top)^{-1}  \partial_{\beta^{(i_2)}}\mathbf{F}_0(\mathbf{y}) \dif \nu_0(\mathbf{y}),
\end{align*}
uniformly in $\bm{\theta}$. Now, we investigate the limit of $\frac{1}{(N-1) \sqrt{h}} \partial_{\beta^{(i_1)}\sigma^{(j_2)}}^2\mathcal{L}_N^{\mathrm{[CR]}}(\bm{\theta})$ and $\partial_{\beta^{(i_1)}\sigma^{(j_2)}}^2\mathcal{L}_N^{\mathrm{[PR]}}(\bm{\theta})$
\begin{align*}
        \frac{1}{(N-1) \sqrt{h}} \partial_{\beta^{(i_1)}\sigma^{(j_2)}}^2\mathcal{L}_N^{\mathrm{[CR]}}\left(\mathbf{Y}_{0:t_N}; \bm{\theta}\right) &= -\frac{2}{N-1} \sum_{k=1}^{N}  \bm{\eta}_{k-1}^\top \bm{\Sigma}_0^\top  \partial_{\sigma^{(j_2)}}(\bm{\Sigma}\bm{\Sigma}^\top)^{-1}\partial_{\beta^{(i_1)}} \mathbf{F}(\mathbf{Y}_{t_{k-1}}; \bm{\beta})\notag\\
        &+ \frac{1}{N-1}\sum_{k=1}^N R(\sqrt{h}, \mathbf{Y}_{t_{k-1}})\notag\\
        \frac{1}{(N-2) \sqrt{h}} \partial_{\beta^{(i_1)}\sigma^{(j_2)}}^2\mathcal{L}_N^{\mathrm{[PR]}}\left(\mathbf{Y}_{0:t_N}; \bm{\theta}\right) &=-\frac{2}{N-2} \sum_{k=1}^{N-1}  \mathbf{U}_{k, k-1}^\top \bm{\Sigma}_0^\top  \partial_{\sigma^{(j_2)}}(\bm{\Sigma}\bm{\Sigma}^\top)^{-1}\partial_{\beta^{(i_1)}} \mathbf{F}(\mathbf{Y}_{t_{k-1}}; \bm{\beta})\notag\\ 
        &+ \frac{1}{N-2}\sum_{k=1}^{N-1} R(\sqrt{h}, \mathbf{Y}_{t_{k-1}})\notag
\end{align*}
Both previous sequences converge to zero due to Lemma 9 in \citep{GenonCatalot&Jacod}. Next, we look at the limits of $\frac{1}{N-1} \partial_{\sigma^{(j_1)}\sigma^{(j_2)}}^2\mathcal{L}_N^{\mathrm{[CR]}}(\bm{\theta})$ and $\frac{1}{N-2} \partial_{\sigma^{(j_1)}\sigma^{(j_2)}}^2\mathcal{L}_N^{\mathrm{[PR]}}(\bm{\theta})$
   \begin{align*}
        \frac{1}{N-1} \partial_{\sigma^{(j_1)}\sigma^{(j_2)}}^2\mathcal{L}_N^{\mathrm{[CR]}}\left(\mathbf{Y}_{0:t_N}; \bm{\theta}\right)&=  \partial_{\sigma^{(j_1)}\sigma^{(j_2)}}^2 \log \det (\bm{\Sigma}\bm{\Sigma}^\top)  \\
        &\hspace{-17ex}+ \frac{1}{N-1}\sum_{k=1}^{N}\partial_{\sigma^{(j_1)}\sigma^{(j_2)}}^2\tr(\bm{\Sigma}_0 \bm{\eta}_{k-1}\bm{\eta}_{k-1}^\top \bm{\Sigma}_0^\top (\bm{\Sigma}\bm{\Sigma}^\top)^{-1})+ \frac{1}{N-1}\sum_{k=1}^N R(\sqrt{h}, \mathbf{Y}_{t_{k-1}})\\
        &\hspace{-17ex}=  \tr((\bm{\Sigma}\bm{\Sigma}^\top)^{-1} \partial_{\sigma^{(j_1)}\sigma^{(j_2)}}^2 \bm{\Sigma}\bm{\Sigma}^\top) - \tr ((\bm{\Sigma}\bm{\Sigma}^\top)^{-1} ( \partial_{\sigma^{(j_1)}} \bm{\Sigma}\bm{\Sigma}^\top)(\bm{\Sigma}\bm{\Sigma}^\top)^{-1} \partial_{\sigma^{(j_2)}} \bm{\Sigma}\bm{\Sigma}^\top)\notag\\
        &\hspace{-17ex}-\frac{1}{N-1}\sum_{k=1}^{N}\tr(\bm{\Sigma}_0 \bm{\eta}_{k-1}\bm{\eta}_{k-1}^\top \bm{\Sigma}_0^\top (\bm{\Sigma}\bm{\Sigma}^\top)^{-1}(\partial_{\sigma^{(j_1)}\sigma^{(j_2)}}^2 \bm{\Sigma}\bm{\Sigma}^\top) (\bm{\Sigma}\bm{\Sigma}^\top)^{-1})\notag\\
        &\hspace{-17ex} +\frac{1}{N-1}\sum_{k=1}^{N}\tr(\bm{\Sigma}_0 \bm{\eta}_{k-1}\bm{\eta}_{k-1}^\top \bm{\Sigma}_0^\top (\bm{\Sigma}\bm{\Sigma}^\top)^{-1}(\partial_{\sigma^{(j_1)}} \bm{\Sigma}\bm{\Sigma}^\top) (\bm{\Sigma}\bm{\Sigma}^\top)^{-1}(\partial_{\sigma^{(j_2)}} \bm{\Sigma}\bm{\Sigma}^\top) (\bm{\Sigma}\bm{\Sigma}^\top)^{-1})\notag\\
        &\hspace{-17ex} +\frac{1}{N-1}\sum_{k=1}^{N}\tr(\bm{\Sigma}_0 \bm{\eta}_{k-1}\bm{\eta}_{k-1}^\top \bm{\Sigma}_0^\top (\bm{\Sigma}\bm{\Sigma}^\top)^{-1}(\partial_{\sigma^{(j_2)}} \bm{\Sigma}\bm{\Sigma}^\top) (\bm{\Sigma}\bm{\Sigma}^\top)^{-1}(\partial_{\sigma^{(j_1)}} \bm{\Sigma}\bm{\Sigma}^\top) (\bm{\Sigma}\bm{\Sigma}^\top)^{-1})\\
        &\hspace{-17ex}+ \frac{1}{N-1}\sum_{k=1}^N R(\sqrt{h}, \mathbf{Y}_{t_{k-1}}),\\
        \frac{1}{N-2} \partial_{\sigma^{(j_1)}\sigma^{(j_2)}}^2\mathcal{L}_N^{\mathrm{[PR]}}\left(\mathbf{Y}_{0:t_N}; \bm{\theta}\right) &=   \frac{2}{3}\partial_{\sigma^{(j_1)}\sigma^{(j_2)}}^2 \log \det (\bm{\Sigma}\bm{\Sigma}^\top)\\
        &\hspace{-17ex}+ \frac{1}{N-2}\sum_{k=1}^{N-1}\partial_{\sigma^{(j_1)}\sigma^{(j_2)}}^2\tr(\bm{\Sigma}_0 \mathbf{U}_{k, k-1}\mathbf{U}_{k, k-1}^\top \bm{\Sigma}_0^\top (\bm{\Sigma}\bm{\Sigma}^\top)^{-1}) + \frac{1}{N-2}\sum_{k=1}^{N-1} R(\sqrt{h}, \mathbf{Y}_{t_{k-1}})\notag\\ 
        &\hspace{-17ex} =  \frac{2}{3}\tr((\bm{\Sigma}\bm{\Sigma}^\top)^{-1} \partial_{\sigma^{(j_1)}\sigma^{(j_2)}}^2 \bm{\Sigma}\bm{\Sigma}^\top) - \frac{2}{3}\tr ((\bm{\Sigma}\bm{\Sigma}^\top)^{-1} ( \partial_{\sigma^{(j_1)}} \bm{\Sigma}\bm{\Sigma}^\top)(\bm{\Sigma}\bm{\Sigma}^\top)^{-1} \partial_{\sigma^{(j_2)}} \bm{\Sigma}\bm{\Sigma}^\top)\notag\\
        &\hspace{-17ex} -\frac{1}{N-2}\sum_{k=1}^{N-1}\tr(\bm{\Sigma}_0 \mathbf{U}_{k, k-1}\mathbf{U}_{k, k-1}^\top \bm{\Sigma}_0^\top (\bm{\Sigma}\bm{\Sigma}^\top)^{-1}(\partial_{\sigma^{(j_1)}\sigma^{(j_2)}}^2 \bm{\Sigma}\bm{\Sigma}^\top) (\bm{\Sigma}\bm{\Sigma}^\top)^{-1})\notag\\
        &\hspace{-17ex}+\frac{1}{N-2}\sum_{k=1}^{N-1}\tr(\bm{\Sigma}_0 \mathbf{U}_{k, k-1}\mathbf{U}_{k, k-1}^\top \bm{\Sigma}_0^\top (\bm{\Sigma}\bm{\Sigma}^\top)^{-1}(\partial_{\sigma^{(j_1)}} \bm{\Sigma}\bm{\Sigma}^\top) (\bm{\Sigma}\bm{\Sigma}^\top)^{-1}(\partial_{\sigma^{(j_2)}} \bm{\Sigma}\bm{\Sigma}^\top) (\bm{\Sigma}\bm{\Sigma}^\top)^{-1})\notag\\
        &\hspace{-17ex} +\frac{1}{N-2}\sum_{k=1}^{N-1} \tr(\bm{\Sigma}_0 \mathbf{U}_{k, k-1}\mathbf{U}_{k, k-1}^\top \bm{\Sigma}_0^\top (\bm{\Sigma}\bm{\Sigma}^\top)^{-1}(\partial_{\sigma^{(j_2)}} \bm{\Sigma}\bm{\Sigma}^\top) (\bm{\Sigma}\bm{\Sigma}^\top)^{-1}(\partial_{\sigma^{(j_1)}} \bm{\Sigma}\bm{\Sigma}^\top) (\bm{\Sigma}\bm{\Sigma}^\top)^{-1})\\
        &\hspace{-17ex}+ \frac{1}{N-2}\sum_{k=1}^{N-1} R(\sqrt{h}, \mathbf{Y}_{t_{k-1}}).
\end{align*}
Using the second moments of $\bm{\eta}_{k-1}$ and $\mathbf{U}_{k, k-1}$ with additional calculations, we can conclude that
\begin{align*}
    \frac{1}{N-1} \partial_{\sigma^{(j_1)}\sigma^{(j_2)}}^2\mathcal{L}_N^{\mathrm{[CR]}}\left(\mathbf{Y}_{0:t_N}; \bm{\theta}\right) \Big|_{\bm{\theta} = \bm{\theta}_0} &\xrightarrow[]{\mathbb{P}_{\bm{\theta}_0}}  \tr ((\bm{\Sigma}\bm{\Sigma}^\top)^{-1} ( \partial_{\sigma^{(j_1)}} \bm{\Sigma}\bm{\Sigma}^\top)(\bm{\Sigma}\bm{\Sigma}^\top)^{-1} \partial_{\sigma^{(j_2)}} \bm{\Sigma}\bm{\Sigma}^\top),\\
    \frac{1}{N-2} \partial_{\sigma^{(j_1)}\sigma^{(j_2)}}^2\mathcal{L}_N^{\mathrm{[PR]}}\left(\mathbf{Y}_{0:t_N}; \bm{\theta}\right) \Big|_{\bm{\theta} = \bm{\theta}_0} &\xrightarrow[]{\mathbb{P}_{\bm{\theta}_0}}  \frac{2}{3}\tr ((\bm{\Sigma}\bm{\Sigma}^\top)^{-1} ( \partial_{\sigma^{(j_1)}} \bm{\Sigma}\bm{\Sigma}^\top)(\bm{\Sigma}\bm{\Sigma}^\top)^{-1} \partial_{\sigma^{(j_2)}} \bm{\Sigma}\bm{\Sigma}^\top),
\end{align*}
uniformly in $\bm{\theta}$. To extend the previous results on the objective functions \eqref{eq:asymptotic_L_CSR} and \eqref{eq:asymptotic_L_PSR}, we start by acknowledging that 
\begin{align*}
    \frac{1}{N h} \partial_{\beta^{(i_1)}\beta^{(i_2)}}^2\mathcal{L}_N^{\mathrm{[\cdot S\mid R]}} \left(\mathbf{Y}_{0:t_N}; \bm{\theta}\right) \Big|_{\bm{\theta} = \bm{\theta}_0} &\xrightarrow[]{\mathbb{P}_{\bm{\theta}_0}} 0,
\end{align*}
uniformly in $\bm{\theta}$. The reasons behind this are the same as in the proof of consistency. The same can be said for the limit of $\frac{1}{N\sqrt{h}}\partial_{\bm{\beta}\bm{\sigma}} \mathcal{L}_N^{\mathrm{[\cdot S\mid R]}}(\bm{\theta})$. Finally, repeating the same derivations as before, we get
\begin{align*}
    \frac{1}{N-1} \partial_{\sigma^{(j_1)}\sigma^{(j_2)}}^2\mathcal{L}_N^{\mathrm{[CS\mid R]}}\left(\mathbf{Y}_{0:t_N}; \bm{\theta}\right) \Big|_{\bm{\theta} = \bm{\theta}_0} &\xrightarrow[]{\mathbb{P}_{\bm{\theta}_0}}  \tr ((\bm{\Sigma}\bm{\Sigma}^\top)^{-1} ( \partial_{\sigma^{(j_1)}} \bm{\Sigma}\bm{\Sigma}^\top)(\bm{\Sigma}\bm{\Sigma}^\top)^{-1} \partial_{\sigma^{(j_2)}} \bm{\Sigma}\bm{\Sigma}^\top),\\
    \frac{1}{N-2} \partial_{\sigma^{(j_1)}\sigma^{(j_2)}}^2\mathcal{L}_N^{\mathrm{[PS\mid R]}}\left(\mathbf{Y}_{0:t_N}; \bm{\theta}\right) \Big|_{\bm{\theta} = \bm{\theta}_0} &\xrightarrow[]{\mathbb{P}_{\bm{\theta}_0}}  2\tr ((\bm{\Sigma}\bm{\Sigma}^\top)^{-1} ( \partial_{\sigma^{(j_1)}} \bm{\Sigma}\bm{\Sigma}^\top)(\bm{\Sigma}\bm{\Sigma}^\top)^{-1} \partial_{\sigma^{(j_2)}} \bm{\Sigma}\bm{\Sigma}^\top),
\end{align*}
uniformly in $\bm{\theta}$. This concludes the first part of the lemma. The second part follows from the fact that all limits are continuous in $\bm{\theta}$.
\end{proof}

To prove Lemma \ref{lemma:LnConvergence}, we state another useful property that provides a general formula to calculate moments of a product of two quadratic forms with Gaussian vectors. 
\begin{lemma}\label{lemma:4thmoments}
Let $(\bm{\alpha}_k)_{k=1}^N$, $(\bm{\beta}_k)_{k=1}^N$, be two sequences of independent $\mathcal{F}_{t_{k+1}}$-measurable Gaussian random variables with mean zero. If $\mathbb{E}_{\bm{\theta}_0}[\bm{\alpha}_{k-1}\bm{\beta}_{k-1}^\top \mid \mathcal{F}_{t_{k-1}}]$ is diagonal, and $\mathbf{A}$ and $\mathbf{B}$ are two symmetric positive definite matrices, then
    \begin{align*}
        \mathbb{E}_{\bm{\theta}_0}[\bm{\alpha}_{k-1}^\top\mathbf{A}\bm{\alpha}_{k-1} \bm{\beta}_{k-1}^\top\mathbf{B}\bm{\beta}_{k-1}\mid \mathcal{F}_{t_{k-1}}] &=  2\tr(\mathbb{E}_{\bm{\theta}_0}[\bm{\alpha}_{k-1}\bm{\beta}_{k-1}^\top \mid \mathcal{F}_{t_{k-1}}]\mathbf{A}\mathbb{E}_{\bm{\theta}_0}[\bm{\alpha}_{k-1}\bm{\beta}_{k-1}^\top \mid \mathcal{F}_{t_{k-1}}]\mathbf{B})\\
        &+ \tr (\mathbf{A}\mathbb{E}_{\bm{\theta}_0}[\bm{\alpha}_{k-1}\bm{\alpha}_{k-1}^\top \mid \mathcal{F}_{t_{k-1}}]) \tr (\mathbf{B}\mathbb{E}_{\bm{\theta}_0}[\bm{\beta}_{k-1}\bm{\beta}_{k-1}^\top \mid \mathcal{F}_{t_{k-1}}]).
    \end{align*}
\end{lemma}

\begin{proof} We start with
\begin{align}
    \mathbb{E}_{\bm{\theta}_0}[\bm{\alpha}_{k-1}^\top\mathbf{A}\bm{\alpha}_{k-1} \bm{\beta}_{k-1}^\top\mathbf{B}\bm{\beta}_{k-1}\mid \mathcal{F}_{t_{k-1}}] &=\sum_{i,j,l,m = 1}^{d}  A_{ij} B_{lm} \mathbb{E}_{\bm{\theta}_0}[\alpha_{k-1}^{(i)}\alpha_{k-1}^{(j)} \beta_{k-1}^{(l)}\beta_{k-1}^{(m)}\mid \mathcal{F}_{t_{k-1}}]\notag\\
    &=\sum_{i,j,l,m = 1}^{d}  A_{ij} B_{lm}  \cov(\alpha_{k-1}^{(i)}\alpha_{k-1}^{(j)},\beta_{k-1}^{(l)}\beta_{k-1}^{(m)})\label{eq:covprod}\\
    &+ \sum_{i,j,l,m = 1}^{d}  A_{ij} B_{lm} \mathbb{E}_{\bm{\theta}_0}[\alpha_{k-1}^{(i)}\alpha_{k-1}^{(j)}\mid \mathcal{F}_{t_{k-1}}] \mathbb{E}_{\bm{\theta}_0}[\beta_{k-1}^{(l)}\beta_{k-1}^{(m)}\mid \mathcal{F}_{t_{k-1}}].\notag
\end{align}
We compute \eqref{eq:covprod} using the formula for the covariance of products of centered Gaussian random variables \citep{BohrnstedtGoldberger69}. Then we get
\begin{align*}
    \mathbb{E}_{\bm{\theta}_0}[\bm{\alpha}_{k-1}^\top\mathbf{A}\bm{\alpha}_{k-1} \bm{\beta}_{k-1}^\top\mathbf{B}\bm{\beta}_{k-1}\mid \mathcal{F}_{t_{k-1}}] &=\sum_{i,j,l,m = 1}^{d}   A_{ij} B_{lm}  \mathbb{E}_{\bm{\theta}_0}[\alpha_{k-1}^{(i)}\beta_{k-1}^{(l)}\mid \mathcal{F}_{t_{k-1}}] \mathbb{E}_{\bm{\theta}_0}[\alpha_{k-1}^{(j)}\beta_{k-1}^{(m)}\mid \mathcal{F}_{t_{k-1}}]\\
    &+ \sum_{i,j,l,m = 1}^{d}   A_{ij} B_{lm}  \mathbb{E}_{\bm{\theta}_0}[\alpha_{k-1}^{(i)}\beta_{k-1}^{(m)}\mid \mathcal{F}_{t_{k-1}}] \mathbb{E}_{\bm{\theta}_0}[\alpha_{k-1}^{(j)}\beta_{k-1}^{(l)}\mid \mathcal{F}_{t_{k-1}}]\\
    &+ \sum_{i,j,l,m = 1}^{d}   A_{ij} B_{lm}  \mathbb{E}_{\bm{\theta}_0}[\alpha_{k-1}^{(i)}\alpha_{k-1}^{(j)}\mid \mathcal{F}_{t_{k-1}}] \mathbb{E}_{\bm{\theta}_0}[\beta_{k-1}^{(l)}\beta_{k-1}^{(m)}\mid \mathcal{F}_{t_{k-1}}]\\
    &= 2 \sum_{i,j = 1}^{d}   A_{ij} B_{ij}  \mathbb{E}_{\bm{\theta}_0}[\alpha_{k-1}^{(i)}\beta_{k-1}^{(i)}\mid \mathcal{F}_{t_{k-1}}] \mathbb{E}_{\bm{\theta}_0}[\alpha_{k-1}^{(j)}\beta_{k-1}^{(j)}\mid \mathcal{F}_{t_{k-1}}]\\
    &+\sum_{i,j = 1}^{d}   A_{ii} B_{jj}  \mathbb{E}_{\bm{\theta}_0}[\alpha_{k-1}^{(i)}\alpha_{k-1}^{(i)}\mid \mathcal{F}_{t_{k-1}}] \mathbb{E}_{\bm{\theta}_0}[\beta_{k-1}^{(j)}\beta_{k-1}^{(j)}\mid \mathcal{F}_{t_{k-1}}].
\end{align*}
That concludes the proof.
\end{proof}

Applying the previous Lemma to our setup, corollary \ref{cor:4thmoments} follows immediately.
\begin{corollary}\label{cor:4thmoments}
    Let $(\bm{\eta}_k)_{k=1}^N$, $(\bm{\xi}_k)_{k=1}^N$, $(\bm{\xi}_k')_{k=1}^N$ be random sequences as defined in \eqref{eq:eta} and \eqref{eq:xi}. Let $\mathbf{B}_{j_1}$ and $\mathbf{B}_{j_2}$ be two symmetric positive definite matrices. Then, it holds
    \begin{align}
        \mathbb{E}_{\bm{\theta}_0}[\bm{\eta}_{k-1}^\top\mathbf{B}_{j_1}\bm{\eta}_{k-1} \bm{\eta}_{k-1}^\top\mathbf{B}_{j_2}\bm{\eta}_{k-1}\mid \mathcal{F}_{t_{k-1}}] &=  2\tr(\mathbf{B}_{j_1}\mathbf{B}_{j_2}) + \tr \mathbf{B}_{j_1} \tr \mathbf{B}_{j_2}, \label{eq:4theta}\\
        \mathbb{E}_{\bm{\theta}_0}[\bm{\xi}_{k-1}^\top\mathbf{B}_{j_1}\bm{\xi}_{k-1} \bm{\xi}_{k-1}^\top\mathbf{B}_{j_2}\bm{\xi}_{k-1}\mid \mathcal{F}_{t_{k-1}}] &=  \frac{2}{9}\tr(\mathbf{B}_{j_1}\mathbf{B}_{j_2}) + \frac{1}{9}\tr \mathbf{B}_{j_1} \tr \mathbf{B}_{j_2}, \label{eq:4thxi}\\
        \mathbb{E}_{\bm{\theta}_0}[\bm{\xi}_{k-1}'^\top\mathbf{B}_{j_1}\bm{\xi}_{k-1}' \bm{\xi}_{k-1}'^\top\mathbf{B}_{j_2}\bm{\xi}_{k-1}'\mid \mathcal{F}_{t_{k-1}}] &=  \frac{2}{9}\tr(\mathbf{B}_{j_1}\mathbf{B}_{j_2}) + \frac{1}{9}\tr \mathbf{B}_{j_1} \tr \mathbf{B}_{j_2},\label{eq:4thxi'}\\
        \mathbb{E}_{\bm{\theta}_0}[\bm{\xi}_{k-1}^\top\mathbf{B}_{j_1}\bm{\xi}_{k-1} \bm{\xi}_{k-1}'^\top\mathbf{B}_{j_2}\bm{\xi}_{k-1}'\mid \mathcal{F}_{t_{k-1}}] &=  \frac{1}{18}\tr(\mathbf{B}_{j_1}\mathbf{B}_{j_2}) + \frac{1}{9}\tr \mathbf{B}_{j_1} \tr \mathbf{B}_{j_2}. \label{eq:4thmix}
    \end{align}
\end{corollary}

\begin{proof}[Proof of Lemma \ref{lemma:LnConvergence}] To prove the lemma, we need to compute $\bm{\lambda}_N^\mathrm{[obj]}$. The main part of the proof focuses only on the rough estimators, while at the end of the proof, we discuss how the same ideas can be adapted for other estimators. Thus, we start with $-\frac{1}{\sqrt{N h}}\partial_{\beta^{(i)}} \mathcal{L}_N^\mathrm{[\cdot R]}$
\begin{align*}
        &-\frac{1}{\sqrt{(N-1)h}}\partial_{\beta^{(i)}}\mathcal{L}_N^{\mathrm{[CR]}}\left(\mathbf{Y}_{0:t_N}; \bm{\theta}\right) = \frac{2}{\sqrt{N-1}} \sum_{k=1}^{N}  \bm{\eta}_{k-1}^\top \bm{\Sigma}_0^\top (\bm{\Sigma}\bm{\Sigma}^\top)^{-1}\partial_{\beta^{(i)}} \mathbf{F}(\mathbf{Y}_{t_{k-1}}; \bm{\beta}) \notag\\
        &\hspace{8ex}+ 2\sqrt{\frac{h}{N-1}}\sum_{k=1}^{N} \partial_{\beta^{(i)}} \mathbf{F}(\mathbf{Y}_{t_{k-1}}; \bm{\beta})^\top (\bm{\Sigma}\bm{\Sigma}^\top)^{-1}(\mathbf{F}(\mathbf{Y}_{t_{k-1}}; \bm{\beta}_0) - \mathbf{F}(\mathbf{Y}_{t_{k-1}}; \bm{\beta}))\notag\\
        &\hspace{8ex}+ \sqrt{\frac{h}{N-1}}\sum_{k=1}^{N} \bm{\eta}_{k-1}^\top \bm{\Sigma}_0^\top  D_\mathbf{v} \partial_{\beta^{(i)}} \mathbf{F}(\mathbf{Y}_{t_{k-1}}; \bm{\beta})^\top(\bm{\Sigma}\bm{\Sigma}^\top)^{-1}\bm{\Sigma}_0 \bm{\eta}_{k-1}  -  \sqrt{\frac{h}{N-1}}\sum_{k=1}^{N} \tr D_\mathbf{v} \partial_{\beta^{(i)}} \mathbf{F}(\mathbf{Y}_{t_k}; \bm{\beta}),\\
        &-\frac{1}{\sqrt{(N-2) h}} \partial_{\beta^{(i)}}\mathcal{L}_N^{\mathrm{[PR]}}\left(\mathbf{Y}_{0:t_N}; \bm{\theta}\right)= \frac{2}{\sqrt{N-2}} \sum_{k=1}^{N-1} \mathbf{U}_{k, k-1}^\top \bm{\Sigma}_0^\top (\bm{\Sigma}\bm{\Sigma}^\top)^{-1}\partial_{\beta^{(i)}} \mathbf{F}(\mathbf{Y}_{t_{k-1}}; \bm{\beta})\notag\\
        &\hspace{8ex}+ 2\sqrt{\frac{h}{N-2}}\sum_{k=1}^{N-1}  \partial_{\beta^{(i)}}\mathbf{F}(\mathbf{Y}_{t_{k-1}}; \bm{\beta})^\top (\bm{\Sigma}\bm{\Sigma}^\top)^{-1}(\mathbf{F}(\mathbf{Y}_{t_{k-1}}; \bm{\beta}_0) - \mathbf{F}(\mathbf{Y}_{t_{k-1}}; \bm{\beta}))\notag\\
        &\hspace{8ex}+ \sqrt{\frac{h}{N-2}}\sum_{k=1}^{N-1} (\mathbf{U}_{k, k-1}+2\bm{\xi}_{k-1})^\top \bm{\Sigma}_0^\top D_\mathbf{v}  \partial_{\beta^{(i)}}\mathbf{F}(\mathbf{Y}_{t_{k-1}}; \bm{\beta})^\top(\bm{\Sigma}\bm{\Sigma}^\top)^{-1}\bm{\Sigma}_0 \mathbf{U}_{k, k-1}\\
        &\hspace{8ex}- \sqrt{\frac{h}{N-2}}\sum_{k=1}^{N-1} \tr D_\mathbf{v}  \partial_{\beta^{(i)}}\mathbf{F}(\mathbf{Y}_{t_k}; \bm{\beta}).
\end{align*}
Similarly, for $-\frac{1}{\sqrt{N}}\partial_{\sigma^{(j)}} \mathcal{L}_N^\mathrm{[\cdot R]}$, we get
\begin{align*}
        -\frac{1}{\sqrt{(N-1)}}\partial_{\sigma^{(j)}}\mathcal{L}_N^{\mathrm{[CR]}}\left(\mathbf{Y}_{0:t_N}; \bm{\theta}\right)&= -\frac{1}{\sqrt{N-1}} \sum_{k=1}^{N} \tr((\bm{\Sigma}\bm{\Sigma}^\top)^{-1}\partial_{\sigma^{(j)}}\bm{\Sigma}\bm{\Sigma}^\top)\\
        &\hspace{-13ex}+\frac{1}{\sqrt{N-1}}\sum_{k=1}^N \bm{\eta}_{k-1}^\top \bm{\Sigma}_0^\top (\bm{\Sigma}\bm{\Sigma}^\top)^{-1}(\partial_{\sigma^{(j)}}\bm{\Sigma}\bm{\Sigma}^\top) (\bm{\Sigma}\bm{\Sigma}^\top)^{-1} \bm{\Sigma}_0 \bm{\eta}_{k-1}\notag\\
        &\hspace{-13ex}+ 2\sqrt{\frac{h}{N-1}}\sum_{k=1}^{N}  \bm{\eta}_{k-1}^\top \bm{\Sigma}_0^\top  (\bm{\Sigma}\bm{\Sigma}^\top)^{-1}(\partial_{\sigma^{(j)}}\bm{\Sigma}\bm{\Sigma}^\top) (\bm{\Sigma}\bm{\Sigma}^\top)^{-1} (\mathbf{F}(\mathbf{Y}_{t_{k-1}}; \bm{\beta}_0) - \mathbf{F}(\mathbf{Y}_{t_{k-1}}; \bm{\beta}))\\
        &\hspace{-13ex}+ \sum_{k=1}^{N}R(\frac{h}{\sqrt{N}}, \mathbf{Y}_{t_{k-1}}),\\
        -\frac{1}{\sqrt{(N-2)}} \partial_{\sigma^{(j)}}\mathcal{L}_N^{\mathrm{[PR]}}\left(\mathbf{Y}_{0:t_N}; \bm{\theta}\right) &= -\frac{2}{3\sqrt{N-2}} \sum_{k=1}^{N-1} \tr((\bm{\Sigma}\bm{\Sigma}^\top)^{-1}\partial_{\sigma^{(j)}}\bm{\Sigma}\bm{\Sigma}^\top)\\
        &\hspace{-13ex}+\frac{1}{\sqrt{N-2}}\sum_{k=1}^{N-1} \mathbf{U}_{k, k-1}^\top \bm{\Sigma}_0^\top (\bm{\Sigma}\bm{\Sigma}^\top)^{-1}(\partial_{\sigma^{(j)}}\bm{\Sigma}\bm{\Sigma}^\top) (\bm{\Sigma}\bm{\Sigma}^\top)^{-1} \bm{\Sigma}_0 \mathbf{U}_{k, k-1}\notag\\
        &\hspace{-13ex}+ 2\sqrt{\frac{h}{N-2}}\sum_{k=1}^{N-1}  \mathbf{U}_{k, k-1}^\top \bm{\Sigma}_0^\top  (\bm{\Sigma}\bm{\Sigma}^\top)^{-1}(\partial_{\sigma^{(j)}}\bm{\Sigma}\bm{\Sigma}^\top) (\bm{\Sigma}\bm{\Sigma}^\top)^{-1} (\mathbf{F}(\mathbf{Y}_{t_{k-1}}; \bm{\beta}_0) - \mathbf{F}(\mathbf{Y}_{t_{k-1}}; \bm{\beta}))\\
        &\hspace{-13ex}+ \sum_{k=1}^{N-1}R(\frac{h}{\sqrt{N}}, \mathbf{Y}_{t_{k-1}}).
\end{align*}
To prove the convergence in distribution of $\bm{\lambda}_N^\mathrm{[\cdot R]}$, we introduce the following triangular arrays that arise from the previous calculations
\begin{align}
    \bm{\phi}_{N,k-1}^{\mathrm{[CR]}(i)}(\bm{\theta}_0) &\coloneqq \frac{2}{\sqrt{N-1}}  \bm{\eta}_{k-1}^\top \bm{\Sigma}_0^\top (\bm{\Sigma}\bm{\Sigma}_0^\top)^{-1}\partial_{\beta^{(i)}} \mathbf{F}_0(\mathbf{Y}_{t_{k-1}})\label{eq:phiC}\\
    &+  \sqrt{\frac{h}{N-1}} (\tr(\bm{\Sigma}_0 \bm{\eta}_{k-1} \bm{\eta}_{k-1}^\top \bm{\Sigma}_0^\top  D_\mathbf{v} \partial_{\beta^{(i)}} \mathbf{F}_0(\mathbf{Y}_{t_{k-1}})^\top(\bm{\Sigma}\bm{\Sigma}_0^\top)^{-1}) - \tr D_\mathbf{v} \partial_{\beta^{(i)}} \mathbf{F}_0(\mathbf{Y}_{t_{k-1}})), \notag\\
    \bm{\phi}_{N,k-1}^{\mathrm{[PR]}(i)}(\bm{\theta}_0) &\coloneqq \frac{2}{\sqrt{N-2}}  \mathbf{U}_{k, k-1}^\top \bm{\Sigma}_0^\top (\bm{\Sigma}\bm{\Sigma}_0^\top)^{-1}\partial_{\beta^{(i)}} \mathbf{F}_0(\mathbf{Y}_{t_{k-1}})\label{eq:phiP}\\
    &\hspace{-7ex}+ \sqrt{\frac{h}{N-2}} (\tr(\bm{\Sigma}_0 \mathbf{U}_{k, k-1}(\mathbf{U}_{k, k-1} + 2\bm{\xi}_{k-1}')^\top \bm{\Sigma}_0^\top D_\mathbf{v}  \partial_{\beta^{(i)}}\mathbf{F}_0(\mathbf{Y}_{t_{k-1}})^\top(\bm{\Sigma}\bm{\Sigma}_0^\top)^{-1}) -  \tr D_\mathbf{v}  \partial_{\beta^{(i)}}\mathbf{F}_0(\mathbf{Y}_{t_k})),\notag\\
    \bm{\rho}_{N,k-1}^{\mathrm{[CR]}(j)}(\bm{\theta}_0) &\coloneqq \frac{1}{\sqrt{N-1}}(\bm{\eta}_{k-1}^\top \bm{\Sigma}_0^\top (\bm{\Sigma}\bm{\Sigma}_0^\top)^{-1}(\partial_{\sigma^{(j)}}\bm{\Sigma}\bm{\Sigma}_0^\top) (\bm{\Sigma}\bm{\Sigma}_0^\top)^{-1} \bm{\Sigma}_0 \bm{\eta}_{k-1} - \tr((\bm{\Sigma}\bm{\Sigma}_0^\top)^{-1}\partial_{\sigma^{(j)}}\bm{\Sigma}\bm{\Sigma}_0^\top)),\label{eq:rhoC}\\
    \bm{\rho}_{N,k-1}^{\mathrm{[PR]}(j)}(\bm{\theta}_0) &\coloneqq \frac{1}{\sqrt{N-2}}(\mathbf{U}_{k, k-1}^\top \bm{\Sigma}_0^\top (\bm{\Sigma}\bm{\Sigma}_0^\top)^{-1}(\partial_{\sigma^{(j)}}\bm{\Sigma}\bm{\Sigma}_0^\top) (\bm{\Sigma}\bm{\Sigma}_0^\top)^{-1} \bm{\Sigma}_0 \mathbf{U}_{k, k-1} - \frac{2}{3}\tr((\bm{\Sigma}\bm{\Sigma}_0^\top)^{-1}\partial_{\sigma^{(j)}}\bm{\Sigma}\bm{\Sigma}_0^\top)).\label{eq:rhoP}
\end{align}
Then, $\bm{\lambda}_N^\mathrm{[\cdot R]}$ rewrites as
\begin{align}
    \bm{\lambda}_N^\mathrm{[\cdot R]} = \sum_{k=1}^N \begin{bmatrix}
        \bm{\phi}_{N,k-1}^{\mathrm{[\cdot R]}(1)}(\bm{\theta}_0)\\
        \vdots\\
        \bm{\phi}_{N,k-1}^{\mathrm{[\cdot R]}(r)}(\bm{\theta}_0)\\
        \bm{\rho}_{N,k-1}^{\mathrm{[\cdot R]}(1)}(\bm{\theta}_0)\\
        \vdots\\        
        \bm{\rho}_{N,k-1}^{\mathrm{[\cdot R]}(s)}(\bm{\theta}_0)
    \end{bmatrix} + \frac{1}{N}\sum_{k=1}^{N}R(\sqrt{N h^2}, \mathbf{Y}_{t_{k-1}}). \label{eq:LN_rewritten}
\end{align}
Thus, to establish estimators' asymptotic normality, we need an extra convergence condition $N h^2 \to 0$. This is common in literature, and it is necessary for most estimators. 

To finish the proof, we apply the central limit theorem for martingale difference arrays (Proposition 3.1 in \citep{CRIMALDI2005571}). However, we can not use the same reasoning in complete and partial observation cases. 

First, we notice that in both complete and partial cases, $\bm{\phi}_{N,k-1}^{\mathrm{[\cdot R]}(i)}(\bm{\theta}_0)$ and $ \bm{\rho}_{N,k-1}^{\mathrm{[\cdot R]}(j)}(\bm{\theta}_0)$ are centered conditionally to $\mathcal{F}_{t_{k-1}}$. Moreover, in the complete case, $\bm{\phi}_{N,k-1}^{\mathrm{[CR]}(i)}(\bm{\theta}_0)$ and $ \bm{\rho}_{N,k-1}^{\mathrm{[CR]}(j)}(\bm{\theta}_0)$ are adapted to the filtration $\mathcal{F}_{t_k}$. Thus, the proof follows directly by applying Proposition 3.1 in \citep{CRIMALDI2005571}. This proposition assumes a martingale difference array centered conditionally to $\mathcal{F}_{t_{k-1}}$ and $\mathcal{F}_{t_k}$-measurable. 

In the partial observation case, $\mathbf{U}_{k, k-1}$ is $\mathcal{F}_{t_{k+1}}$-measurable as it depends on random variables $\bm{\xi}_{k-1}$ and $\bm{\xi}_k'$. Consequently, to apply Proposition 3.1 in \citep{CRIMALDI2005571}, it is not enough for $\bm{\phi}_{N,k-1}^{\mathrm{[PR]}(i)}(\bm{\theta}_0)$ and $ \bm{\rho}_{N,k-1}^{\mathrm{[PR]}(j)}(\bm{\theta}_0)$ to be centered conditionally to $\mathcal{F}_{t_{k-1}}$, they also need to be centered conditionally to $\mathcal{F}_{t_{k}}$. The previous condition, however, does not hold. Thus, we use the idea of reordering the sum in $\bm{\lambda}_N^\mathrm{[PR]}$ \eqref{eq:LN_rewritten} to obtain the $\mathcal{F}_{t_{k}}$-measurable and centered conditionally on $\mathcal{F}_{t_{k-1}}$, as proposed by \cite{Gloter2000, Gloter2006} and later used by \cite{SamsonThieullen2012}.  

First, use Lemma 9 from \citep{GenonCatalot&Jacod} to notice that
\begin{equation*}
    \sum_{k=1}^{N-1}\bm{\phi}_{N,k-1}^{\mathrm{[PR]}(i)}(\bm{\theta}_0) = \frac{2}{\sqrt{N-2}}   \sum_{k=1}^{N-1}\mathbf{U}_{k, k-1}^\top \bm{\Sigma}_0^\top (\bm{\Sigma}\bm{\Sigma}_0^\top)^{-1}\partial_{\beta^{(i)}} \mathbf{F}_0(\mathbf{Y}_{t_{k-1}}) + o_{\mathbb{P}_{\bm{\theta}_0}}(1).
\end{equation*}
Then, reorder the sum of $\bm{\phi}_{N,k-1}^{\mathrm{[PR]}(i)}(\bm{\theta}_0)$ as follows
\begin{align*}
    \sum_{k=1}^{N-1}\bm{\phi}_{N,k-1}^{\mathrm{[PR]}(i)}(\bm{\theta}_0) 
    &=\frac{2}{\sqrt{N-2}}\left(\bm{\xi}_0^\top  \bm{\Sigma}_0^\top (\bm{\Sigma}\bm{\Sigma}_0^\top)^{-1}\partial_{\beta^{(i)}} \mathbf{F}_0(\mathbf{Y}_{t_0}) +  \bm{\xi}_{N-1}'^\top  \bm{\Sigma}_0^\top (\bm{\Sigma}\bm{\Sigma}_0^\top)^{-1}\partial_{\beta^{(i)}} \mathbf{F}_0(\mathbf{Y}_{t_{N-2}})\right)\notag\\
    &\hspace{-3ex}+\frac{2}{\sqrt{N-2}}\sum_{k=2}^{N-1}\left(\bm{\xi}_{k-1}^\top  \bm{\Sigma}_0^\top (\bm{\Sigma}\bm{\Sigma}_0^\top)^{-1}\partial_{\beta^{(i)}} \mathbf{F}_0(\mathbf{Y}_{t_{k-1}}) +  \bm{\xi}_{k-1}'^\top  \bm{\Sigma}_0^\top (\bm{\Sigma}\bm{\Sigma}_0^\top)^{-1}\partial_{\beta^{(i)}} \mathbf{F}_0(\mathbf{Y}_{t_{k-2}})\right) + o_{\mathbb{P}_{\bm{\theta}_0}}(1)\notag\\
    &\hspace{-3ex}=\frac{2}{\sqrt{N-2}}\sum_{k=2}^{N-1}\left(\bm{\xi}_{k-1}^\top  \bm{\Sigma}_0^\top (\bm{\Sigma}\bm{\Sigma}_0^\top)^{-1}\partial_{\beta^{(i)}} \mathbf{F}_0(\mathbf{Y}_{t_{k-1}}) +  \bm{\xi}_{k-1}'^\top  \bm{\Sigma}_0^\top (\bm{\Sigma}\bm{\Sigma}_0^\top)^{-1}\partial_{\beta^{(i)}} \mathbf{F}_0(\mathbf{Y}_{t_{k-2}})\right) + o_{\mathbb{P}_{\bm{\theta}_0}}(1).
\end{align*}
Now, the triangular arrays under the sum are centered conditionally on $\mathcal{F}_{t_{k-1}}$ and $\mathcal{F}_{t_k}$ measurable. Thus, define
\begin{equation*}
        \bm{\phi}_{N,k-1}^{\star\mathrm{[PR]}(i)}(\bm{\theta}_0) \coloneqq \frac{2}{\sqrt{N-2}}  \left(\bm{\xi}_{k-1}^\top  \bm{\Sigma}_0^\top (\bm{\Sigma}\bm{\Sigma}_0^\top)^{-1}\partial_{\beta^{(i)}} \mathbf{F}_0(\mathbf{Y}_{t_{k-1}}) +  \bm{\xi}_{k-1}'^\top  \bm{\Sigma}_0^\top (\bm{\Sigma}\bm{\Sigma}_0^\top)^{-1}\partial_{\beta^{(i)}} \mathbf{F}_0(\mathbf{Y}_{t_{k-2}})\right).
\end{equation*}
To apply Proposition 3.1 from \citep{CRIMALDI2005571}, we need the following limits in probability
\begin{align*}
    \sum_{k=2}^{N-1} \mathbb{E}_{\bm{\theta}_0}[\bm{\phi}_{N,k-1}^{\star\mathrm{[PR]}(i_1)}(\bm{\theta}_0)\bm{\phi}_{N,k-1}^{\star\mathrm{[PR]}(i_2)}(\bm{\theta}_0)\mid \mathcal{F}_{t_{k-1}}] &\xrightarrow{\mathbb{P}_{\bm{\theta}_0}} 4[\mathbf{C}_{\bm{\beta}}(\bm{\theta}_0)]_{i_1i_2},\\
    \sum_{k=2}^{N-1} \mathbb{E}_{\bm{\theta}_0}[(\bm{\phi}_{N,k-1}^{\star\mathrm{[PR]}(i_1)}(\bm{\theta}_0)\bm{\phi}_{N,k-1}^{\star\mathrm{[PR]}(i_2)})^2(\bm{\theta}_0)\mid \mathcal{F}_{t_{k-1}}] &\xrightarrow{\mathbb{P}_{\bm{\theta}_0}} 0.
\end{align*}
The first limit follows from properties \eqref{eq:etaxi} and \eqref{eq:xi'xi'}. The second limit follows due to an additional order of $1/N$. 

When looking at $\bm{\rho}_{N, k -1}^{\mathrm{[\cdot R]}(j)}$, we repeat the same reasoning. For notational simplicity, start with defining
\begin{align*}
    \mathbf{B}_j(\bm{\theta}_0) \coloneqq \bm{\Sigma}_0 (\bm{\Sigma}\bm{\Sigma}_0^\top)^{-1}(\partial_{\sigma^{(j)}}\bm{\Sigma}\bm{\Sigma}_0^\top)(\bm{\Sigma}\bm{\Sigma}_0^\top)^{-1}\bm{\Sigma}_0.
\end{align*}
It follows immediately that $\tr(\mathbf{B}_j(\bm{\theta}_0)) = \tr((\bm{\Sigma}\bm{\Sigma}_0^\top)^{-1}\partial_{\sigma^{(j)}}\bm{\Sigma}\bm{\Sigma}_0^\top)$. Again, reorder the sum of $\bm{\rho}_{N,k-1}^{\mathrm{[PR]}(j)}(\bm{\theta}_0)$ as follows
\begin{align*}
    \sum_{k=1}^{N-1}\bm{\rho}_{N,k-1}^{\mathrm{[PR]}(j)}(\bm{\theta}_0) &= \frac{1}{\sqrt{N-2}} \sum_{k=1}^{N-1}(\mathbf{U}_{k, k-1}^\top \mathbf{B}_j(\bm{\theta}_0)\mathbf{U}_{k, k-1} - \frac{2}{3}\tr(\mathbf{B}_j(\bm{\theta}_0)))\notag\\
    &=\frac{1}{\sqrt{N-2}}\sum_{k=2}^{N-1}\left(\bm{\xi}_{k-1}^\top \mathbf{B}_j(\bm{\theta}_0)\bm{\xi}_{k-1} + 2 \bm{\xi}_{k-2}^\top \mathbf{B}_j(\bm{\theta}_0)\bm{\xi}'_{k-1} + \bm{\xi}_{k-1}'^\top \mathbf{B}_j(\bm{\theta}_0)\bm{\xi}_{k-1}' - \frac{2}{3}\tr(\mathbf{B}_j(\bm{\theta}_0))\right)\notag\\
    &+\frac{1}{\sqrt{N-2}}\left(\bm{\xi}_0^\top \mathbf{B}_j(\bm{\theta}_0)\bm{\xi}_{0} + 2 \bm{\xi}_{N-2}^\top \mathbf{B}_j(\bm{\theta}_0)\bm{\xi}'_{N-1} + \bm{\xi}_{N-1}'^\top \mathbf{B}_j(\bm{\theta}_0)\bm{\xi}_{N-1}' - \frac{2}{3}\tr(\mathbf{B}_j(\bm{\theta}_0))\right).
\end{align*}
Since the last term in the previous equation is $o_{\mathbb{P}_{\bm{\theta}_0}}(1)$, we focus only on
\begin{equation*}
    \bm{\rho}_{N,k-1}^{\star\mathrm{[PR]}(j)}(\bm{\theta}_0) \coloneqq \frac{1}{\sqrt{N-2}}\left(\bm{\xi}_{k-1}^\top \mathbf{B}_j(\bm{\theta}_0)\bm{\xi}_{k-1} + 2 \bm{\xi}_{k-2}^\top \mathbf{B}_j(\bm{\theta}_0)\bm{\xi}'_{k-1} + \bm{\xi}_{k-1}'^\top \mathbf{B}_j(\bm{\theta}_0)\bm{\xi}_{k-1}' - \frac{2}{3}\tr(\mathbf{B}_j(\bm{\theta}_0))\right).
\end{equation*}
Notice that $\bm{\rho}_{N,k-1}^{\star\mathrm{[PR]}(j)}(\bm{\theta}_0) $ is $\mathcal{F}_{t_k}$ measurable and centered conditionally on $\mathcal{F}_{t_{k-1}}$. Again, to apply Proposition 3.1 from \citep{CRIMALDI2005571}, we need the following limits in probability
\begin{align*}
    \sum_{k=2}^{N-1} \mathbb{E}_{\bm{\theta}_0}[\bm{\rho}_{N,k-1}^{\star\mathrm{[PR]}(j_1)}(\bm{\theta}_0)\bm{\rho}_{N,k-1}^{\star\mathrm{[PR]}(j_2)}(\bm{\theta}_0)\mid \mathcal{F}_{t_{k-1}}] &\xrightarrow{\mathbb{P}_{\bm{\theta}_0}} [\mathbf{C}_{\bm{\sigma}}(\bm{\theta}_0)]_{j_1j_2},\\
    \sum_{k=2}^{N-1} \mathbb{E}_{\bm{\theta}_0}[(\bm{\rho}_{N,k-1}^{\star\mathrm{[PR]}(j_1)}(\bm{\theta}_0)\bm{\rho}_{N,k-1}^{\star\mathrm{[PR]}(j_2)})^2(\bm{\theta}_0)\mid \mathcal{F}_{t_{k-1}}] &\xrightarrow{\mathbb{P}_{\bm{\theta}_0}} 0.
\end{align*}
Once again, the second limit follows trivially. To prove the first limit, start by noticing that
\begin{align*}
    \mathbb{E}_{\bm{\theta}_0}[\bm{\xi}_{k-1}^\top \mathbf{B}_j(\bm{\theta}_0)\bm{\xi}_{k-1} \mid \mathcal{F}_{t_{k-1}}] = \mathbb{E}_{\bm{\theta}_0}[\bm{\xi}_{k-1}'^\top \mathbf{B}_j(\bm{\theta}_0)\bm{\xi}_{k-1}' \mid \mathcal{F}_{t_{k-1}}] = \frac{1}{3}\tr(\mathbf{B}_j(\bm{\theta}_0)).
\end{align*} 
Then, we multiply the expectation with $N-2$ for notational simplicity and compute
\begin{align}
    &(N-2)\mathbb{E}_{\bm{\theta}_0}[\bm{\rho}_{N,k-1}^{\star\mathrm{[PR]}(j_1)}(\bm{\theta}_0)\bm{\rho}_{N,k-1}^{\star\mathrm{[PR]}(j_2)}(\bm{\theta}_0)\mid \mathcal{F}_{t_{k-1}}]\notag\\
    &=  \mathbb{E}_{\bm{\theta}_0}[\bm{\xi}_{k-1}^\top \mathbf{B}_{j_1}(\bm{\theta}_0)\bm{\xi}_{k-1}\bm{\xi}_{k-1}^\top \mathbf{B}_{j_2}(\bm{\theta}_0)\bm{\xi}_{k-1} \mid \mathcal{F}_{t_{k-1}}]+4 \mathbb{E}_{\bm{\theta}_0}[ \bm{\xi}_{k-2}^\top \mathbf{B}_{j_1}(\bm{\theta}_0)\bm{\xi}'_{k-1}\bm{\xi}_{k-2}^\top \mathbf{B}_{j_2}(\bm{\theta}_0)\bm{\xi}'_{k-1} \mid \mathcal{F}_{t_{k-1}}] \notag\\
    &+\mathbb{E}_{\bm{\theta}_0}[\bm{\xi}_{k-1}'^\top \mathbf{B}_{j_1}(\bm{\theta}_0)\bm{\xi}_{k-1}'\bm{\xi}_{k-1}'^\top \mathbf{B}_{j_2}(\bm{\theta}_0)\bm{\xi}_{k-1}' \mid \mathcal{F}_{t_{k-1}}] +\mathbb{E}_{\bm{\theta}_0}[\bm{\xi}_{k-1}^\top \mathbf{B}_{j_1}(\bm{\theta}_0)\bm{\xi}_{k-1}\bm{\xi}_{k-1}'^\top \mathbf{B}_{j_2}(\bm{\theta}_0)\bm{\xi}_{k-1}'  \mid \mathcal{F}_{t_{k-1}}] \notag\\
    &+\mathbb{E}_{\bm{\theta}_0}[\bm{\xi}_{k-1}'^\top \mathbf{B}_{j_1}(\bm{\theta}_0)\bm{\xi}_{k-1}'\bm{\xi}_{k-1}^\top \mathbf{B}_{j_2}(\bm{\theta}_0)\bm{\xi}_{k-1}  \mid \mathcal{F}_{t_{k-1}}] - \frac{4}{9}\tr(\mathbf{B}_{j_1}(\bm{\theta}_0))\tr(\mathbf{B}_{j_2}(\bm{\theta}_0)). \label{eq:rhorho}
\end{align}
Applying Corollary \ref{cor:4thmoments} on \eqref{eq:rhorho} yields
\begin{align*}
     \sum_{k=2}^{N-1}\mathbb{E}_{\bm{\theta}_0}[\bm{\rho}_{N,k-1}^{\star\mathrm{[PR]}(j_1)}(\bm{\theta}_0)\bm{\rho}_{N,k-1}^{\star\mathrm{[PR]}(j_2)}(\bm{\theta}_0)\mid \mathcal{F}_{t_{k-1}}] &= \frac{5}{9}\tr(\mathbf{B}_{j_1}(\bm{\theta}_0)\mathbf{B}_{j_2}(\bm{\theta}_0)) \\
     &+ \frac{4}{3}\frac{1}{N-2} \sum_{k=2}^{N-1} \bm{\xi}_{k-2}^\top \mathbf{B}_{j_1}(\bm{\theta}_0)\mathbf{B}_{j_2}(\bm{\theta}_0)\bm{\xi}_{k-2}.
\end{align*}
Once again, applying Proposition 3.1 from \citep{CRIMALDI2005571} yields 
\begin{equation*}
    \frac{4}{3}\frac{1}{N-2} \sum_{k=2}^{N-1} \bm{\xi}_{k-2}^\top \mathbf{B}_{j_1}(\bm{\theta}_0)\mathbf{B}_{j_2}(\bm{\theta}_0)\bm{\xi}_{k-2} \xrightarrow{\mathbb{P}_{\bm{\theta}_0}} \frac{4}{9}\tr(\mathbf{B}_{j_1}(\bm{\theta}_0)\mathbf{B}_{j_2}(\bm{\theta}_0)),
\end{equation*}
since
\begin{align*}
    \frac{1}{N-2} \sum_{k=2}^{N-1} \mathbb{E}_{\bm{\theta}_0}[\bm{\xi}_{k-2}^\top \mathbf{B}_{j_1}(\bm{\theta}_0)\mathbf{B}_{j_2}(\bm{\theta}_0)\bm{\xi}_{k-2}\mid \mathcal{F}_{t_{k-2}}] &\xrightarrow{\mathbb{P}_{\bm{\theta}_0}} \frac{1}{3}\tr(\mathbf{B}_{j_1}(\bm{\theta}_0)\mathbf{B}_{j_2}(\bm{\theta}_0)),\\
    \frac{1}{(N-2)^2} \sum_{k=2}^{N-1} \mathbb{E}_{\bm{\theta}_0}[(\bm{\xi}_{k-2}^\top \mathbf{B}_{j_1}(\bm{\theta}_0)\mathbf{B}_{j_2}(\bm{\theta}_0)\bm{\xi}_{k-2})^2\mid \mathcal{F}_{t_{k-2}}] &\xrightarrow{\mathbb{P}_{\bm{\theta}_0}} 0.
\end{align*}
This concludes the convergence in distribution of $\bm{\lambda}_N^\mathrm{[PR]}$.

To find the asymptotic distributions of $\bm{\lambda}_N^\mathrm{[\cdot S \mid R]}$, the main issue is the fact that $-\frac{1}{\sqrt{N h}}\partial_{\beta^{(i)}} \mathcal{L}_N^\mathrm{[\cdot S\mid R]} \to 0$ in probability. The proof of this follows the same ideas as in the proof of consistency. Thus, we focus only on $-\frac{1}{\sqrt{N}}\partial_{\sigma^{(jin.)}} \mathcal{L}_N^\mathrm{[\cdot S\mid R]}$. This is then used together with equations \eqref{eq:asymptotic_L_CF} and \eqref{eq:asymptotic_L_PF} to obtain the asymptotic distributions of $\bm{\lambda}_N^\mathrm{[\cdot F]}$. Thus, we start with  $-\frac{1}{\sqrt{N}}\partial_{\sigma^{(j)}} \mathcal{L}_N^\mathrm{[\cdot S\mid R]}$
\begin{align*}
        &-\frac{1}{\sqrt{(N-1)}}\partial_{\sigma^{(j)}}\mathcal{L}_N^{\mathrm{[CS\mid R]}}\left(\mathbf{Y}_{0:t_N}; \bm{\theta}\right)= -\frac{1}{\sqrt{N-1}} \sum_{k=1}^{N} \tr((\bm{\Sigma}\bm{\Sigma}^\top)^{-1}\partial_{\sigma^{(j)}}\bm{\Sigma}\bm{\Sigma}^\top) \notag\\
        &+\frac{3}{\sqrt{N-1}}\sum_{k=1}^N (\bm{\eta}_{k-1} - 2\bm{\xi}'_{k-1})^\top \bm{\Sigma}_0^\top (\bm{\Sigma}\bm{\Sigma}^\top)^{-1}(\partial_{\sigma^{(j)}}\bm{\Sigma}\bm{\Sigma}^\top) (\bm{\Sigma}\bm{\Sigma}^\top)^{-1} \bm{\Sigma}_0 (\bm{\eta}_{k-1} - 2\bm{\xi}'_{k-1}) + \sum_{k=1}^{N}R(\frac{h}{\sqrt{N}}, \mathbf{Y}_{t_{k-1}}),\\
        &-\frac{1}{\sqrt{(N-2)}} \partial_{\sigma^{(j)}}\mathcal{L}_N^{\mathrm{[PS\mid R]}}\left(\mathbf{Y}_{0:t_N}; \bm{\theta}\right) = -\frac{2}{\sqrt{N-2}} \sum_{k=1}^{N-1} \tr((\bm{\Sigma}\bm{\Sigma}^\top)^{-1}\partial_{\sigma^{(j)}}\bm{\Sigma}\bm{\Sigma}^\top)\\
        &+\frac{3}{\sqrt{N-2}}\sum_{k=1}^{N-1} \mathbf{U}_{k, k-1}^\top \bm{\Sigma}_0^\top (\bm{\Sigma}\bm{\Sigma}^\top)^{-1}(\partial_{\sigma^{(j)}}\bm{\Sigma}\bm{\Sigma}^\top) (\bm{\Sigma}\bm{\Sigma}^\top)^{-1} \bm{\Sigma}_0 \mathbf{U}_{k, k-1}\notag\\
        &- 6\sqrt{\frac{h}{N-2}}\sum_{k=1}^{N-1}  \mathbf{U}_{k, k-1}^\top \bm{\Sigma}_0^\top  (\bm{\Sigma}\bm{\Sigma}^\top)^{-1}(\partial_{\sigma^{(j)}}\bm{\Sigma}\bm{\Sigma}^\top) (\bm{\Sigma}\bm{\Sigma}^\top)^{-1} \mathbf{F}(\mathbf{Y}_{t_{k-1}}; \bm{\beta}_0) + \sum_{k=1}^{N-1}R(\frac{h}{\sqrt{N}}, \mathbf{Y}_{t_{k-1}}).
\end{align*}
\end{proof}
Once again, we define
\begin{align}
    \bm{\rho}_{N,k-1}^{\mathrm{[CS\mid R]}(j)}(\bm{\theta}_0) &\coloneqq \frac{1}{\sqrt{N-1}}\Bigg(3(\bm{\eta}_{k-1} - 2\bm{\xi}'_{k-1})^\top \bm{\Sigma}_0^\top (\bm{\Sigma}\bm{\Sigma}_0^\top)^{-1}(\partial_{\sigma^{(j)}}\bm{\Sigma}\bm{\Sigma}_0^\top) (\bm{\Sigma}\bm{\Sigma}_0^\top)^{-1} \bm{\Sigma}_0 (\bm{\eta}_{k-1} - 2\bm{\xi}'_{k-1})\notag\\
    &\hspace{12ex}-\tr((\bm{\Sigma}\bm{\Sigma}_0^\top)^{-1}\partial_{\sigma^{(j)}}\bm{\Sigma}\bm{\Sigma}_0^\top)\Bigg),\label{eq:rhoCSR}\\
    \bm{\rho}_{N,k-1}^{\mathrm{[PS\mid R]}(j)}(\bm{\theta}_0) &\coloneqq \frac{3}{\sqrt{N-2}}(\mathbf{U}_{k, k-1}^\top \bm{\Sigma}_0^\top (\bm{\Sigma}\bm{\Sigma}_0^\top)^{-1}(\partial_{\sigma^{(j)}}\bm{\Sigma}\bm{\Sigma}_0^\top) (\bm{\Sigma}\bm{\Sigma}_0^\top)^{-1} \bm{\Sigma}_0 \mathbf{U}_{k, k-1} - \frac{2}{3}\tr((\bm{\Sigma}\bm{\Sigma}_0^\top)^{-1}\partial_{\sigma^{(j)}}\bm{\Sigma}\bm{\Sigma}_0^\top))\notag\\
    &- 6\sqrt{\frac{h}{N-2}} \mathbf{U}_{k, k-1}^\top \bm{\Sigma}_0^\top  (\bm{\Sigma}\bm{\Sigma}^\top)^{-1}(\partial_{\sigma^{(j)}}\bm{\Sigma}\bm{\Sigma}^\top) (\bm{\Sigma}\bm{\Sigma}^\top)^{-1} \mathbf{F}_0(\mathbf{Y}_{t_{k-1}}).\label{eq:rhoPSR}
\end{align}
We skip the proof of the complete case, but it can be shown analogously that
\begin{align*}
    \sum_{k=1}^{N} \mathbb{E}_{\bm{\theta}_0}[\bm{\rho}_{N,k-1}^{\mathrm{[PS\mid R]}(j_1)}(\bm{\theta}_0)\bm{\rho}_{N,k-1}^{\mathrm{[PS\mid R]}(j_2)}(\bm{\theta}_0)\mid \mathcal{F}_{t_{k-1}}] &\xrightarrow{\mathbb{P}_{\bm{\theta}_0}} [\mathbf{C}_{\bm{\sigma}}(\bm{\theta}_0)]_{j_1j_2},\\
    \sum_{k=1}^{N} \mathbb{E}_{\bm{\theta}_0}[(\bm{\rho}_{N,k-1}^{\mathrm{[PS\mid R]}(j_1)}(\bm{\theta}_0)\bm{\rho}_{N,k-1}^{\mathrm{[PS\mid R]}(j_2)})^2(\bm{\theta}_0)\mid \mathcal{F}_{t_{k-1}}] &\xrightarrow{\mathbb{P}_{\bm{\theta}_0}} 0.
\end{align*}
Focusing on the partial case, we first notice:
\begin{equation*}
    \bm{\rho}_{N,k-1}^{\mathrm{[PS\mid R]}(j)}(\bm{\theta}_0) = 3 \bm{\rho}_{N,k-1}^{\mathrm{[PR]}(j)}(\bm{\theta}_0) + o_{\mathbb{P}_{\bm{\theta}_0}}(1).
\end{equation*}
Thus, the same derivations from before hold. Moreover, 
\begin{equation*}
    \bm{\rho}_{N,k-1}^{\mathrm{[PF]}(j)}(\bm{\theta}_0) = 4 \bm{\rho}_{N,k-1}^{\mathrm{[PR]}(j)}(\bm{\theta}_0) + o_{\mathbb{P}_{\bm{\theta}_0}}(1),
\end{equation*}
which concludes the proof.

\end{document}